\documentclass[11pt]{article}

\usepackage[utf8]{inputenc}
\usepackage{authblk}

\usepackage{amsfonts,amsmath,amssymb,mathtools,amsthm}
\usepackage{graphicx}
\IfFileExists{stmaryrd.sty}{
  \usepackage{stmaryrd}
  \SetSymbolFont{stmry}{bold}{U}{stmry}{m}{n}}{
  \providecommand{\llbracket}{\mathopen{[\mkern-3mu[}}
  \providecommand{\rrbracket}{\mathclose{]\mkern-3mu]}}}
\IfFileExists{thmtools.sty}{\usepackage{thmtools}}{}
\DeclareGraphicsExtensions{.eps}
\usepackage{float}

\usepackage{xcolor}
\definecolor{linkblue}{HTML}{164CA0}
\definecolor{citepink}{HTML}{B62570}
\usepackage{hyperref}
	\hypersetup{
		colorlinks = true,
		linkcolor = linkblue,
		citecolor = citepink,
		urlcolor = linkblue,
		linktoc = all,
		hypertexnames = false,
		unicode = true,
		bookmarksnumbered = false,
		pdfmenubar = true,
		pdftoolbar = true}

\usepackage
	[top=1.8cm,
	bottom=1.8cm,
	left=1.8cm,
	right=1.8cm,
	a4paper]
	{geometry}
		\usepackage[margin=2cm]{caption}
\renewcommand{\theequation}{\thesection.\arabic{equation}}

\newcommand\encadremath[1]{\vbox{\hrule\hbox{\vrule\kern8pt
\vbox{\kern8pt \hbox{$\displaystyle #1$}\kern8pt}
\kern8pt\vrule}\hrule}}
\def\enca#1{\vbox{\hrule\hbox{
\vrule\kern8pt\vbox{\kern8pt \hbox{$\displaystyle #1$}
\kern8pt} \kern8pt\vrule}\hrule}}

\newcommand\framefig[1]{
\begin{figure}[bth]
\hrule\hbox{\vrule\kern8pt
\vbox{\kern8pt \vbox{
\begin{center}
{#1}
\end{center}
}\kern8pt}
\kern8pt\vrule}\hrule
\end{figure}
}

\newcommand\figureframex[3]{
\begin{figure}[bth]
\hrule\hbox{\vrule\kern8pt
\vbox{\kern8pt \vbox{
\begin{center}
{\mbox{\epsfxsize=#1.truecm\epsfbox{#2}}}
\end{center}
\caption{#3}
}\kern8pt}
\kern8pt\vrule}\hrule
\end{figure}
}
\newcommand\figureframey[3]{
\begin{figure}[bth]
\hrule\hbox{\vrule\kern8pt
\vbox{\kern8pt \vbox{
\begin{center}
{\mbox{\epsfysize=#1.truecm\epsfbox{#2}}}
\end{center}
\caption{#3}
}\kern8pt}
\kern8pt\vrule}\hrule
\end{figure}
}
 \makeatother

\IfFileExists{mathdots.sty}{\usepackage{mathdots}}{}
 \usepackage{amsfonts,amsmath,amssymb,mathtools,amsthm}
 \usepackage{graphicx}
\IfFileExists{thmtools.sty}{\usepackage{thmtools}}{}
 \DeclareGraphicsExtensions{.eps}
\AtBeginDocument{%
}

\newcounter{appendixsubsection}[section]
\renewcommand*{\theappendixsubsection}{\thesection.\arabic{appendixsubsection}}
\makeatletter
\newcommand*{\toclevel@appendixsubsection}{2}

\newcommand*{\l@appendixsubsection}[2]{%
  \begingroup
  \let\originalnumberline\numberline
  \def\numberline##1{\originalnumberline{\appendixname~##1}}%
  \@dottedtocline{2}{1.5em}{7em}{#1}{#2}%
  \endgroup}
\makeatother

\renewcommand{\thesection}{\arabic{section}}
\renewcommand{\theequation}{\arabic{section}-\arabic{equation}}
\makeatletter
\@addtoreset{equation}{section}
\makeatother
\newtheorem{theorem}{Theorem}[section]

\newtheorem{proposition}{Proposition}[section]
\newtheorem{lemma}{Lemma}[section]
\newtheorem{corollary}{Corollary}[section]

\theoremstyle{definition}
\newtheorem{remark}{Remark}[section]
\newtheorem{definition}{Definition}[section]

\def\br{\begin{remark}\rm\small}
\def\er{\end{remark}}
\def\bt{\begin{theorem}}
\def\et{\end{theorem}}
\def\bd{\begin{definition}}
\def\ed{\end{definition}}
\def\bp{\begin{proposition}}
\def\ep{\end{proposition}}
\def\bl{\begin{lemma}}
\def\el{\end{lemma}}
\def\bc{\begin{corollary}}
\def\ec{\end{corollary}}
\def\beaq{\begin{eqnarray}}
\def\eeaq{\end{eqnarray}}

\theoremstyle{definition}

\newcommand{\be}{\begin{equation}}
\newcommand{\ee}{\end{equation}}
\newcommand{\beq}{\begin{equation}}
\newcommand{\eeq}{\end{equation}}
\newcommand{\bea}{\begin{eqnarray}}
\newcommand{\eea}{\end{eqnarray}}
\newcommand{\beqq}{\begin{equation*}}
\newcommand{\eeqq}{\end{equation*}}
\newcommand{\beaa}{\begin{eqnarray*}}
\newcommand{\eeaa}{\end{eqnarray*}}

\newcommand{\diag}{{\operatorname{diag}}}

\newcommand{\td}{\tilde}
\newcommand\blfootnote[1]{%
  \begingroup
  \renewcommand\thefootnote{}\footnote{#1}%
  \addtocounter{footnote}{-1}%
  \endgroup
}

\newcommand{\Res}{\mathop{\,\rm Res\,}}

\title{\bf{A symmetry reduction of the Painlev\'{e} IV hierarchy to the Flaschka--Newell Painlev\'{e} II hierarchy
}}

\author{$_{1}$Mohamad Alameddine\footnote{Section de math\'ematiques, Universit\'e de Gen\`eve, Rue du Conseil-G\'en\'eral 7--9, 1205 Gen\`eve, Switzerland}\,\,,
$_{2}$Olivier Marchal\footnote{\usefont{T1}{lmr}{m}{n}Université Jean Monnet Saint-Étienne, CNRS, Institut Camille Jordan UMR 5208, Institut Universitaire de France, Les Forges 2, 20 Rue du Dr Annino, 42000 Saint-Étienne, France.}
}

\date{\vspace{-5ex}}

\begin{document}

\maketitle

\textbf{Abstract}: We study the isomonodromic deformation problem associated with rank-two meromorphic connections on the Riemann sphere having one regular singularity and one irregular singularity of even order at infinity, corresponding to the even Painlev\'{e} IV hierarchy. We show that the symmetry $\Psi(-\lambda)= \sigma_1 \Psi(\lambda)\sigma_1$ defines, on the generic nonzero-residue chart, an invariant submanifold whose induced isomonodromic dynamics coincides with the Flaschka--Newell Painlev\'{e} II hierarchy. For every $d\geq 1$, we identify the spectral matrix and all $d$ normalized auxiliary matrices. In the surviving geometric Lax coordinates, the reduced Hamiltonians are the restrictions of the odd-time Hamiltonians of the even Painlev\'{e} IV hierarchy. We give the explicit comparison with the Hamiltonians of Mazzocco and Mo for $d=1,2$. This provides a geometric interpretation of the Flaschka--Newell hierarchy as a symmetry reduction of an isomonodromic deformation problem, complementing its classical description as a similarity reduction of the modified Korteweg–de Vries hierarchy. 

\medskip
 

\blfootnote{\textit{Email Addresses:}
$_{1}$\textsf{Mohamad.Alameddine@unige.ch}
$_{2}$\textsf{olivier.marchal@univ-st-etienne.fr}}

\tableofcontents

\section{Introduction}\label{Section1}

The theory of isomonodromic deformations has a long history dating back to
Riemann. It brings together the geometry of moduli spaces of connections,
integrable Hamiltonian systems, and the complex analysis of nonlinear
differential equations. Its best-known examples are the Painlev\'e equations,
which arose from the classification of second-order differential equations
with the Painlev\'e property \cite{Gambier,Painleve,Picard} and admit both
isomonodromic and Hamiltonian descriptions \cite{Fuchs,Garnier,Malmquist1922}.
For rank-two Fuchsian systems, the Schlesinger equations
\cite{schlesinger1912klasse} give rise to the Garnier systems, with Painlev\'e VI
as the one-variable case \cite{Okamoto1986Iso,Okamoto1986}. For systems with
irregular singularities, the theory developed in
\cite{JimboMiwa,JimboMiwaUeno,JMMS} allows the coefficients of the irregular
type to vary while preserving the generalized monodromy data, including the
Stokes data. The geometry of the resulting deformation spaces was developed
further in \cite{Boalch2001,Boalch2014}.

\medskip

Understanding the relations between different classes of meromorphic
connections is a fundamental problem in the geometry of isomonodromic
systems. Examples include different Lax representations of the same
Painlev\'e equation \cite{MazzoccoP6}, the symmetries studied by Okamoto
\cite{Okamoto1986,zbMATH00929752}, and Harnad's duality between systems
of different ranks and pole structures \cite{Harnad:1993hw}, subsequently
extended by Boalch \cite{Boalch2012}. The present article concerns a
relation between two Painlev\'e hierarchies. For every $d\geq 1$, the
Flaschka--Newell Lax representation of the $d$-th member of the
Painlev\'e II hierarchy has a regular singularity at zero and an unramified
irregular singularity at infinity of order $r_\infty=2d+2$. These are
precisely the pole data of the corresponding member of the even
Painlev\'e IV hierarchy. What singles out the Flaschka--Newell deformations
within this larger family? We show that the odd-time flows of the even
Painlev\'e IV hierarchy preserve the locus selected by the Flaschka--Newell
symmetry. On the normalized chart of \autoref{PropIdentification}, we
identify both the spectral matrix and the normalized auxiliary matrices
of these restricted flows with those of the Flaschka--Newell hierarchy.
The identification holds for every $d$ and includes all $d$ deformation
times.

\medskip

A second approach to Painlev\'e hierarchies comes from reductions of
isospectral integrable systems. Classical examples of isospectral hierarchies
include the Kadomtsev--Petviashvili (KP) hierarchy and its
Korteweg--de Vries (KdV) reduction \cite{1791585}. Their Hamiltonian
description uses spectral invariants and Poisson brackets on loop algebras
\cite{Adams1993,HarnadHurtubise1997}. Finite-gap integration of the KdV hierarchy was developed in \cite{DubrovinNovikov1974,Dubrovin1976,Krichever1977}. For isomonodromic deformations,
relating the spectral invariants to the Hamiltonians also requires
accounting for the explicit dependence on the deformation times
\cite{Alameddine2026,BertolaHarnadHurtubise2022,MarchalAlameddineIsospectralIsomono2023}.
For the Painlev\'e I hierarchy, the isomonodromic and KdV Hamiltonian
descriptions have been compared explicitly for the
$(2,2g+1)$ minimal models in \cite{Alameddine20262}. For the
Flaschka--Newell Painlev\'e II hierarchy, the starting point is the
similarity reduction of the modified KdV hierarchy \cite{Flaschka1980}.
Its $d$-th member is a nonlinear ordinary differential equation of order
$2d$, with the Painlev\'e II equation itself occurring at $d=1$.
For further accounts of these hierarchies, see
\cite{Joshi1995,GordoaPickering1999}.

\medskip

We use the abbreviation FN PII to distinguish the Flaschka--Newell
hierarchy from the hierarchy associated with the Jimbo--Miwa Lax
representation of Painlev\'e II \cite{JimboMiwa}, which we denote by
JMU PII. The latter is obtained from meromorphic connections with a single
irregular pole at infinity
\cite{marchal2024hamiltonianrepresentationisomonodromicdeformations}.
Both constructions give the Painlev\'e II equation for $d=1$. At higher
orders we distinguish the hierarchies by the pole data of their Lax
representations. The regular pole at zero in the FN representation
is essential to the identification with the even PIV hierarchy studied
here.

\medskip

The symmetry of the FN Lax pair was already observed by Flaschka and
Newell \cite{Flaschka1980}. In the gauge used here, it takes the form
\[
  \Psi(-\lambda)=\sigma_1\Psi(\lambda)\sigma_1,
  \qquad
  \sigma_1=\begin{pmatrix}0&1\\1&0\end{pmatrix}.
\]
Consequently, the spectral matrix satisfies
$\check{L}(-\lambda)=-\sigma_1\check{L}(\lambda)\sigma_1$.
Mazzocco and Mo \cite{mazzocco2007hamiltonian} incorporate this involution
in their twisted loop algebra. Adapting Sklyanin's separation of variables
\cite{Sklyanin1995}, they construct $d$ Darboux pairs on the corresponding
$2d$-dimensional symplectic leaf and obtain polynomial Hamiltonians for
the evolutions in $z,t_1,\ldots,t_{d-1}$. Our identification gives these
times directly in terms of the surviving deformation times of the even
PIV hierarchy,
\[
  t_{\infty,1}=z,
  \qquad
  t_{\infty,2k+1}=-4^k t_k,
  \qquad 1\leq k\leq d-1.
\]

\medskip

To describe the reduction in Hamiltonian coordinates, we start from the
geometric Darboux coordinates $(\mathbf{Q},\mathbf{P})$ of the even PIV
hierarchy \cite{marchal2024hamiltonianrepresentationisomonodromicdeformations}.
We introduce variables $\mathbf{R}$ read off from the Lax matrix. They
depend on both the positions $\mathbf{Q}$ and their conjugate momenta
$\mathbf{P}$. Using the triangular Toeplitz structure of this change of
variables, we construct conjugate coordinates $\mathbf{S}(\mathbf{Q})$
and prove that $(\mathbf{S},\mathbf{R})$ are Darboux coordinates for
every $d$. In these coordinates the symmetry forces half of the momenta
to vanish, namely $R_{0,1}=0$ and $R_{\infty,2k-1}=0$ for
$1\leq k\leq d-1$. The remaining conjugate pairs give the restricted
symplectic form wherever they are independent. We then identify the
coordinates of Mazzocco and Mo in terms of the coefficients of the
restricted Lax matrix. For $d=1$ and $d=2$, we write out the coordinate
correspondence and the reduced Hamiltonians explicitly. In particular,
the $d=2$ example exhibits both time evolutions of the fourth-order FN
equation.

\medskip

In \autoref{sec2}, we recall the construction of the FN PII hierarchy
from the mKdV hierarchy, following \cite{mazzocco2007hamiltonian}.
In \autoref{sec3}, we describe the even PIV hierarchy, construct the
geometric Lax coordinates in \autoref{TheodS}, and prove that the change
of coordinates is symplectic in \autoref{TheoSymplectic}.
In \autoref{SecSymmetryReduction}, we impose the symmetry on the
horizontal sections and obtain the reduced Hamiltonians by restricting the
odd-time PIV Hamiltonians in the surviving Darboux pairs
(\autoref{TheoHamiltonianSymmetryReduction}). We then identify the two
hierarchies in \autoref{PropIdentification}. The examples $d=1,2$ are
worked out in \autoref{sec5}. Finally, \autoref{sec6} discusses questions
arising from this identification.

\subsection*{Statements and declarations}

The authors declare that they have no conflict of interest. Data sharing
is not applicable to this article as no datasets were generated or
analyzed during the study.

\subsection*{AI Declaration}
AI was only used during the revision of this article, the authors used OpenAI's ChatGPT Astra 6.2 model to help improve the exposition, notation, and bibliography as well as perform symbolic checks of the formulas and examples. The authors take full responsibility for the final manuscript,
including its proofs and computations.

\subsection*{Acknowledgments}

The authors would like to thank P. Boalch and N. Hayford for fruitful
discussions on different aspects of this project. M. A. was partially
funded by the Swiss National Science Foundation (SNSF) under the grant
agreement PZ00P2 223297 and by the DFG -- project ID 551478549.
O. M. was supported by the fundamental junior IUF grant G752IUFMAR.

\section{The Lenard operator approach to the FN PII hierarchy} \label{sec2}

This section reviews the Hamiltonian construction of the FN PII hierarchy due to Mazzocco and Mo \cite{mazzocco2007hamiltonian}. The section starts by recalling the mKdV similarity reduction and linear reformulation of the hierarchy. Our goal is to highlight the main ingredients preparing for the identification presented in \autoref{SecSymmetryReduction}. 

\subsection{The PII hierarchy as a reduction of the mKdV hierarchy}
The PII hierarchy is obtained from the mKdV hierarchy via the \textit{self-similarity reduction} following the symmetries of the mKdV hierarchy \cite{Ablowitz1977ExactLO,Airault1979RationalSO,Flaschka1980}.

\begin{definition}[The mKdV hierarchy] Define the following Lenard recursion for $n\geq 0$:
\begin{align}
      \frac{\partial}{\partial x} R_{n+1} =& \left( \frac{\partial^3}{\partial x^3} + 4(v_x - v^2) \frac{\partial}{\partial x} + 2(v_x-v^2)_x \right) R_n, \qquad  \text{with} \qquad  
    R_0[v_x-v^2] = \frac{1}{2}.
\end{align}
We fix $R_j[0]=0$ for $j\geq1$. The mKdV hierarchy is defined as the set of PDEs given by
\begin{align}
    \frac{\partial}{\partial T_{n+1}} v + \frac{\partial}{\partial x} \left( \frac{\partial}{\partial x} +2v \right)R_n (v_x-v^2) = 0, \,\,\,\,\, \forall\, n \geq 0,
\end{align}
where $(T_n)_{n\geq 1}$ and $x$ are the parameters.
\end{definition}
The mKdV flows commute. The first flow gives $\partial_{T_1}v=-\partial_xv$, so we absorb $T_1$ into $x$ by translation. A solution stationary under the scaling symmetry satisfies
\begin{equation}
v+xv_x+\sum_{\ell=1}^n(2\ell+1)T_{\ell+1}\,\partial_{T_{\ell+1}}v=0.
\end{equation}
Substituting the mKdV flows gives the $x$-derivative of $xv-\underset{\ell=1}{\overset{n}{\sum}}(2\ell+1)T_{\ell+1}(\partial_x+2v)R_\ell$. Taking the integration constant to be $-\alpha_n$ yields
\begin{equation}
\sum_{\ell=1}^n(2\ell+1)T_{\ell+1}(\partial_x+2v)R_\ell(v_x-v^2)=xv+\alpha_n.
\end{equation}
For $T_{n+1}\ne0$, fix a local branch of its rescaling factor and introduce the similarity variables
\begin{align}
    v(x,T_2,\ldots,T_{n+1}) =& \frac{u(z,t_1,\ldots,t_{n-1})}{[(2n+1) T_{n+1}]^{\frac{1}{2n+1}}} , \qquad z = \frac{x}{[(2n+1) T_{n+1}]^{\frac{1}{2n+1}}}, \cr
    R_\ell (v_x -v^2) =& \frac{1}{[(2n+1) T_{n+1}]^{\frac{2\ell}{2n+1}}} \mathcal{L}_\ell (u_z-u^2), \cr
    t_0 = & -z, \qquad t_\ell = \frac{(2\ell+1) T_{\ell+1}}{[(2n+1) T_{n+1}]^{\frac{2\ell+1}{2n+1}}},\,\, \forall\, \ell \in \llbracket 1, n\rrbracket \,\,,\,\,\, t_n = 1.
 \end{align}

This gives the following definition of the FN PII hierarchy.
\begin{definition}[The FN PII hierarchy] The FN PII hierarchy is defined as the set of nonlinear ODEs given by 
\begin{align}
    \left( \frac{d}{d z} + 2  u \right) \mathcal{L}_n (u_z-u^2) + \sum_{\ell=1}^{n-1}t_\ell  \left( \frac{d}{d z} + 2  u \right) \mathcal{L}_\ell (u_z-u^2) = z u + \alpha_n, 
 \end{align}
 with the operator $\mathcal{L}_n$ normalized by $\mathcal{L}_j[0]=0$ for $j\geq1$ and satisfying the Lenard recursion
 \begin{align}
      \mathcal{L}_0[u_z-u^2] =& \frac{1}{2}, \qquad  \text{and} \qquad
    \frac{d}{d z} \mathcal{L}_{n+1} = \left( \frac{d^3}{d z^3} + 4(u_z - u^2) \frac{d}{d z} + 2(u_z-u^2)_z \right) \mathcal{L}_n.
 \end{align}    
\end{definition}
For example, the first two cases of the FN PII hierarchy are 
\begin{itemize}
    \item For $n=1$, one has 
    \begin{align}\label{PII}
{u}_{zz} - 2 u ^3 =  z u + \alpha_1,
    \end{align}
    which is the PII equation, the first element of the hierarchy.
    \item For $n=2$, one has
    \begin{align}\label{PII-2}
t_1 ( {u}_{zz} - 2 u ^3) + (u_{zzzz} - 10 u u_z^2 -10 u^2 u_{zz} + 6u^5) = zu + \alpha_2,
    \end{align}
    which is the second element of the hierarchy.
\end{itemize}

\subsection{The linear differential equation formulation}
In \cite{Flaschka1980}, the FN PII equation \eqref{PII} was recovered from the compatibility condition for a pair of linear differential equations involving the FN PII Lax matrix.  Following a proposal of \cite{https://doi.org/10.1002/sapm1974534249}, the construction has been generalized to the full hierarchy in \cite{articleCJM,Kudryashov} for specific values of the times and then by \cite{mazzocco2007hamiltonian} for generic values. This last work provides the isomonodromic formulation of the FN PII hierarchy summarized below.

\begin{definition} \label{DefIsomondoromicProb}
    Let $n\geq 1$. The isomonodromic deformation problem of the $n$-th member of the FN PII hierarchy is given by 
    \begin{align}
        \frac{\partial}{\partial \lambda} \Psi(\lambda) &= \mathcal{L}^{(n)}(\lambda) \Psi(\lambda), \qquad \frac{\partial}{\partial z} \Psi(\lambda) = \begin{pmatrix}
            - \lambda &u \\ u& \lambda    \end{pmatrix} \Psi(\lambda)  \nonumber \\
         \text{and } \,& (2 k +1) \frac{\partial}{\partial t_k} \Psi(\lambda) = \left( M^{(k)}(\lambda) - \begin{pmatrix}
             0  & (\partial_z  + 2 u ) \mathcal{L}_k \\
             (\partial_z  + 2 u ) \mathcal{L}_k  & 0
         \end{pmatrix} \right) \Psi(\lambda),\,\, \forall\, k\in \llbracket 1,n-1\rrbracket,
    \end{align}
    with the Lax matrix 
    \beq \mathcal{L}^{(n)}(\lambda):=\frac{1}{\lambda}\left[\begin{pmatrix}-\lambda z& -\alpha_n\\ -\alpha_n& \lambda z \end{pmatrix}+M^{(n)}(\lambda)+\sum_{\ell=1}^{n-1}t_\ell M^{(\ell)}(\lambda)\right],\eeq
and with $M^{(\ell)}(\lambda)$ a polynomial in $\lambda$ given by
\beq \forall\, \ell\in \llbracket 1,n\rrbracket\,:\,  M^{(\ell)}(\lambda):=\begin{pmatrix}\underset{j=1}{\overset{2\ell+1}{\sum}} A_j^{(\ell)} \lambda^j & \underset{j=1}{\overset{2\ell}{\sum}} B_j^{(\ell)} \lambda^j  \\ \underset{j=1}{\overset{2\ell}{\sum}} C_j^{(\ell)}\lambda^j & -\underset{j=1}{\overset{2\ell+1}{\sum}} A_j^{(\ell)} \lambda^j \end{pmatrix},
\eeq 
with the coefficients 
\begin{align}
    A^{(\ell)}_{2\ell+1} :=& 4^\ell, \qquad A^{(\ell)}_{2k} = 0,\qquad  \forall\,k \in \llbracket 0,\ell\rrbracket,  \nonumber \\
    A^{(\ell)}_{2k+1} :=& \frac{4^{k+1}}{2} \left[ \mathcal{L}_{\ell-k} (u_z-u^2) - \frac{d}{d z } \left( \frac{d}{d z }  + 2 u \right) \mathcal{L}_{\ell-k-1} (u_z-u^2)   \right], \qquad \forall\,k \in \llbracket 0,\ell-1\rrbracket,\nonumber \\
     B^{(\ell)}_{2k+1} :=& \frac{4^{k+1}}{2} \frac{d}{d z } \left( \frac{d}{d z }  + 2 u \right) \mathcal{L}_{\ell-k-1} (u_z-u^2), \qquad \forall\,k \in \llbracket 0,\ell-1\rrbracket, \nonumber \\ 
      B^{(\ell)}_{2k} :=& -4^{k} \left( \frac{d}{d z }  + 2 u \right) \mathcal{L}_{\ell-k} (u_z-u^2), \qquad \forall\,k\in \llbracket 1, \ell\rrbracket, \\
      C^{(\ell)}_{2k+1} :=& - B^{(\ell)}_{2k+1}, \,\,\,\forall\, k\in \llbracket0 ,\ell-1\rrbracket \qquad \qquad  C^{(\ell)}_{2k} := B^{(\ell)}_{2k}, \,\, \forall\, k \in \llbracket 1,\ell\rrbracket.
\end{align}
\end{definition}

In the previous definition, the deformation parameters are $(t_0=-z,t_1,\dots,t_{n-1})$ while the spectral parameter is $\lambda$ and $t_n=1$. Coefficients outside the ranges above are zero. The superscript in the spectral matrix $\mathcal{L}^{(n)}(\lambda)$ distinguishes it from the Lenard polynomials $\mathcal{L}_k$. The compatibility conditions for these equations give the FN PII hierarchy, as shown in \cite{mazzocco2007hamiltonian}. We will not reproduce the computations here and will state the main observation behind the above definition: the linear problem of \autoref{DefIsomondoromicProb} is the isomonodromic problem of the FN PII hierarchy. Grouping the coefficients by parity gives
\beq \mathcal{L}^{(n)}(\lambda):=\begin{pmatrix}\underset{k=0}{\overset{n}{\sum}} a_{2k+1}^{(n)}\lambda^{2k} & \underset{k=0}{\overset{n}{\sum}}  b_{2k}^{(n)}\lambda^{2k-1}+ \underset{k=0}{\overset{n-1}{\sum}}  b_{2k+1}^{(n)}\lambda^{2k}\\\underset{k=0}{\overset{n}{\sum}}  b_{2k}^{(n)}\lambda^{2k-1}- \underset{k=0}{\overset{n-1}{\sum}}  b_{2k+1}^{(n)}\lambda^{2k}& -\underset{k=0}{\overset{n}{\sum}}  a_{2k+1}^{(n)}\lambda^{2k}
    \end{pmatrix},
\eeq
and the correspondence between $\left(A_j^{(\ell)},B_j^{(\ell)},C_j^{(\ell)}\right)_{j,\ell}$ and the coefficients $a_j^{(n)},b_j^{(n)}$ is given by
\begin{align}
a_1^{(n)}=&\sum_{\ell=1}^n t_\ell A_{1}^{(\ell)} -z, \qquad  b_0^{(n)}=-\alpha_n, \qquad a_{2k+1}^{(n)}=\sum_{\ell=1}^n t_\ell A_{2k+1}^{(\ell)} \,\,,\,\, \forall \, k\in \llbracket 1,n\rrbracket, \cr
b_{2k+1}^{(n)}=&\sum_{\ell=1}^n t_\ell B_{2k+1}^{(\ell)} \,\,,\,\, \forall \, k\in \llbracket 0,n-1\rrbracket, \quad \quad b_{2k}^{(n)}=\sum_{\ell=1}^n t_\ell B_{2k}^{(\ell)} \,\,,\,\, \forall \, k\in \llbracket 1,n\rrbracket.
\end{align} 
The isomonodromic reformulation of the hierarchy allows one to attach a linear differential system whose isomonodromic deformations provide the hierarchy. However, in order to have explicit expressions, one needs to introduce some Darboux coordinates. In \cite{mazzocco2007hamiltonian}, the construction of the Darboux coordinates is based on considering the problem as a vector field on the coadjoint orbits of a twisted loop algebra. Using the Poisson structure on the dual loop algebra one obtains a degenerate Lie--Poisson bracket whose symplectic leaves are the coadjoint orbits of the elements. This idea of the Poisson structure on dual loop algebras is known and well documented, see for instance \cite{Babelon_Bernard_Talon_2003,HarnadRouthier1994}. This is the main construction of the isospectral approach to integrable hierarchies. The times $t_1,\ldots,t_{n-1}$ are Casimirs of this bracket (Lemma~4.5 of \cite{mazzocco2007hamiltonian}), and $t_n=1$ is fixed by normalization. The generic coadjoint orbit has dimension $2n$.

\subsection{Canonical coordinates and the Hamiltonians}

The idea behind the canonical coordinates of \cite{mazzocco2007hamiltonian} is the use of the algebro-geometric construction on the cotangent bundle of the Jacobian of the spectral curve \cite{Adams1993,Flaschka1976}. This construction uses the characteristic polynomial of the Lax matrix.
\begin{definition}[Spectral curve] The spectral curve attached to the matrix $\mathcal{L}^{(n)}(\lambda)$ is the affine plane curve defined by the characteristic equation 
\begin{align}
    \Gamma(\mu, \lambda) := \left\{ \det\left( \mu I_2 - \mathcal{L}^{(n)}(\lambda) \right)=0 \right\} = \left\{ \mu^2 = - \det\left( \mathcal{L}^{(n)}(\lambda) \right) \right\}.
\end{align}
This characteristic equation defines the eigenvalues as functions on the corresponding $2$-sheeted Riemann surface of genus $g$.     
\end{definition}
The construction uses the normalized Baker--Akhiezer eigenvector for the eigenvalue $\mu(\lambda)$, following Sklyanin's separation of variables \cite{Sklyanin1995}. The coordinates $(\td{q}_j)_{1\leq j\leq 2n}$ are the projections of the moving poles of the normalized eigenvector $\psi$, where
\begin{align}
    (c_1, c_2 ) \psi(\lambda,\mu) = 1,
\end{align}
and $(c_1,c_2)$ is proportional to $(1,-1)$. At each such pole the eigenline is spanned by $(1,1)^t$. Thus, writing $A=\mathcal{L}^{(n)}$,
\begin{align}
    A_{1,2} (\td{q}_j) + A_{1,1} (\td{q}_j) - A_{2,2} (\td{q}_j) - A_{2,1} (\td{q}_j) = 0.
\end{align}

The second half of the coordinates $(\td{p}_j)_{1\leq j\leq 2n}$ corresponds to the eigenvalues. The choice of normalization made in \cite{mazzocco2007hamiltonian} provides the coordinates $(\td{q}_j)_{1\leq j\leq 2n}$ as the roots of the equation 
\begin{align}
    \sum_{k=0}^{n-1} \left(  b_{2k+1}^{(n)}+  a_{2k+1}^{(n)} \right) \lambda^{2k} + a_{2n+1}^{(n)} \lambda^{2n} = 0,
\end{align}
and their dual coordinates are given by 
\begin{align}
   \forall\, j\in \llbracket 1, 2n\rrbracket\,:\, \td{p}_j =  \sum_{k=0}^{n}   b_{2k}^{(n)} \td{q}_j^{2k-1}.
\end{align} 
The $2n$ pairs are constrained by the involution. The generic coadjoint leaf has dimension $2n$, whereas the cotangent bundle of the Jacobian of the full spectral curve has dimension $4n$. Indeed, the characteristic equation reads
\begin{align}
    \mu^2 = 4^{2n} \lambda^{4n} + W_{2n-1} (\lambda^2) + \frac{\alpha_n^2}{\lambda^2},
\end{align}
with $\deg W\leq 2n-1$. On the generic smooth model, $y=\lambda\mu$ gives a hyperelliptic equation of degree $4n+2$ in $\lambda$, hence $g=2n$. The quotient by $(\lambda,\mu)\mapsto(-\lambda,\mu)$ has genus $n$ and equation, with $x=\lambda^2$,
\begin{align}
    \mathcal{C} := \left\{ \mu^2 = 4^{2n} x^{2n} + W_{2n-1}(x) + \frac{\alpha_n^2}{x} \right\}.
\end{align}
The separated points satisfy $\td q_{n+j}=-\td q_j$ and $\td p_{n+j}=-\td p_j$. Their pairs therefore descend to the other quotient, with $x=\lambda^2$, $y=\lambda\mu$ and $y^2=4^{2n}x^{2n+1}+xW_{2n-1}(x)+\alpha_n^2$. For the use of symmetric divisor coordinates, see \cite{Mumford2007}. Mazzocco and Mo pass to symmetric coordinates to eliminate the choice of ordering of the roots.
\begin{definition}[Symmetric algebro-geometric coordinates \cite{mazzocco2007hamiltonian}]\label{DefDarbouxMarta} Define the following set of coordinates
\begin{align}
    \td{P}_k = \td{\Pi}_{2k} = \frac{a_{2(n-k)+1}^{(n)}+b_{2(n-k)+1}^{(n)}}{a_{2n+1}^{(n)}}, \qquad \td{Q}_k = \sum_{j=1}^n \frac{1}{2 j } b_{2j}^{(n)} \frac{\partial \td{S}_{2j}}{\partial\td{\Pi}_{2k} },\qquad k \in \llbracket 1,n \rrbracket,
\end{align} 
where $\td{S}_k = \underset{j=1}{\overset{2n}{\sum}} (\td{q}_j)^k$ for $k \in \llbracket 1,2n\rrbracket$ and $(\td{\Pi}_j)_{1\leq j\leq 2n}$ are the symmetric functions given by 
\begin{align}
    \td{\Pi}_1 = \td{q}_1+\td{q}_2+\cdots+\td{q}_{2n}, \qquad \td{\Pi}_2 = \sum_{1 \leq j <k \leq 2n} \td{q}_j \td{q}_k ,\qquad \td{\Pi}_{2n} = \td{q}_1 \td{q}_2\cdots\td{q}_{2n}. 
\end{align}
\end{definition}
 
Theorem 5.1 of \cite{mazzocco2007hamiltonian} proves that the $n$ independent pairs $(\td Q_k,\td P_k)_{1\leq k\leq n}$ are canonical on the symplectic leaf. Furthermore, with these coordinates, Mazzocco and Mo in \cite{mazzocco2007hamiltonian} showed that the Lax matrices are written in a particular polynomial form 
\begin{proposition}[Theorem 6.1 of \cite{mazzocco2007hamiltonian}]\label{LaxMatrixMarta} Let us define 
\begin{align}
    \td{\mathcal{A}}(\lambda) :=\sum_{i=0}^n  a_{2i+1}^{(n)} \lambda^{2i +1} , \qquad \td{\mathcal{B}}_{\mathrm{odd}}(\lambda)  := \sum_{i=0}^{n-1}  b_{2i+1}^{(n)} \lambda^{2i +1}  , \qquad \td{\mathcal{B}}_{\mathrm{even}}(\lambda)  := \sum_{i=1}^n  b_{2i}^{(n)} \lambda^{2i}.
\end{align}
Furthermore, define the following truncated series 
\begin{align}
    \td{\mathcal{Q}}(\lambda)  :=\sum_{i=1}^n \td{Q}_i \lambda^{2i} , \qquad \td{\mathcal{P}}(\lambda)  := \sum_{i=1}^n \td{P}_i \lambda^{-2i} ,  \qquad 
    \td{\mathcal{T}}(\lambda)  := \sum_{i=1}^{n-1} t_i (2 \lambda)^{2i-2n}- z (2 \lambda)^{-2n}.
 \end{align} 
Then, the following relations hold
\begin{align}\label{CoeffABB}
    \td{\mathcal{A}}(\lambda) = & \left( \frac{1}{4} (2\lambda)^{2n+1} \left( 1 + \td{\mathcal{P}} + \frac{\left( 1 + \td{\mathcal{T}} \right)^2}{1+\td{\mathcal{P}}} \right) - (2 \lambda)^{-2n-1} \td{\mathcal{Q}}^2(1+ \td{\mathcal{P}}) \right)_+, \cr
    \td{\mathcal{B}}_{\mathrm{odd}}(\lambda)  = &  \left( \frac{1}{4} (2\lambda)^{2n+1} \left( 1 + \td{\mathcal{P}} - \frac{\left( 1 + \td{\mathcal{T}} \right)^2}{1+\td{\mathcal{P}}} \right) + (2 \lambda)^{-2n-1} \td{\mathcal{Q}}^2(1+ \td{\mathcal{P}}) \right)_+, \cr
    \td{\mathcal{B}}_{\mathrm{even}}(\lambda)  = & -\lambda^2 \left( \lambda^{-2} \td{\mathcal{Q}} (1+ \td{\mathcal{P}}) \right)_+,
\end{align}
where $(\cdot)_+=[\cdot]_{\infty,+}$ retains the nonnegative powers of $\lambda$ in the Laurent expansion at infinity, and $(1+\td{\mathcal{P}})^{-1}=\sum_{\ell\geq0}(-\td{\mathcal{P}})^\ell$.
\end{proposition}

Knowing the Lax matrix in terms of $(\td{\mathcal{T}},\td{\mathcal{Q}},\td{\mathcal{P}})$ allows us to express the spectral curve and thus to define the \textit{spectral invariants}.

\begin{definition} \label{LemmaSpecInv} The coefficients of the spectral curve of the Lax matrix $\mathcal{L}^{(n)}$ can be expressed as polynomials of the irregular times and symmetric coordinates: 
\small{\begin{align}
    \mu^2 =& (2 \lambda) ^{2 n -1} (1 + \td{\mathcal{P}}) \left( (2 \lambda) ^{2 n +1 } \frac{ \left( 1 + \td{\mathcal{T}} \right)^2}{\left( 1 + \td{\mathcal{P}} \right)} - 4  (2 \lambda) ^{-2 n -1} \td{\mathcal{Q}}^2 \left( 1 + \td{\mathcal{P}} \right) \right)_+ 
     + \lambda^{-2} \left( \lambda^2 \left( \lambda^{-2} \td{\mathcal{Q}} (1 + \td{\mathcal{P}})  \right)_+ + \alpha_n \right)^2.
\end{align}}
\normalsize{The} expansion of the spectral curve at its poles defines the spectral invariants:
\begin{align}
    \mu^2  = & (2 \lambda)^{2n} \left( (2 \lambda)^{2n}  \left( 1 + \td{\mathcal{T}} \right)^2\right)_+ - 4^n \td{h}_n^{(n)} \lambda^{2n-2} + \sum_{k=1}^{n-1} \td{h}^{(n)}_k \lambda^{2k-2} + \frac{\alpha_n^2 }{\lambda^2}. 
\end{align}
The leading coefficients are fixed by the times $z,t_1,\ldots,t_{n-1}$, and the coefficient of $\lambda^{-2}$ is $\alpha_n^2$. The phase-dependent spectral invariants are $(\td h_k^{(n)})_{1\leq k\leq n}$.  
\end{definition}

The normalized spectral invariant $\td{h}_n^{(n)}=-4^{-n}[\lambda^{2n-2}]\mu^2$ has several specific properties:

\begin{proposition}[Theorem $7.2$ of \cite{mazzocco2007hamiltonian}]\label{Prophnn} $\td{h}_n^{(n)}$ is given by
\begin{align}
    \td{h}_n^{(n)}(\td{\mathbf{P}},\td{\mathbf{Q}};z) &:= -\frac{1}{4^n} \bigg( \sum_{\ell=0}^{n-1}   a_{2\ell+1}^{(n)}  a_{2(n-\ell)-1}^{(n)} - \sum_{\ell=0}^{n-1}   b_{2\ell+1}^{(n)}  b_{2(n-\ell)-1}^{(n)} +  \sum_{\ell=0}^{n}   b_{2\ell}^{(n)}  b_{2(n-\ell)}^{(n)} \bigg). 
\end{align}
Here the higher times $t_1,\ldots,t_{n-1}$ and the parameter $\alpha_n$ are held fixed. The Hamiltonian for the $z$-flow is
\beq \td{\mathcal{H}}_n^{(n)}(\td{\mathbf{P}},\td{\mathbf{Q}};z) :=\td{h}_n^{(n)}(\td{\mathbf{P}},\td{\mathbf{Q}};z) +\frac{\td{Q}_n}{4^n}.\eeq
Its Hamilton equations give the $n$-th member of the FN PII hierarchy
\begin{align}
   \forall \, k\in \llbracket 1,n\rrbracket\,:\,  \frac{ d \td{P}_k}{d z} = - \frac{ \partial \td{\mathcal{H}}_n^{(n)}}{ \partial \td{Q}_k}, \qquad \frac{ d \td{Q}_k}{d z} = \frac{ \partial \td{\mathcal{H}}_n^{(n)}}{ \partial \td{P}_k}. 
\end{align}
\end{proposition}

Let us conclude this section by pointing out a practical difficulty in
writing the Hamiltonians explicitly in the symmetric Darboux coordinates
of \autoref{DefDarbouxMarta}. For the $z$-flow, the correction to the
spectral invariant $\td{h}_n^{(n)}$ is simply $\frac{1}{4^n}\td{Q}_n$, as
shown in \autoref{Prophnn}. For the higher-time flows, Theorem~8.4 of
\cite{mazzocco2007hamiltonian} gives the shifts
$\delta\td{h}_k^{(n)}$ explicitly in these coordinates for the
auxiliary $s_k$-flows. To write the Hamiltonian of an original
$t_k$-flow as a polynomial in $(\td{\mathbf{P}},\td{\mathbf{Q}})$,
one must still expand the spectral invariants in these variables and
extract the appropriate combination of the corrected invariants from
the generating series in Corollary~8.5 of
\cite{mazzocco2007hamiltonian}. These polynomial expansions become
lengthy and require substantial computations as $n$ increases. We give the resulting expressions for the
first two members of the hierarchy below.

\begin{itemize}
    \item For $n=1$, write $(\td Q,\td P)=(\td Q_1,\td P_1)$. The Hamilton equations are
    \begin{align}
        \frac{d \td{Q}}{d z} = 8 \td{P} + \frac{1}{4} \td{Q}^2 + 2z   ,\qquad  \frac{d \td{P}}{d z} = - \frac{1}{4} -\frac{1}{2} \td{P} \td{Q} + \frac{\alpha_1}{2},
    \end{align}
with the Hamiltonian
\begin{align}\label{Example1}
    \td{\mathcal{H}}_1^{(1)} = 4 \td{P}^2 + \frac{1}{4} \td{Q} + \frac{1}{4} \td{P} \td{Q}^2 + 2 z\td{P} - \frac{1}{2} \td{Q} \alpha_1.
 \end{align}
    \item For $n=2$ write $h_{2,\mathrm{raw}}^{(2)}=-16\td h_2^{(2)}$ for the coefficient of $\lambda^2$ in $\mu^2$. One obtains (see Example $8.6$ in \cite{mazzocco2007hamiltonian}):
\small{\bea \label{Example2}  \td{h}_{1}^{(2)}&=&-32z\td{P}_2+256\td{P}_1^2\td{P}_2-256\td{P}_2^2-2\td{P}_2\td{Q}_1\td{Q}_2-\td{P}_1\td{P}_2\td{Q}_2^2+2\alpha_2(\td{P}_1\td{Q}_2+\td{Q}_1)-16t_1\td{P}_2(8\td{P}_1-t_1),\cr
h_{2,\mathrm{raw}}^{(2)}&=&-32z\td{P}_1+256\td{P}_1^3-512\td{P}_1\td{P}_2-\td{P}_2\td{Q}_2^2+\td{Q}_1^2+2\alpha_2\td{Q}_2
-16t_1(8\td{P}_1^2-8\td{P}_2-t_1\td{P}_1),\cr
\delta \td{h}_{1}^{(2)}&=&-\td{P}_1\td{Q}_2-3\td{Q}_1,\cr
\delta h_{2,\mathrm{raw}}^{(2)}&=&-\td{Q}_2,\cr
\td{\mathcal{H}}_2^{(2)}&:=&\mathrm{Ham}_z(\td{P}_1,\td{P}_2,\td{Q}_1,\td{Q}_2):=\frac{1}{16}\left(-h_{2,\mathrm{raw}}^{(2)}-\delta h_{2,\mathrm{raw}}^{(2)}\right)\cr
&=&2z\td{P}_1-16\td{P}_1^3+32\td{P}_1\td{P}_2+\frac{1}{16}\td{P}_2\td{Q}_2^2-\frac{1}{16}\td{Q}_1^2-\frac{1}{8}\alpha_2\td{Q}_2+8t_1\td{P}_1^2-8t_1\td{P}_2-t_1^2\td{P}_1+\frac{\td{Q}_2}{16},\cr
\td{\mathcal{H}}^{(2)}_1 &:=&\mathrm{Ham}_{t_1}(\td{P}_1,\td{P}_2,\td{Q}_1,\td{Q}_2):=\frac{1}{3}\left(\frac{1}{4}(\td{h}_1^{(2)}+ \delta \td{h}_1^{(2)}) -\frac{1}{16}t_1(h_{2,\mathrm{raw}}^{(2)} +\delta h_{2,\mathrm{raw}}^{(2)})\right)\cr
&=&
\frac{64}{3}\td{P}_1^2 \td{P}_2-\frac{64}{3}\td{P}_2^2-\frac{1}{6}\td{Q}_1\td{Q}_2\td{P}_2-\frac{1}{12}\td{P}_1\td{P}_2\td{Q}_2^2+\frac{1}{6}\alpha_2 \td{P}_1\td{Q}_2-\frac{8}{3}z\td{P}_2 \cr&&-\frac{4}{3}t_1^2 \td{P}_2
-\frac{1}{12}\td{P}_1\td{Q}_2-\frac{1}{4}\td{Q}_1-\frac{16}{3}t_1\td{P}_1^3+\frac{8}{3}t_1^2\td{P}_1^2-\frac{1}{3}t_1^3\td{P}_1+\frac{1}{48}t_1\td{P}_2\td{Q}_2^2+\frac{2}{3}zt_1\td{P}_1\cr&&-\frac{1}{48}t_1\td{Q}_1^2
-\frac{1}{24}t_1\alpha_2\td{Q}_2+\frac{1}{6}\alpha_2 \td{Q}_1+\frac{1}{48}t_1\td{Q}_2.
\eea}
\normalsize{}In our conventions, put $H=\td{\mathcal{H}}^{(2)}_2$ and $K=\td{\mathcal{H}}^{(2)}_1$. Direct substitution gives $\partial_{t_1}H-\partial_zK+\{H,K\}=\frac{1}{6}(2z-t_1^2)$, with $\{\td Q_i,\td P_j\}=\delta_{ij}$. This is independent of the phase coordinates, so the two time evolutions commute. The Hamilton equations are
\bea 
            \forall\, i\in \llbracket 1,2\rrbracket\,:\, \frac{\partial \td{Q}_i}{\partial z}&=&\frac{\partial \td{\mathcal{H}}_2^{(2)}}{\partial \td{P}_i},\qquad \frac{\partial \td{P}_i}{\partial z}=-\frac{\partial \td{\mathcal{H}}_2^{(2)}}{\partial \td{Q}_i},\cr
            \forall\, i\in \llbracket 1,2\rrbracket\,:\,  \frac{\partial \td{Q}_i}{\partial t_1}&=&\frac{\partial \td{\mathcal{H}}_1^{(2)}}{\partial \td{P}_i},\qquad \frac{\partial \td{P}_i}{\partial t_1}=-\frac{\partial \td{\mathcal{H}}_1^{(2)}}{\partial \td{Q}_i}.
\eea
Moreover, upon the identification $\td{Q}_2 = 16u$, we have the equivalence
        \begin{equation}
    \forall\, i\in \llbracket 1,2\rrbracket\,:\, \frac{\partial \td{Q}_i}{\partial z}=\frac{\partial \td{\mathcal{H}}_2^{(2)}}{\partial \td{P}_i},\qquad \frac{\partial \td{P}_i}{\partial z}=-\frac{\partial \td{\mathcal{H}}_2^{(2)}}{\partial \td{Q}_i}\Longleftrightarrow \,\,  \eqref{PII-2}.
        \end{equation}
Note that the auxiliary mKdV flow $\frac{\partial u}{\partial t_1} = 2u^2u_z-\frac{1}{3}u_{zzz}$ (which here is just the statement that $u$ solves the mKdV equation) is equivalent to
        \begin{equation}
          \forall\, i\in \llbracket 1,2\rrbracket\,:\,  \frac{\partial \td{Q}_i}{\partial t_1}=\frac{\partial \td{\mathcal{H}}_1^{(2)}}{\partial \td{P}_i},\qquad \frac{\partial \td{P}_i}{\partial t_1}=-\frac{\partial \td{\mathcal{H}}_1^{(2)}}{\partial \td{Q}_i}.
        \end{equation}
\end{itemize}
\begin{remark}
    The Miura transformation relating the PII and PXXXIV hierarchies suggests a corresponding transformation of their Hamiltonians.
\end{remark}

\section{The even PIV hierarchy} \label{sec3}
The even PIV hierarchy is a specific case of rank $2$ meromorphic connections over the Riemann sphere with a given pole structure. More precisely, it corresponds to one regular pole located at $X_0$ and one irregular and non-twisted pole of even order $r_\infty=2d+2\geq 4$ located at $\infty$. Meromorphic connections over the Riemann sphere have been studied in detail in \cite{MarchalAlameddineIsospectralIsomono2023,marchal2024hamiltonianrepresentationisomonodromicdeformations} where expressions of the Lax matrices and Hamiltonians have been obtained for an arbitrary pole structure. This section is thus devoted to presenting the main results derived from this setting and to supplementing them with additional Darboux coordinates. Note that one cannot directly use the results of \cite{MarchalAlameddineIsospectralIsomono2023} because of a different choice of reduction of the symplectic space using M\"obius transformations.

\subsection{Isomonodromic formulation of the even PIV hierarchy}
The even PIV hierarchy governs the isomonodromic deformations of rank $2$ generic connections defined on the trivial vector bundle over the Riemann sphere with an effective divisor having one regular pole (i.e. $n=1$ and $r_0=1$ in the notation of \cite{marchal2024hamiltonianrepresentationisomonodromicdeformations}) located at $X_0$, and one irregular pole of even order $r_\infty=2d+2\geq 4$ located at $ \infty$. Note that the PIV hierarchy can be defined by allowing the pole at infinity to have odd order as well. This motivates the following definition:

\begin{definition}[The space of connections for the even PIV hierarchy]\label{DefP4hierarchy} Considering a meromorphic connection defined on the rank $2$ trivial vector bundle over the Riemann sphere, we define the following Poisson space:
\bea
F_{d} &:=& \Big\{\hat{L}(\lambda) = \sum_{k=1}^{2d+1} \hat{L}^{[\infty,k]} \lambda^{k-1} + \frac{\hat{L}^{[X_0,0]}}{\lambda-X_0}  \text{ with }\,\, (\hat{L}^{[X_0,0]},\hat{L}^{[\infty,1]},\ldots,\hat{L}^{[\infty,2d+1]}) \in \left(\mathfrak{sl}_2(\mathbb{C})\right)^{2d+2}\cr
&& \text{ and }  \hat{L}^{[\infty,2d+1]} \text{ and } \hat{L}^{[X_0,0]} \text{ each have two distinct eigenvalues}
\Big\}/\mathrm{SL}_2(\mathbb{C}),
\eea
where the $\mathrm{SL}_2(\mathbb{C})$ action is the conjugation action of the reductive group. Upon trivializing the bundle, one obtains the meromorphic differential system
\beq \label{eqhatL}\partial_\lambda \hat{\Psi}(\lambda)= \hat{L}(\lambda)\hat{\Psi}(\lambda),\eeq
defined in the usual complex-analytic charts of the Riemann sphere with $\hat{\Psi}(\lambda)$ a fundamental matrix of horizontal sections. 
\end{definition}

Since the leading coefficient at infinity has distinct eigenvalues, the Hukuhara--Levelt--Turrittin theorem gives a diagonal formal normal form there \cite{Birkhoff,Wasowbook}.

\begin{proposition}[Local diagonalization and irregular times] At infinity a formal gauge transformation gives the diagonal normal form below. A finite jet of the formal gauge suffices for the displayed equality up to $O(\lambda^{-2})$: 
\beq \hat{L}(\lambda) \mathrel{\overset{\rm formal}{\sim}} \diag\left(-\sum_{k=0}^{2d+1} t_{\infty,k}\lambda^{k-1},\sum_{k=0}^{2d+1} t_{\infty,k}\lambda^{k-1} \right)+ O(\lambda^{-2}).\eeq
At $\lambda=X_0$, the residue is conjugate to its diagonal form
\beq \hat{L}^{[X_0,0]}\sim \operatorname{diag}(-t_{X_0,0},t_{X_0,0}).\eeq
The parameters $\mathbf{t}:=(t_{\infty,k})_{1\leq k\leq 2d+1}$ appearing in the formal normal form are called the irregular times. They constitute the base manifold of parameters for isomonodromic deformations together with the location of the pole $X_0$. The formal monodromy exponents at infinity and the residue eigenvalues at $X_0$ are determined by $t_{\infty,0}$ and $t_{X_0,0}$, respectively. These parameters are held fixed during the isomonodromic deformations.
\end{proposition}

The geometric picture is then captured by the definition of the tangent space to the base $\mathbb{B}=\{X_0,t_{\infty,1},\dots,t_{\infty,2d+1}\}$ of deformation parameters. The time base has dimension $2d+2$, and the fixed-time phase fiber has symplectic dimension $4d$. Therefore, in our setting, a general isomonodromic deformation corresponds to
\beq \label{GeneralDef}\mathcal{L}_{\td{\boldsymbol{\alpha}}}:=\sum_{k=1}^{2d+1}\td{\alpha}_{\infty,k}\partial_{t_{\infty,k}} + \td{\alpha}_{X_0}\partial_{X_0},\eeq
where $\td{\boldsymbol{\alpha}} \in \mathbb{C}^{2d+2}$ is a complex vector of the tangent space. This gives rise to a space of deformations of dimension $2d+2$. To a general deformation operator \eqref{GeneralDef}, we associate an auxiliary matrix $\hat{A}_{\td{\boldsymbol{\alpha}}}(\lambda)$ solution of:
\beq \label{AuxEq}\mathcal{L}_{\td{\boldsymbol{\alpha}}}[\hat{\Psi}(\lambda)]=\hat{A}_{\td{\boldsymbol{\alpha}}}(\lambda)\hat{\Psi}(\lambda).\eeq
The auxiliary matrix is meromorphic with poles only at the singularities of the connection. The Hamiltonian structure obtained from the compatibility equations of \eqref{eqhatL} and \eqref{AuxEq} admits various symmetries, one of which is the set of M\"obius transformations on $\lambda$ \cite{Boalch2012,marchal2024hamiltonianrepresentationisomonodromicdeformations}. In practice, since we have fixed a pole at infinity, there remain two degrees of freedom from the action of M\"obius transformations that allow us to choose two time parameters. In \cite{marchal2024hamiltonianrepresentationisomonodromicdeformations}, the canonical choice taken was $(t_{\infty,2d+1},t_{\infty,2d})=(1,0)$ and simplified versions of the Hamiltonians and Lax matrices were derived under this canonical choice. However, in order to obtain a more direct identification with the formalism of \autoref{sec2}, we will use the M\"obius transformations to take a different normalization.

\begin{definition}[Normalization from M\"obius transformations]\label{NormalizationLaxMatrix} Using part of the $\mathrm{SL}_2(\mathbb{C})$ action, we normalize the Lax matrix at infinity by imposing
\beq \hat{L}^{[\infty,2d+1]}=\operatorname{diag}(-t_{\infty,2d+1},t_{\infty,2d+1}). \eeq
Moreover, we use the M\"obius transformations on $\lambda$ to take:
\beq X_0=0 \text{ and } t_{\infty,2d+1}=-4^d.\eeq
We shall refer to this normalization as the ``FN normalization''. Moreover, to stress that this normalization is enforced throughout the rest of this paper we shall denote the corresponding matrices by $\check{\Psi}(\lambda)$, $\check{L}(\lambda)$ and $\check{A}_{\boldsymbol{\alpha}}(\lambda)$. 
\end{definition}

We use a check rather than a tilde because the tilde notation was used in \cite{marchal2024hamiltonianrepresentationisomonodromicdeformations} for the normalization $(t_{\infty,r_\infty-1},t_{\infty,r_\infty-2})=(1,0)$ and $L^{[\infty,2d]}_{1,2}:=\omega=1$ which is not the one we are considering in this paper. Note that the normalization $t_{\infty,2d+1}=-4^d$ anticipates the leading term of the FN PII hierarchy and simplifies the comparison carried out in \autoref{SecSymmetryReduction}.

\medskip

\autoref{NormalizationLaxMatrix} implies that a general isomonodromic deformation simplifies into
\beq  \label{GeneralDef2}\mathcal{L}_{\boldsymbol{\alpha}}:=\sum_{k=1}^{2d}\alpha_{\infty,k}\partial_{t_{\infty,k}}, \eeq 
where $\boldsymbol{\alpha} \in \mathbb{C}^{2d}$. This gives rise to a reduced space of deformations of dimension $2d$. In order to obtain explicit expressions for the Lax matrices, one needs to trivialize the symplectic structure on the phase space by defining $4d$ Darboux coordinates $(\mathbf{u},\mathbf{v}):=(u_j,v_j)_{1\leq j\leq 2d}$. In these coordinates, the compatibility condition reads
\beq \label{compatibility} \partial_\lambda \check{A}_{\boldsymbol{\alpha}}(\lambda)- \mathcal{L}_{\boldsymbol{\alpha}}[\check{L}(\lambda)] + \left[\check{A}_{\boldsymbol{\alpha}}(\lambda),\check{L}(\lambda) \right]=0.
\eeq
It gives the evolutions of the chosen coordinates $\mathcal{L}_{\boldsymbol{\alpha}}[u_j], \mathcal{L}_{\boldsymbol{\alpha}}[v_j]$ in any direction $\boldsymbol{\alpha}$ of the tangent space. If these coordinates are chosen in a canonical way with respect to the Poisson structure, then these coordinates will satisfy the usual Hamiltonian evolution equations
\beq \mathcal{L}_{\boldsymbol{\alpha}}[u_j]= \frac{\partial \mathrm{Ham}_{\boldsymbol{\alpha}}(\mathbf{u},\mathbf{v};\mathbf{t},t_{\infty,0}, t_{0,0})}{\partial v_j}\,\,,\,\, \mathcal{L}_{\boldsymbol{\alpha}}[v_j]= -\frac{\partial \mathrm{Ham}_{\boldsymbol{\alpha}}(\mathbf{u},\mathbf{v};\mathbf{t},t_{\infty,0}, t_{0,0})}{\partial u_j},\eeq
for any isomonodromic deformation characterized by $\boldsymbol{\alpha}$. One of the main goals of \cite{MarchalAlameddineIsospectralIsomono2023,marchal2024hamiltonianrepresentationisomonodromicdeformations} was to provide suitable sets of Darboux coordinates and provide explicit formulas for both the Lax matrices and the Hamiltonians. We summarize these results in the following section.

\subsection{Oper Darboux coordinates and explicit parameterization}
The first set of canonical Darboux coordinates used to solve the compatibility equation is the set of ``oper Darboux coordinates'': $(\mathbf{q},\mathbf{p}):=(q_1,\dots,q_{2d},p_1,\dots,p_{2d})$ on the chart of simple, nonzero roots, defined by 
\beq\label{DefCheckL12}\check{L}_{1,2}(\lambda)=\frac{\omega \underset{i=1}{\overset{2d}{\prod}}(\lambda-q_i)}{\lambda}, \qquad \text{ and } \qquad  
p_i=\check{L}_{1,1}(q_i) \,\,, \qquad  \forall \, i\in \llbracket 1, 2d\rrbracket.
\eeq
Each pair $(q_i,p_i)$ defines a point on the spectral curve $\det( yI_2-\check{L}(\lambda))=0$, since $\det(p_i I_2-\check{L}(q_i))=0$ for all $i\in \llbracket 1,2d\rrbracket$. Moreover, they are canonical coordinates and formulas for the Hamiltonians and Lax matrices are available in \cite{marchal2024hamiltonianrepresentationisomonodromicdeformations}. We recall them below. The coefficient $\omega$ is a nonzero gauge parameter. It reflects the fact that we did not use all the degrees of freedom of the $\mathrm{SL}_2(\mathbb{C})$ action to select a representative of $F_d$. In particular, one has a residual centralizer action that allows us to fix one off-diagonal entry of the subleading coefficient without changing the leading order (conjugation by matrices $\operatorname{diag}(c,c^{-1})$, $c\in\mathbb C^\times$). As we will see below, the symmetry imposed on the system  determines a specific choice of $\omega$. In this setting, the entry $\check{L}_{1,1}(\lambda)$ is given by:
\bea \label{DefCheckL11}\check{L}_{1,1}(\lambda)&=&\frac{-O_{2d}(\lambda)- (t_{\infty,2d+1}\lambda+t_{\infty,2d}+t_{\infty,2d+1}\underset{i=1}{\overset{2d}{\sum}}q_i)\underset{i=1}{\overset{2d}{\prod}}(\lambda-q_i)}{\lambda} \,, \,\text{  with }\cr
O_{2d}(\lambda)&:=&-\sum_{i=1}^{2d} p_iq_i\prod_{j\neq i}\frac{\lambda-q_j}{q_i-q_j},
\eea
a Lagrange interpolating term such that $O_{2d}(q_i)=-p_iq_i$ for all $i\in \llbracket 1, 2d\rrbracket$. As suggested by their name, the oper Darboux coordinates are particularly convenient to express the Lax matrices in the oper gauge. The oper gauge is reached through the transformation:
\beq \Psi_{\mathrm{oper}}(\lambda) =\begin{pmatrix}1 &0\\
\check{L}_{1,1}(\lambda)& \check{L}_{1,2}(\lambda)\end{pmatrix} \check{\Psi}(\lambda),\eeq
for which the corresponding Lax matrix $L_{\mathrm{oper}}(\lambda)$ is companion-like:
\begin{align} \left[L_{\mathrm{oper}}(\lambda)\right]_{1,1}=&0, \qquad \left[L_{\mathrm{oper}}(\lambda)\right]_{1,2}=1, \qquad 
  \left[L_{\mathrm{oper}}(\lambda)\right]_{2,2}=\sum_{i=1}^{2d} \frac{1}{\lambda-q_i} -\frac{1}{\lambda},\cr
\left[L_{\mathrm{oper}}(\lambda)\right]_{2,1}=&-t_{\infty,2d+1}\lambda^{2d-1} -\td{P}_2(\lambda) +\sum_{j=0}^{2d-2} H_{\infty, j}\lambda^{j}+\frac{H_{0,1}}{\lambda} -\sum_{j=1}^{2d}\frac{p_j}{\lambda-q_j},
\end{align}
where $\td{P}_2(\lambda)$ denotes the following rational function of the spectral parameter (the univariate function is distinct from the FN phase coordinate $\td{P}_2$)
\bea \td{P}_2(\lambda)&:=& -\sum_{k=0}^{2d+1} \left(\sum_{j=0}^k t_{\infty,2d+1-j} t_{\infty,2d+1-(k-j)}\right)\lambda^{4d-k} -\frac{(t_{0,0})^2}{\lambda^2}\cr
&=& -\underset{j= 2 d-1}{\overset{4d}{\sum}}\left(\underset{m=0}{\overset{4d-j}{\sum}} t_{\infty,2 d +1-m}t_{\infty,j+m-2d+1}\right) \lambda^{j}-\frac{(t_{0,0})^2}{\lambda^2},
\eea
and the coefficients $\mathbf{H}_{\infty}:=\left(H_{\infty,0},\dots, H_{\infty,2d-2}\right)$ and $H_{0,1}$ are the coefficients of the quantum curve satisfied by $\check{\Psi}_{1,1}=\left[\Psi_{\mathrm{oper}}\right]_{1,1}$ and $\check{\Psi}_{1,2}=\left[\Psi_{\mathrm{oper}}\right]_{1,2}$. The main advantage of the oper gauge is that one can solve more easily the compatibility equation \eqref{compatibility} to obtain the Hamiltonian evolutions of the oper Darboux coordinates. This result, derived in \cite{marchal2024hamiltonianrepresentationisomonodromicdeformations}, is recalled in the following proposition.

\begin{proposition}[Oper Hamiltonian evolutions for the even PIV hierarchy \cite{marchal2024hamiltonianrepresentationisomonodromicdeformations}]\label{PropHamOper}
Define the matrix $M_\infty(\mathbf{t}) \in \mathcal{M}_{2d}(\mathbb{C})$ and the vector $\boldsymbol{\nu}^{(\boldsymbol{\alpha})}_{\infty}:= \left(\nu^{(\boldsymbol{\alpha})}_{\infty,0},\dots, \nu^{(\boldsymbol{\alpha})}_{\infty,2d-1}\right)^t\in \mathbb{C}^{2d}$ by 
\beq\label{MatrixMInfty} M_\infty(\mathbf{t}):=\begin{pmatrix}t_{\infty,2d+1}&0&\dots &\dots &0\\
\vdots &\ddots &\ddots  & &\vdots\\
\vdots &&\ddots&0&0\\
t_{\infty,3} &\dots&& t_{\infty,2d+1}&0\\
t_{\infty,2}& \dots & & t_{\infty,2d}& t_{\infty,2d+1}
 \end{pmatrix} \text{ and }\begin{pmatrix} 
\nu^{(\boldsymbol{\alpha})}_{\infty,0}\\ \vdots \\ \nu^{(\boldsymbol{\alpha})}_{\infty,2d-1}\end{pmatrix}:= (M_\infty(\mathbf{t}))^{-1}\begin{pmatrix} 
\frac{\alpha_{\infty,2d}}{2d}\\ \vdots \\ \frac{\alpha_{\infty,1}}{1}\end{pmatrix}.
\eeq
Then the Hamiltonians are given by
\small{\beq \mathrm{Ham}_{\boldsymbol{\alpha}}(\mathbf{q},\mathbf{p};\mathbf{t},t_{\infty,0},t_{0,0})=\sum_{k=0}^{2d-2} \nu_{\infty,k+1}^{(\boldsymbol{\alpha})}H_{\infty,k}(\mathbf{q},\mathbf{p},\mathbf{t},t_{\infty,0},t_{0,0}) +\nu^{(\boldsymbol{\alpha})}_{\infty,0}\left(H_{0,1}(\mathbf{q},\mathbf{p},\mathbf{t},t_{\infty,0},t_{0,0})- \sum_{j=1}^{2d} p_j\right),
\eeq}
\normalsize{which} is equivalent to
\beq \label{NewHamReduced}\begin{pmatrix}\mathrm{Ham}_{{t_{\infty,1}}}(\mathbf{q},\mathbf{p};\mathbf{t},t_{\infty,0},t_{0,0})\\ \vdots \\ (2d-1)\mathrm{Ham}_{{t_{\infty,2d-1}}}(\mathbf{q},\mathbf{p};\mathbf{t},t_{\infty,0},t_{0,0})\\
(2d)\mathrm{Ham}_{{t_{\infty,2d}}}(\mathbf{q},\mathbf{p};\mathbf{t},t_{\infty,0},t_{0,0})
\end{pmatrix}=\left(M_{\infty}(\mathbf{t})\right)^{-1}\begin{pmatrix}H_{\infty,2d-2}(\mathbf{q},\mathbf{p},\mathbf{t},t_{\infty,0},t_{0,0})\\ \vdots\\ H_{\infty,0}(\mathbf{q},\mathbf{p},\mathbf{t},t_{\infty,0},t_{0,0}) \\ H_{0,1}(\mathbf{q},\mathbf{p},\mathbf{t},t_{\infty,0},t_{0,0})-\underset{j=1}{\overset{2d}{\sum}}p_j
\end{pmatrix}.
\eeq
Finally, the coefficients $\mathbf{H}_{\infty}$ and $H_{0,1}$ are given in terms of the oper Darboux coordinates by
\beq\label{HCoeffOperDarbouxCoord} \begin{pmatrix}1&q_1 &\dots &\dots &q_1^{2d-2}& \frac{1}{q_1}\\
1& q_2&\dots &\dots& q_{2}^{2d-2}& \frac{1}{q_2}\\
\vdots & & & & \vdots&\vdots\\
\vdots & & & & \vdots&\vdots\\
1& q_{2d} &\dots & \dots& q_{2d}^{2d-2}&  \frac{1}{q_{2d}}\end{pmatrix}\begin{pmatrix}H_{\infty, 0}\\ \vdots\\ H_{\infty,2d-2} \\ H_{0,1} \end{pmatrix}=\begin{pmatrix} p_1^2 + \frac{p_1}{q_1}+\td{P}_2(q_1)+\overset{2d}{\underset{i\neq 1}{\sum}}\frac{p_i-p_1}{q_1-q_i}+ t_{\infty,2d+1}q_1^{2d-1}\\
\vdots\\ 
\vdots\\
 p_{2d}^2 + \frac{p_{2d}}{q_{2d}}+\td{P}_2(q_{2d})+\overset{2d}{\underset{i\neq 2d}{\sum}}\frac{p_i-p_{2d}}{q_{2d}-q_i}+ t_{\infty,2d+1}q_{2d}^{2d-1}
\end{pmatrix}.
\eeq
For completeness, the expression of the auxiliary matrix $A_{\boldsymbol{\alpha},\mathrm{oper}}(\lambda)$ is given in \autoref{AppendixAuxiliaryMatrixOperGauge}.
\end{proposition}

We record three consequences of \autoref{PropHamOper}.
\begin{itemize}
    \item From \eqref{NewHamReduced}, Hamiltonians are obtained by multiplication between the time-dependent matrix $M_\infty(\mathbf{t})^{-1}$ and the vector involving the coefficients $(H_{\infty,k})_{0\leq k\leq 2d-2}$ and $H_{0,1}$ appearing in the quantum curve (i.e. the Lax matrix in the oper gauge).  This structure is preserved by any time-independent change of Darboux coordinates that preserves the symplectic form $\Omega:=\underset{i=1}{\overset{2d}{\sum}} dq_i\wedge dp_i$ (and this will always give canonical coordinates since $(\mathbf{q},\mathbf{p})$ are themselves canonical). As we will see below, this is the case for the geometric Darboux coordinates. However, while the symplectic structure is preserved, one still needs to express the quantities $(H_{\infty,k})_{0\leq k\leq 2d-2}$ and $H_{0,1}$ using the new coordinates.
    \item The structure of the Hamiltonian described above immediately implies that the oper Darboux coordinates, or any Darboux coordinates that are obtained from a time-independent change preserving $\underset{i=1}{\overset{2d}{\sum}} dq_i\wedge dp_i$, differ from the so-called isospectral Darboux coordinates. In fact, the relation between both sets of coordinates is studied in \cite{MarchalAlameddineIsospectralIsomono2023}.
    \item The Hamiltonians are not polynomial in the oper Darboux coordinates. In the geometric Darboux coordinates introduced below, they are polynomial in the momenta and can have poles at $Q_{0,1}=0$, as the $d=1$ example in \autoref{sec5} illustrates.
\end{itemize}

\subsection{The geometric Darboux coordinates}
\sloppy
We write the geometric Darboux coordinates as
\[
  (\mathbf{Q}_\infty,Q_{0,1},\mathbf{P}_\infty,P_{0,1})
  =(Q_{\infty,0},\ldots,Q_{\infty,2d-2},Q_{0,1},
    P_{\infty,0},\ldots,P_{\infty,2d-2},P_{0,1}).
\]
They are related to the oper Darboux coordinates by the time-independent
symplectic change of coordinates
\bea\label{DefNewcoor} \frac{\omega\underset{j=1}{\overset{2d}{\prod}}(\lambda-q_j)}{\lambda}&=& \omega\left(\frac{Q_{0,1}}{\lambda}+\sum_{k=0}^{2d-2} Q_{\infty,k}\lambda^k +\lambda^{2d-1}\right),\\
p_i&=&\sum_{k=0}^{2d-2} P_{\infty,k}\frac{\partial Q_{\infty,k}(q_1,\dots,q_{2d})}{\partial q_i}+ P_{0,1}\frac{\partial Q_{0,1}(q_1,\dots,q_{2d})}{\partial q_i} \,\,,\,\, \forall \, i\in \llbracket 1,2d\rrbracket. \nonumber
\eea
The fact that this change of coordinates is symplectic is a consequence of Lemma $6.3$ of \cite{MarchalP1Hierarchy}. We point out that the geometric Darboux coordinates defined in this paper differ from those of \cite{MarchalAlameddineIsospectralIsomono2023} by a factor $\omega$ for $(\mathbf{Q}_{\infty},Q_{0,1})$ and $\omega^{-1}$ for $(\mathbf{P}_{\infty},P_{0,1})$. The present choice is motivated by the fact that we want Darboux coordinates to be independent of the choice of $\omega$ so that quantities $(\mathbf{H}_\infty, H_{0,1})$ expressed in these coordinates will be independent of $\omega$.
    
Expressions for the Lax matrices and Hamiltonians for the geometric Darboux coordinates are adapted from formulas in \cite{MarchalAlameddineIsospectralIsomono2023} taking into account that the normalization of times and poles is not exactly the same.

\begin{proposition}[Lax matrix in terms of geometric Darboux coordinates]\label{GeoLaxMatrices}The Lax matrix has the following parameterization in terms of the Darboux coordinates:
\small{\bea\label{L11OldPaper} \check{L}_{1,2}(\lambda)&=&\omega\left(\frac{Q_{0,1}}{\lambda}+\sum_{k=0}^{2 d-2} Q_{\infty,k}\lambda^k+\lambda^{2d-1}\right),\cr
\check{L}_{1,1}(\lambda)&=&-\sum_{k=0}^{2 d-2} P_{\infty,2 d-2-k}\lambda^k-\sum_{k=0}^{2 d-3}\sum_{m=0}^{2d-3-k}P_{\infty,m}Q_{\infty,k+1+m}\lambda^k\cr
&&+ \frac{P_{0,1}Q_{0,1}}{\lambda}-\left(t_{\infty,2 d+1}\lambda+t_{\infty, 2d}-t_{\infty,2d+1}Q_{\infty,2 d-2}\right)\frac{\check{L}_{1,2}(\lambda)}{\omega},\cr
\check{L}_{2,2}(\lambda)&=&-\check{L}_{1,1}(\lambda),\cr
\check{L}_{2,1}(\lambda)
&=& \frac{1}{\omega}\left[\frac{\underset{j= 2 d-1}{\overset{4d}{\sum}}\left(\underset{m=0}{\overset{4d-j}{\sum}} t_{\infty,2 d +1-m}t_{\infty,j+m-2d+1}\right) \lambda^{j} -\check{L}_{1,1}(\lambda)^2}{\frac{\check{L}_{1,2}(\lambda)}{\omega}}\right]_{\infty,+}\cr&&
+\frac{1}{\omega}\frac{\frac{(t_{0,0})^2}{Q_{0,1}}-Q_{0,1}\left(P_{0,1}-t_{\infty,2d}+t_{\infty,2d+1}Q_{\infty,2d-2}\right)^2 }{\lambda}.
\eea}
\normalsize{}
\end{proposition}

Following \autoref{PropHamOper}, one only needs to express the quantities $(\mathbf{H}_{\infty},H_{0,1})$ in terms of the new coordinates to obtain their evolutions. This is given by the following proposition.

\begin{proposition}[Hamiltonians for the geometric Darboux coordinates]\label{PropHamGeometricDarboux} We have
\footnotesize{\beq \label{NewHamReduced2}\begin{pmatrix}\mathrm{Ham}_{{t_{\infty,1}}}(\mathbf{Q}_{\infty}, Q_{0,1},\mathbf{P}_{\infty},P_{0,1};\mathbf{t},t_{\infty,0},t_{0,0})\\ \vdots \\ (2d-1)\mathrm{Ham}_{{t_{\infty,2d-1}}}(\mathbf{Q}_{\infty},Q_{0,1},\mathbf{P}_{\infty},P_{0,1};\mathbf{t},t_{\infty,0},t_{0,0})\\
(2d)\mathrm{Ham}_{{t_{\infty,2d}}}(\mathbf{Q}_{\infty},Q_{0,1},\mathbf{P}_{\infty},P_{0,1};\mathbf{t},t_{\infty,0},t_{0,0})
\end{pmatrix}=\left(M_{\infty}(\mathbf{t})\right)^{-1}\begin{pmatrix}H_{\infty,2d-2}(\mathbf{Q}_{\infty},Q_{0,1},\mathbf{P}_{\infty},P_{0,1},\mathbf{t},t_{\infty,0},t_{0,0})\\ \vdots\\ H_{\infty,0}(\mathbf{Q}_{\infty},Q_{0,1},\mathbf{P}_{\infty},P_{0,1},\mathbf{t},t_{\infty,0},t_{0,0}) \\ H_{0,1}(\mathbf{Q}_{\infty},Q_{0,1},\mathbf{P}_{\infty},P_{0,1},\mathbf{t},t_{\infty,0},t_{0,0})+\delta 
\end{pmatrix},\eeq}
\normalsize{where} $\forall\, j\in \llbracket 0, 2d-2\rrbracket$: 
\begin{align}  \label{DefHs}
H_{\infty,j}(\mathbf{Q}_{\infty},Q_{0,1},\mathbf{P}_{\infty},P_{0,1},\mathbf{t},t_{\infty,0},t_{0,0})= &-\Res_{\lambda\to \infty}\lambda^{-j-1} \left[ (\check{L}_{1,1})^2+\check{L}_{2,1}\check{L}_{1,2} +\check{L}_{1,2}\partial_\lambda\left(\frac{ \check{L}_{1,1}}{\check{L}_{1,2}}\right)\right], \nonumber \\
H_{0,1} (\mathbf{Q}_{\infty},Q_{0,1},\mathbf{P}_{\infty},P_{0,1},\mathbf{t},t_{\infty,0},t_{0,0})=&\Res_{\lambda\to 0}\bigl[(\check{L}_{1,1})^2+\check{L}_{2,1}\check{L}_{1,2} +\check{L}_{1,2} \partial_\lambda\left(\frac{ \check{L}_{1,1}}{\check{L}_{1,2}}\right)\bigr].
\end{align}
Finally, the additional term is given by
\beq \delta=-\sum_{j=1}^{2d} p_j=2dP_{\infty,2d-2}+\sum_{k=0}^{2d-3} (k+2)Q_{\infty,k+1}P_{\infty,k}+ Q_{\infty,0}P_{0,1},\eeq
and the entries of $\check{L}(\lambda)$ are given in \autoref{GeoLaxMatrices}.
In particular, $\mathbf{H}_\infty$ and $H_{0,1}$ are independent of $\omega$.
\end{proposition}

The proof of the above results follows from adaptations of results of \cite{marchal2024hamiltonianrepresentationisomonodromicdeformations} and \cite{MarchalAlameddineIsospectralIsomono2023}. For completeness, we also give the expression of the auxiliary matrix $\check{A}_{\boldsymbol{\alpha}}(\lambda)$ in \autoref{AppendixAuxiliaryMatrixGeometricDarboux}. 

\subsection{The geometric Lax coordinates}

\subsubsection{Definition of the coordinates \texorpdfstring{$(\mathbf{R},\mathbf{S})$}{(R,S)}}

In \cite{MarchalAlameddineIsospectralIsomono2023}, the coordinates $R_{\infty,0},\ldots,R_{\infty,2d-2},R_{0,1}$ were introduced to express $\check{L}_{1,1}(\lambda)$ in a simpler way. Note that we changed here the ordering in the entries of $\mathbf{R}_\infty$ compared to \cite{MarchalAlameddineIsospectralIsomono2023}. We have also modified the definition of $R_{0,1}$ by dividing by $Q_{0,1}$. Since our normalization leaves $t_{\infty,2d}$ free before symmetry reduction, we define
\bea \label{RPQ}
R_{0,1}&:=&P_{0,1}+t_{\infty,2d+1}Q_{\infty,2d-2},
\cr
R_{\infty,0}&:=&- \, P_{\infty,0}+t_{\infty,2d+1}(Q_{\infty,2d-2})^2
-t_{\infty,2d+1}Q_{\infty,2d-3}, 
\cr
R_{\infty, 2d-2-k}&:=&- P_{\infty,2d-2-k}-\sum_{m=0}^{2d-3-k}P_{\infty,m}Q_{\infty,k+1+m}-t_{\infty,2d+1}Q_{\infty,k-1}+t_{\infty,2d+1}Q_{\infty,2d-2}Q_{\infty,k}
\cr&&\,\,,\,\, \forall \, k\in \llbracket 1, 2d-3\rrbracket,\cr
R_{\infty,2d-2}&:=&-P_{\infty,2d-2}-\sum_{m=0}^{2d-3}P_{\infty,m}Q_{\infty,m+1}+t_{\infty,2d+1}Q_{\infty,2d-2}Q_{\infty,0}
-t_{\infty,2d+1}Q_{0,1}, 
\eea
with $Q_{\infty,-1}:=Q_{0,1}$ when $d=1$. This yields
\beq \label{CheckL11RS}\check{L}_{1,1}(\lambda)=
-t_{\infty,2d+1}\lambda^{2d}
+  \sum_{k=0}^{2d-2}R_{\infty,2d-2-k}\lambda^k+ \frac{R_{0,1}Q_{0,1}}{\lambda}-t_{\infty,2d}\frac{\check{L}_{1,2}(\lambda)}{\omega}.
\eeq
The following definition gives the matrix form of \eqref{RPQ}.

\begin{definition}[Change of coordinates]
With
\begin{equation}
  \mathbf{R}_\infty := (R_{\infty,0},\ldots,R_{\infty,2d-2})^t,\,\,
  \mathbf{P}_\infty := (P_{\infty,0},\ldots,P_{\infty,2d-2})^t,\,\,
  \mathbf{Q}_\infty := (Q_{\infty,0},\ldots,Q_{\infty,2d-2})^t,
\end{equation}
the relations \eqref{RPQ} take the matrix form
\begin{equation}
  \begin{pmatrix}\mathbf{R}_\infty\\ R_{0,1}\end{pmatrix}
  =
  T(\mathbf{Q}_\infty)
  \begin{pmatrix}\mathbf{P}_\infty\\ P_{0,1}\end{pmatrix}
  +\mathbf{K}(\mathbf{Q}_\infty,Q_{0,1}).
\end{equation}
Here $J\in\mathcal{M}_{2d}(\mathbb{C})$ is the shift matrix
\[
  J=
  \begin{pmatrix}
    0 & 1      &        & 0\\
      & \ddots & \ddots &  \\
      &        & \ddots & 1\\
    0 &        &        & 0
  \end{pmatrix},
  \qquad
  J_{ij}=\delta_{j,i+1},
  \qquad \forall\, (i,j)\in\llbracket 1,2d\rrbracket^2,
\]
and $N(\mathbf{Q}_\infty)\in\mathcal{M}_{2d-1}(\mathbb{C})$ is the
strictly lower triangular Toeplitz matrix
\[
  N(\mathbf{Q}_\infty)=
  \begin{pmatrix}
0&0&\dots&0\\
Q_{\infty,2d-2}&0& &0\\
\vdots& \ddots &\ddots &\vdots \\
Q_{\infty,1}&\dots&Q_{\infty,2d-2}&0\\
\end{pmatrix},
  \qquad
  N(\mathbf{Q}_\infty)_{ij}=
  \begin{cases}
    Q_{\infty,2d-1-i+j}, & j<i,\\
    0,                    & j\geq i
  \end{cases}.
\]
In particular, $N(\mathbf{Q}_\infty)=(0)$ for $d=1$. The matrix $T$
and the vector $\mathbf{K}$ are given by
\begin{equation}\label{DefTK}
\begin{aligned}
  T(\mathbf{Q}_\infty)
  &:=
  \begin{pmatrix}
    -I_{2d-1}-N(\mathbf{Q}_\infty) & 0\\
    0 & 1
  \end{pmatrix},\\[4pt]
  \mathbf{K}(\mathbf{Q}_\infty,Q_{0,1})
  &:=
  t_{\infty,2d+1}Q_{\infty,2d-2}
  \begin{pmatrix}Q_{\infty,2d-2}\\ \vdots\\ Q_{\infty,0}\\ 1\end{pmatrix}
  -
  t_{\infty,2d+1}J
  \begin{pmatrix}Q_{\infty,2d-2}\\ \vdots\\ Q_{\infty,0}\\ Q_{0,1}\end{pmatrix}.
\end{aligned}
\end{equation}
\end{definition}

In practice, choosing $(\mathbf{Q}_\infty, Q_{0,1},\mathbf{R}_{\infty}, R_{0,1})$ as coordinates has the advantage that both entries $\check{L}_{1,1}(\lambda)$ and $\check{L}_{1,2}(\lambda)$ are expressed easily. In particular, we will see below that the variables $(\mathbf{R}_\infty, R_{0,1})$ are well-adapted to the symmetry (half of them vanish). However, the price to pay is that these coordinates are no longer canonical since the change of coordinates from $(\mathbf{Q}_\infty, Q_{0,1},\mathbf{P}_{\infty}, P_{0,1})$  to $(\mathbf{Q}_\infty, Q_{0,1},\mathbf{R}_{\infty}, R_{0,1})$ is not symplectic (it does not preserve the symplectic two-form $\Omega$). In particular, the standard Darboux equations do not apply to these pairs. The fact that the change of coordinates involves both coordinates $(\mathbf{P}_\infty,P_{0,1}, \mathbf{Q}_\infty,Q_{0,1})$ adds more difficulty to the problem. Indeed, to have a symplectic change of coordinates, one needs to construct conjugate coordinates for $\mathbf R$. Their existence requires the integrability of the system of PDEs defining them. We construct the symplectic completion for the triangular Toeplitz matrix in \eqref{DefTK}. For this, let us first observe the following standard result.  

\begin{lemma}[Inverse of a Toeplitz matrix]\label{PropInverseToep}We have:
\beq \begin{pmatrix}
-1&0&\dots&0\\
-Q_{\infty,2d-2}&-1& &0\\
\vdots& \ddots &\ddots &\vdots \\
-Q_{\infty,1}&\dots&-Q_{\infty,2d-2}&-1
\end{pmatrix}^{-1}= \begin{pmatrix}
-1&0&\dots&0\\
b_{2d-2}&-1& &0\\
\vdots& \ddots &\ddots &\vdots \\
b_{1}&\dots&b_{2d-2}&-1
\end{pmatrix},
\eeq
with coefficients $(b_1,\dots,b_{2d-2})$ determined by the generating series
\beq \label{GenSeries} -1+ \sum_{n=1}^{2d-2} b_{2d-1-n} u^n \overset{u\to 0}{=}-\frac{1}{1+\underset{i=1}{\overset{2d-2}{\sum}} Q_{\infty,2d-1-i} u^i} +O\left(u^{2d-1}\right).\eeq
Note that the previous identities are equivalent (take $u=\lambda^{-1}$ and multiply by $-\lambda^{1-2d}$) to
\beq \label{GenSeries2}\frac{1}{\lambda^{2d-1}+\underset{i=1}{\overset{2d-2}{\sum}} Q_{\infty,2d-1-i} \lambda^{2d-1-i}}\overset{\lambda\to \infty}{=} \lambda^{1-2d}- \sum_{n=1}^{2d-2} b_{2d-1-n} \lambda^{1-2d-n}  +O\left(\lambda^{2-4d}\right).\eeq
\end{lemma}

\begin{proof}
    Multiplication of lower triangular Toeplitz matrices corresponds to multiplication of their generating polynomials modulo $u^{2d-1}$. Equation~\eqref{GenSeries} therefore gives the inverse. 
\end{proof}

For clarity, the first coefficients are as follows whenever their indices lie in $\llbracket1,2d-2\rrbracket$
\bea b_{2d-2}&=&Q_{\infty, 2d-2},\cr
b_{2d-3}&=&-(Q_{\infty, 2d-2})^2 + Q_{\infty,2d-3},\\
b_{2d-4}&=&(Q_{\infty, 2d-2})^3-2 Q_{\infty, 2d-2}Q_{\infty, 2d-3}+ Q_{\infty, 2d-4}, \nonumber
\eea
and the general formula follows from the identity:
\beq \left(1- \sum_{n=1}^{2d-2} b_{2d-1-n} u^n\right)\left(1+\underset{i=1}{\overset{2d-2}{\sum}} Q_{\infty,2d-1-i} u^i\right)=1+O\left(u^{2d-1}\right),\eeq 
giving the induction:
\bea \forall\, k\in \llbracket 1, 2d-2\rrbracket\,:\,  b_{2d-1-k}&= &Q_{\infty,2d-1-k} -\sum_{n=1}^{k-1} b_{2d-1-n}Q_{\infty, 2d-1-k+n}, \cr
\text{i.e. }\forall\, j\in \llbracket 1, 2d-2\rrbracket\,:\,  b_{j}&=& Q_{\infty,j} -\sum_{m=j+1}^{2d-2}b_{m}Q_{\infty, 2d-1+j-m}.\eea 

In order to define a suitable change of Darboux coordinates, we need the following theorem.

\begin{theorem}[Existence of dual coordinates]\label{TheodS}
There exists a column vector
\[
  \mathbf{S}_\infty(\mathbf{Q}_\infty)
   =(S_{\infty,0},\ldots,S_{\infty,2d-2})^t
   \in\mathcal{M}_{2d-1,1}(\mathbb{C}).
\]
Its derivatives satisfy, for all $(i,j)\in\llbracket1,2d-1\rrbracket^2$,
\beq \partial_{Q_{\infty,j-1}}[\mathbf{S}_\infty(\mathbf{Q}_{\infty})]_i=\left[ \begin{pmatrix}
-1&0&\dots&0\\
-Q_{\infty,2d-2}&-1& &0\\
\vdots& \ddots &\ddots &\vdots \\
\vdots& \ddots &\ddots &0 \\
-Q_{\infty,1}&\dots&-Q_{\infty,2d-2}&-1\\
\end{pmatrix}^{-1} \right]_{j,i}=[(-I_{2d-1}-N(\mathbf{Q}_{\infty}))^{-1}]_{j,i}. \eeq
In other words, for all  $(i,j)\in \llbracket 1,2d-1\rrbracket^2$, $\partial_{Q_{\infty,j-1}} S_{\infty,i-1} =\left[( -I_{2d-1}-N(\mathbf{Q}_\infty))^{-1}\right]_{j,i}= b_{2d-1+i-j} \delta_{j\geq i}$ with $b_{2d-1}=-1$ and $b_\ell=0$ for $\ell\geq2d$.
\end{theorem}

\begin{proof}
The existence of $\mathbf{S}_\infty$ is equivalent to the integrability
condition
\begin{equation}\label{Toprove}
  \partial_{Q_{\infty,k-1}}[(-I_{2d-1}-N)^{-1}]_{j,i}
  =
  \partial_{Q_{\infty,j-1}}[(-I_{2d-1}-N)^{-1}]_{k,i},
  \qquad
  (i,j,k)\in\llbracket 1,2d-1\rrbracket^3.
\end{equation}
By \autoref{PropInverseToep}, this amounts to proving
\[
  \partial_{Q_{\infty,k-1}} b_{2d-1-(j-i)}
  =
  \partial_{Q_{\infty,j-1}} b_{2d-1-(k-i)},
\]
with the conventions $b_{2d-1}=-1$ and $b_\ell=0$ for $\ell\geq 2d$.
Differentiating the generating series \eqref{GenSeries} gives
\[
  \sum_{n=1}^{2d-2}
  \frac{\partial b_{2d-1-n}}{\partial Q_{\infty,k-1}}\,t^n
  =
  \frac{t^{2d-k}}
       {\left(1+\underset{\ell=1}{\overset{2d-2}{\sum}}
          Q_{\infty,2d-1-\ell}t^\ell\right)^2}
  +O(t^{2d-1}).
\]
Hence, for $1\leq n\leq 2d-2$,
\[
  \frac{\partial b_{2d-1-n}}{\partial Q_{\infty,k-1}}
  =
  [t^{n+k-2d}]
  \frac{1}{\left(1+\underset{\ell=1}{\overset{2d-2}{\sum}}
          Q_{\infty,2d-1-\ell}t^\ell\right)^2}.
\]
With the convention $[t^r]=0$ for $r<0$, this identity also holds
for $n\leq 0$, since both sides then vanish. Taking $n=j-i$ gives
\[
  \partial_{Q_{\infty,k-1}} b_{2d-1-(j-i)}
  =
  [t^{j-i+k-2d}]
  \frac{1}{\left(1+\underset{\ell=1}{\overset{2d-2}{\sum}}
          Q_{\infty,2d-1-\ell}t^\ell\right)^2}.
\]
The right-hand side is symmetric in $j$ and $k$, so condition
\eqref{Toprove} is satisfied.
\end{proof}

Choosing the additive constants so that $\mathbf S_\infty(0)=0$, we give the first examples of $\mathbf{S}_\infty(\mathbf{Q}_{\infty})$:

\begin{itemize}
    \item For $d=2$: $(-I-N)^{-1}=\begin{pmatrix} -1 &0&0\\
    Q_{\infty, 2}&-1&0\\
    -(Q_{\infty, 2})^2 + Q_{\infty,1}& Q_{\infty, 2}& -1\end{pmatrix}$ 
    so that we can take $S_{\infty,2}= -Q_{\infty, 2}$ and $S_{\infty,1}= \frac{1}{2}(Q_{\infty, 2})^2 -Q_{\infty,1}$ and $S_{\infty,0}=-\frac{1}{3}(Q_{\infty,2})^3 +Q_{\infty,1}Q_{\infty,2}-Q_{\infty,0}$ that satisfy 
    $\frac{\partial S_{\infty,i-1}}{\partial Q_{\infty,j-1}}=[(-I-N)^{-1}]_{j,i}$ for all $(i,j)\in \llbracket 1, 3\rrbracket^2$.
    \item For $d=3$: 
\tiny{\beq (-I-N)^{-1}=\begin{pmatrix}
    -1 &0&0&0&0\\
    Q_{\infty, 4}&-1&0&0&0\\
    -(Q_{\infty, 4})^2 + Q_{\infty,3}&Q_{\infty, 4}&-1&0&0\\
    (Q_{\infty, 4})^3-2 Q_{\infty, 4}Q_{\infty, 3}+ Q_{\infty, 2}& -(Q_{\infty, 4})^2 + Q_{\infty,3}&Q_{\infty, 4}&-1&0\\
    X&(Q_{\infty, 4})^3-2 Q_{\infty, 4}Q_{\infty, 3}+ Q_{\infty, 2}& -(Q_{\infty, 4})^2 + Q_{\infty,3}&+Q_{\infty, 4}&-1
    \end{pmatrix},
   \eeq}
   \normalsize{with}
   \beq X= -(Q_{\infty, 4})^4+3(Q_{\infty, 4})^2Q_{\infty, 3}- (Q_{\infty, 3})^2-2 Q_{\infty, 4}Q_{\infty, 2}+ Q_{\infty, 1}, \eeq
   so that
   \bea S_{\infty,4}&=&-Q_{\infty,4},\cr
   S_{\infty,3}&=&\frac{1}{2}(Q_{\infty,4})^2- Q_{\infty,3},\cr
   S_{\infty,2}&=&-\frac{1}{3}(Q_{\infty,4})^3+Q_{\infty,4}Q_{\infty,3}-Q_{\infty,2},\cr
   S_{\infty,1}&=&\frac{1}{4}(Q_{\infty,4})^4-(Q_{\infty,4})^2Q_{\infty,3}+Q_{\infty,4}Q_{\infty,2}+\frac{1}{2}(Q_{\infty,3})^2-Q_{\infty,1},\cr
   S_{\infty,0}&=&-\frac{1}{5}(Q_{\infty,4})^5+(Q_{\infty,4})^3Q_{\infty,3}-Q_{\infty,4}(Q_{\infty,3})^2-(Q_{\infty,4})^2Q_{\infty,2}+ Q_{\infty,4}Q_{\infty,1} \cr&&
   +Q_{\infty,3}Q_{\infty,2}-Q_{\infty,0}.
   \eea
\end{itemize}

\begin{definition}[Definition of the coordinates $(\mathbf{S}_\infty,S_{0,1})$]\label{DefS} We define:
\beq \mathbf{S}:=(S_{\infty,0},\dots,S_{\infty,2d-2}, S_{0,1})^t,
\eeq 
with the coordinates $\mathbf{S}_{\infty}:=(S_{\infty,0},\dots,S_{\infty,2d-2})^t$ defined by \autoref{TheodS} and
$S_{0,1}:=Q_{0,1}$.
\end{definition}

\begin{theorem}\label{TheoSymplectic} The change from the geometric Darboux coordinates
$(\mathbf Q_\infty,Q_{0,1},\mathbf P_\infty,P_{0,1})$ to the geometric Lax coordinates
$(\mathbf S_\infty,S_{0,1},\mathbf R_\infty,R_{0,1})$ is a time-independent symplectomorphism. In particular,
\bea\label{CanonicalChangeEq} \sum_{j=0}^{2d-2} dQ_{\infty,j}\wedge dP_{\infty,j} +dQ_{0,1}\wedge dP_{0,1}&=&\sum_{k=0}^{2d-2} dS_{\infty,k}\wedge dR_{\infty,k} +dS_{0,1}\wedge dR_{0,1}. 
\eea
In other words, the Darboux coordinates $(S_{\infty,0},\dots, S_{\infty,2d-2} ,S_{0,1},R_{\infty,0},\dots, R_{\infty,2d-2},R_{0,1})$ are also canonical for the Poisson bracket.
\end{theorem}

\begin{proof}The proof is given in \autoref{AppendixProofDarbouxCoord}.
\end{proof}

\subsubsection{Expression of the Lax matrix in the coordinates \texorpdfstring{$(\mathbf{R},\mathbf{S})$}{(R,S)}}
We recall that we have:
\beq \label{CheckL12RS}\check{L}_{1,2}(\lambda)= \omega\left( \frac{S_{0,1}}{\lambda}+ \sum_{k=0}^{2d-2} Q_{\infty,k}(\mathbf{S}_\infty) \lambda^k +\lambda^{2d-1}\right),\eeq
where the functions $Q_{\infty,k}(\mathbf{S}_\infty)$ are determined by \autoref{TheodS}.
Then, from \autoref{GeoLaxMatrices} and the definition of the geometric Lax coordinates we have:
{\small\setlength{\arraycolsep}{2pt}\bea \check{L}_{1,1}(\lambda)&=&-t_{\infty,2d+1}\lambda^{2d}
+  \sum_{k=0}^{2d-2}R_{\infty,2d-2-k}\lambda^k+ \frac{R_{0,1}S_{0,1}}{\lambda}-t_{\infty,2d}\frac{\check{L}_{1,2}(\lambda)}{\omega} \\
&=&-t_{\infty,2d+1}\lambda^{2d}-t_{\infty,2d}\lambda^{2d-1}+\sum_{k=0}^{2d-2}R_{\infty,2d-2-k}\lambda^k+ \frac{R_{0,1}S_{0,1}}{\lambda} -t_{\infty,2d}\left(\frac{S_{0,1}}{\lambda}+ \sum_{k=0}^{2d-2} Q_{\infty,k}(\mathbf{S}_\infty) \lambda^k\right)\cr
&=&-t_{\infty,2d+1}\lambda^{2d}-t_{\infty,2d}\lambda^{2d-1}+\sum_{k=0}^{2d-2}(R_{\infty,2d-2-k}-t_{\infty,2d}Q_{\infty,k}(\mathbf{S}_\infty))\lambda^k+ \left(R_{0,1} -t_{\infty,2d} \right)S_{0,1}\frac{1}{\lambda}. \nonumber
 \eea}
The polar term satisfies
\beq\label{SpecialValue} \frac{\frac{(t_{0,0})^2}{Q_{0,1}}-Q_{0,1}\left(P_{0,1}-t_{\infty,2d}+t_{\infty,2d+1}Q_{\infty,2d-2}\right)^2 }{\lambda}=\frac{\frac{(t_{0,0})^2}{S_{0,1}}-S_{0,1}(R_{0,1}-t_{\infty,2d})^2  }{\lambda}. \eeq
This gives the residue of $\omega\check L_{2,1}$. For the polynomial part, write
\beq \label{DefPtildeInfinity}\td{P}_\infty(\lambda):= \underset{j= 2 d-1}{\overset{4d}{\sum}}\left(\underset{m=0}{\overset{4d-j}{\sum}} t_{\infty,2 d +1-m}t_{\infty,j+m-2d+1}\right) \lambda^{j}:= \sum_{j=2d-1}^{4d} \td{P}_{\infty,j} \lambda^j.\eeq
Write
\[
  \check{L}_{1,1}(\lambda)
  =\sum_{k=0}^{2d}c_k\lambda^k
   +\frac{(R_{0,1}-t_{\infty,2d})S_{0,1}}{\lambda},
\]
where $c_{2d}=-t_{\infty,2d+1}$, $c_{2d-1}=-t_{\infty,2d}$ and
$c_k=R_{\infty,2d-2-k}-t_{\infty,2d}Q_{\infty,k}$
for $0\leq k\leq 2d-2$.
Observe first that the terms of degrees $4d$ and $4d-1$ cancel in
$\td{P}_\infty(\lambda)-\check{L}_{1,1}(\lambda)^2$. At the next order,
\[
 [\lambda^{4d-2}]\bigl(\td{P}_\infty(\lambda)
                         -\check{L}_{1,1}(\lambda)^2\bigr)
 =2t_{\infty,2d+1}\bigl(t_{\infty,2d-1}+c_{2d-2}\bigr).
\]
Thus the numerator is generically of degree $4d-2$ at infinity.
To obtain the constant term after division, we need one more term in
the reciprocal expansion \eqref{GenSeries2}. At relative order $2d-1$,
the denominator $\frac{\check{L}_{1,2}}{\omega}$ involves $Q_{\infty,0}$.
Its pole term $\frac{S_{0,1}}{\lambda}$ first enters at relative order $2d$.
We use the coefficients $b_1,\ldots,b_{2d-2}$ of the inverse Toeplitz matrix, defined in
\eqref{GenSeries}, with the convention $b_{2d-1}=-1$.
The additional coefficient is
\begin{equation}\label{Defb0}
 b_0=Q_{\infty,0}
     -\sum_{k=1}^{2d-2}Q_{\infty,k}\,b_{2d-1-k}.
\end{equation}
Multiplying the two series and using \autoref{PropInverseToep} gives
\[
 \frac{\omega}{\check{L}_{1,2}(\lambda)}
 =-\sum_{n=0}^{2d-1}b_{2d-1-n}\lambda^{1-2d-n}
   +O(\lambda^{1-4d}).
\]
The remainder therefore contributes only negative powers of $\lambda$.

After cancelling the two highest powers in the numerator, the contribution
from $\td{P}_{\infty,j}\lambda^j$ to the polynomial part comes from
$j\geq 2d-1+n$. Similarly, the contribution from
$-c_jc_k\lambda^{j+k}$ comes from $j+k\geq 2d-1+n$.
Finally, the pole term in $\check{L}_{1,1}$ contributes only through
its product with $c_{2d}\lambda^{2d}$ in $-\check{L}_{1,1}^2$.
After division this gives the constant
$-2c_{2d}(R_{0,1}-t_{\infty,2d})S_{0,1}$.
Adding the pole term of $\omega\check{L}_{2,1}$ from \eqref{SpecialValue},
we obtain
\begin{equation}\label{ExpressionL21}
\begin{aligned}
 \omega\check{L}_{2,1}(\lambda)
 ={}&-\sum_{n=0}^{2d-1}b_{2d-1-n}
      \sum_{j=2d-1+n}^{4d-2}
      \td{P}_{\infty,j}\lambda^{j+1-2d-n}\\
 &+\sum_{n=0}^{2d-1}b_{2d-1-n}
      \sum_{\substack{0\leq j,k\leq 2d\\
                       2d-1+n\leq j+k\leq 4d-2}}
      c_jc_k\lambda^{j+k+1-2d-n}\\
 &+2t_{\infty,2d+1}(R_{0,1}-t_{\infty,2d})S_{0,1}
   +\frac{1}{\lambda}\left(
      \frac{(t_{0,0})^2}{S_{0,1}}
      -S_{0,1}(R_{0,1}-t_{\infty,2d})^2\right).
\end{aligned}
\end{equation}
As in \eqref{CheckL12RS}, each $Q_{\infty,k}$ occurring in the coefficients
$b_j$ and $c_j$ stands for $Q_{\infty,k}(\mathbf{S}_\infty)$, obtained by
inverting the triangular change of coordinates in \autoref{TheodS}.
The coefficients $\td{P}_{\infty,j}$ in \eqref{DefPtildeInfinity}
depend only on the irregular times and the formal monodromy exponent $t_{\infty,0}$. Thus \eqref{ExpressionL21}
expresses $\omega\check{L}_{2,1}$ in the coordinates
$(\mathbf{R},\mathbf{S})$ at fixed times.

\subsubsection{Expression of the Hamiltonians for the Darboux coordinates \texorpdfstring{$(\mathbf{R},\mathbf{S})$}{(R,S)}}
Since the change of coordinates giving $(S_{\infty,0},\dots, S_{\infty,2d-2} ,S_{0,1},R_{\infty,0},\dots, R_{\infty,2d-2},R_{0,1})$ preserves the symplectic structure and because  $t_{\infty, 2d+1}$ is not a deformation parameter in this paper, we can immediately obtain the Hamiltonians for the Darboux coordinates $(\mathbf{R},\mathbf{S})$ by just replacing the old variables $(\mathbf{Q},\mathbf{P})$ by the new variables $(\mathbf{R},\mathbf{S})$ in the previous Hamiltonians.

\begin{theorem}[Hamiltonians for the $(\mathbf{R},\mathbf{S})$ Darboux coordinates for the even PIV hierarchy]\label{PropHamGeometricDarbouxRS} We have
\footnotesize{\beq \label{NewHamReduced3}\begin{pmatrix}\mathrm{Ham}_{{t_{\infty,1}}}(\mathbf{S}_{\infty}, S_{0,1},\mathbf{R}_{\infty},R_{0,1};\mathbf{t},t_{\infty,0},t_{0,0})\\ \vdots \\ (2d-1)\mathrm{Ham}_{{t_{\infty,2d-1}}}(\mathbf{S}_{\infty}, S_{0,1},\mathbf{R}_{\infty},R_{0,1};\mathbf{t},t_{\infty,0},t_{0,0})\\
(2d)\mathrm{Ham}_{{t_{\infty,2d}}}(\mathbf{S}_{\infty}, S_{0,1},\mathbf{R}_{\infty},R_{0,1};\mathbf{t},t_{\infty,0},t_{0,0})
\end{pmatrix}=\left(M_{\infty}(\mathbf{t})\right)^{-1}\begin{pmatrix}H_{\infty,2d-2}(\mathbf{S}_{\infty}, S_{0,1},\mathbf{R}_{\infty},R_{0,1},\mathbf{t},t_{\infty,0},t_{0,0})\\ \vdots\\ H_{\infty,0}(\mathbf{S}_{\infty}, S_{0,1},\mathbf{R}_{\infty},R_{0,1},\mathbf{t},t_{\infty,0},t_{0,0}) \\ H_{0,1}(\mathbf{S}_{\infty}, S_{0,1},\mathbf{R}_{\infty},R_{0,1},\mathbf{t},t_{\infty,0},t_{0,0})+\td{\delta} 
\end{pmatrix},
\eeq
}
\normalsize{with} the coefficient $\td{\delta}$ given in terms of the new coordinates from the initial term $\delta$. The expression of $(H_{\infty,j})_{0\leq j\leq 2d-2}$ and $H_{0,1}$ in terms of $(\mathbf{R}, \mathbf{S})$ can be obtained from \eqref{DefHs} and the expression of the entries of $\check{L}(\lambda)$ in terms of $(\mathbf{R},\mathbf{S})$ given in the previous section.
\end{theorem}

The important point to remember is that the general form of the Hamiltonians is the same in the three sets of Darboux coordinates given in this section. In each case the same matrix $M_\infty(\mathbf t)^{-1}$ acts on the vector with entries $\left(H_{\infty,2d-2},\ldots,H_{\infty,0},H_{0,1}-\underset{j=1}{\overset{2d}{\sum}}p_j\right)^t$, expressed in the chosen coordinates.

\section{The symmetry reduction to the FN PII hierarchy}\label{SecSymmetryReduction}

The purpose of this section is to show that the FN PII hierarchy can be obtained as a reduction of the even PIV hierarchy described in the previous section by imposing a symmetry at the level of the Lax system. We will start by discussing this symmetry and its origin and then impose it and discuss its constraints on the Lax pair.
\subsection{Involution and its consequences for the Lax matrices} 
The $n$-th element of the FN PII hierarchy admits a Lax pair representation that exhibits a symmetry when considering $\lambda\to -\lambda$ that we impose on the even PIV hierarchy.

\begin{definition} [Definition of the symmetry]\label{DefSymmetry}On coherently normalized fundamental matrices of the even Painlev\'e IV hierarchy on paired domains, we impose
\beq \check{\Psi}(-\lambda)= \sigma_1 \check{\Psi}(\lambda)\sigma_1, \qquad \text{ where } \qquad  \sigma_1:=\begin{pmatrix} 0&1\\1&0\end{pmatrix}. \eeq
Equivalently, the connection one-form is invariant under the pullback by $\lambda\mapsto-\lambda$ followed by conjugation by $\sigma_1$.
\end{definition}

The right factor $\sigma_1$ exchanges the two columns in the formal normalization at infinity. It cancels when passing to the Lax matrix in \autoref{PropositionSymmetryLaxMatrices}. Note that the symmetry is an involution and depends on the gauge in the sense that it is imposed on $\check{L}$ (it would take other forms in different gauges). Note that, since the symmetry governs the transformation $\lambda\to -\lambda$, it is not compatible with arbitrary M\"obius transformations acting on $\lambda$. For example, if the position of the finite pole had been taken arbitrary (and fixing $t_{\infty,2d}=0$ as in the canonical choice made in \cite{marchal2024hamiltonianrepresentationisomonodromicdeformations}) then the symmetry condition would have to be adapted. As we will see below, this symmetry is compatible with the normalization at infinity of $\check{L}(\lambda)$, i.e. the fact that $\check{L}(\lambda)=\operatorname{diag}(-t_{\infty,2d+1},t_{\infty,2d+1})\lambda^{2d} +O(\lambda^{2d-1})$. However, it will fix the value of $\omega$ to a non-trivial value thereby leaving no remaining degrees of freedom for the diagonal gauge action (the center $\{\pm I_2\}$ acts trivially). 

\medskip

The symmetry considered in this section admits a natural geometric interpretation. At the level of the spectral parameter, the involution $\lambda \mapsto -\lambda$ preserves the pole structure of the even PIV hierarchy and acts on the corresponding space of meromorphic connections. Requiring the horizontal sections to satisfy $\Psi(-\lambda)=\sigma_1\Psi(\lambda)\sigma_1$ selects the fixed-point locus of this action. As a consequence, the associated spectral data acquire an additional symmetry and naturally descend to the quotient by $\lambda \mapsto -\lambda$. From this perspective, the FN PII hierarchy appears as the isomonodromic dynamics induced on the invariant submanifold determined by the involution. The purpose of this section is to make this correspondence explicit at the level of the Lax matrices and of the Darboux coordinates.

\begin{proposition}[Symmetry for the Lax matrices]\label{PropositionSymmetryLaxMatrices}The symmetry condition of \autoref{DefSymmetry} for a family on the fixed base (so $\alpha_{\infty,2k}=0$) implies
\beq \check{L}(-\lambda)=-\sigma_1 \check{L}(\lambda) \sigma_1\, , \qquad  \text{ and } \qquad  
\check{A}_{\boldsymbol{\alpha}}(-\lambda)=\sigma_1 \check{A}_{\boldsymbol{\alpha}}(\lambda) \sigma_1,
\eeq
which translates at the level of the entries into
\begin{align}
 \check{L}_{1,1}(-\lambda)=&\check{L}_{1,1}(\lambda)\,, \qquad  \qquad \check{L}_{1,2}(-\lambda)=-\check{L}_{2,1}(\lambda)\,, \qquad  \check{L}_{2,1}(-\lambda)=-\check{L}_{1,2}(\lambda), \\
[\check{A}_{\boldsymbol{\alpha}}(-\lambda)]_{1,1}=&-[\check{A}_{\boldsymbol{\alpha}}(\lambda)]_{1,1}\,, \qquad \qquad [\check{A}_{\boldsymbol{\alpha}}(-\lambda)]_{1,2}=[\check{A}_{\boldsymbol{\alpha}}(\lambda)]_{2,1}\,, \qquad  [\check{A}_{\boldsymbol{\alpha}}(-\lambda)]_{2,1}=[\check{A}_{\boldsymbol{\alpha}}(\lambda)]_{1,2}. \nonumber
\end{align}
Consequently we have
\begin{align}
    \det \check{L}(-\lambda) = \det \check{L}(\lambda), \qquad \text{implying} \qquad t_{\infty,2k}=0\,\,,\,\, \forall \, k\in \llbracket 0,d\rrbracket.
\end{align}
On the chart $Q_{0,1}\ne0$, choosing the residue normalization $\underset{\lambda\to 0}{\Res}\check L_{1,2}=-t_{0,0}$ gives
\bea R_{0,1}&=&0,\cr
R_{\infty,2k-1}&=&0\,\,,\,\, \forall\, k\in \llbracket 1, d-1\rrbracket,\cr
\omega&=&-\frac{t_{0,0}}{Q_{0,1}}=-\frac{t_{0,0}}{S_{0,1}}.
\eea
\end{proposition}

\begin{proof} The impact of the symmetry on the Lax pair is deduced from the Lax system directly. The first part of the proof follows from easy computations and the fact that the highest part of the asymptotic expansion of $\det \check{L}(\lambda)$ at $\lambda=\infty$ is given by the irregular times:
\beq \det \check{L}(\lambda)\overset{\lambda\to \infty}{=} -\underset{j= 2 d-1}{\overset{4d}{\sum}}\left(\underset{m=0}{\overset{4d-j}{\sum}} t_{\infty,2 d +1-m}t_{\infty,j+m-2d+1}\right) \lambda^{j} +O(\lambda^{2d-2}).\eeq
The parity of $\check{L}_{1,1}(\lambda)$ and the fact that $t_{\infty,2d}=0$ imply from \eqref{CheckL11RS} that $R_{0,1}=0$ and $R_{\infty,2k-1}=0$. Finally, the value of $\omega$ follows from \eqref{SpecialValue} giving that $\check{L}_{2,1}(\lambda)\overset{\lambda\to 0}{=}\frac{(t_{0,0})^2}{\omega S_{0,1}\lambda}+ O(1)$ combined with $\check{L}_{1,2}(-\lambda)\overset{\lambda\to 0}{=}-\frac{\omega S_{0,1}}{\lambda} +O(1)$ so that $\omega^2=\frac{(t_{0,0})^2}{(S_{0,1})^2}$. The chosen sign gives $\underset{\lambda\to 0}{\Res}\check L(\lambda)=-t_{0,0}\sigma_1$.
\end{proof}

Note that the previous relations can also be translated to the variables $(Q_{\infty,k},P_{\infty,k},Q_{0,1},P_{0,1})$. For example $R_{0,1}=0$ is equivalent to $P_{0,1}=-t_{\infty,2d+1}Q_{\infty,2d-2}$ while $R_{\infty,2k-1}=0$ can be used to express $P_{\infty,2k-1}$ in terms of the other coordinates using \eqref{RPQ}.

\medskip

The previous proposition implies that the symmetry reduces the number of
independent deformation times from $2d$ to $d$, leaving the odd times
$t_{\infty,1},t_{\infty,3},\ldots,t_{\infty,2d-1}$. The same parity condition fixes the formal monodromy exponent $t_{\infty,0}=0$. The normalization
$\omega=-\frac{t_{0,0}}{S_{0,1}}$ fixes the diagonal gauge. Both the symplectic
form and the Hamiltonians of \autoref{sec3} are independent of this
choice.

\begin{theorem}[Hamiltonian reduction in the geometric Lax coordinates]
\label{TheoHamiltonianSymmetryReduction}
On the normalized chart $t_{0,0}Q_{0,1}\ne 0$, let $\iota$ denote the
inclusion of the symmetric locus at fixed deformation times. Wherever the
$d$ pairs $(S_{\infty,2j},R_{\infty,2j})_{0\leq j\leq d-1}$ are independent,
they are Darboux coordinates for the reduced symplectic form
\[
  \iota^*\Omega=\sum_{j=0}^{d-1}dS_{\infty,2j}\wedge dR_{\infty,2j}.
\]
For every $t=t_{\infty,2k+1}$, $0\leq k\leq d-1$, the Hamiltonian of
the reduced $t$-flow in these coordinates is $\iota^*\mathrm{Ham}_t$.
It is obtained by substituting the symmetry relations directly into the
Hamiltonian of \autoref{PropHamGeometricDarbouxRS}.
\end{theorem}

\begin{proof}
By \autoref{PropositionSymmetryLaxMatrices}, $R_{0,1}$ and
$R_{\infty,2j-1}$, $1\leq j\leq d-1$, vanish on the symmetric locus.
Their differentials therefore drop out of \eqref{CanonicalChangeEq}
after pullback, giving the stated Darboux form.
The odd-time flows preserve the symmetric locus, as shown in the proof of
\autoref{PropIdentification}. For each eliminated pair $(S_a,R_a)$,
the ambient Hamilton equation then gives
\[
  \iota^*\!\left(\frac{\partial\mathrm{Ham}_t}{\partial S_a}\right)
  =-\left.\frac{dR_a}{dt}\right|_{\mathrm{sym}}=0.
\]
Hence, at fixed times, the chain-rule terms from the eliminated $S_a$
vanish when differentiating $\iota^*\mathrm{Ham}_t$ with respect to a
surviving coordinate. Thus $\iota^*\mathrm{Ham}_t$ generates the
restricted equations for all surviving pairs.
\end{proof}

In the normalized gauge, the same substitution gives the reduced spectral
matrix and all odd-time auxiliary matrices, with $\mathcal L_{\boldsymbol\alpha}[\omega]$ obtained by differentiating $\omega=-\frac{t_{0,0}}{S_{0,1}}$. The Darboux coordinates
$(\td{Q}_i,\td{P}_i)_{1\leq i\leq d}$ of Mazzocco and Mo
\cite{mazzocco2007hamiltonian} allow us to identify these matrices with
those of the FN PII hierarchy for every $d$ and to compare the reduced
Hamiltonians with their polynomial expressions. The latter comparison
must account for the time dependence of the coordinate change.

\subsection{Identification with the canonical coordinates \texorpdfstring{$(\td{\mathbf{Q}},\td{\mathbf{P}})$}{(Q tilde, P tilde)}}
Identifying the Lax matrix from \cite{mazzocco2007hamiltonian} with the Lax matrix developed in \autoref{sec3} gives the following result.

\begin{proposition}[Relation between both formalisms after symmetry]\label{PropIdentification}
On the chart $t_{0,0}Q_{0,1}\omega\ne0$, choose the residue normalization of \autoref{PropositionSymmetryLaxMatrices}. The odd-time PIV flows preserve the symmetric spectral matrices. Both their spectral matrix and their normalized auxiliary matrices agree with the FN system of \autoref{DefIsomondoromicProb} under the correspondence
\begin{align}
d=&n, \qquad  t_{0,0}=\alpha_n, \qquad  t_{\infty,1}=z, \qquad t_{\infty,2k+1}=-4^{k}t_k\,\,\,,\,\, \forall\, k\in \llbracket 1,d-1\rrbracket, \nonumber \\ \omega Q_{\infty,2k}+R_{\infty,2d-2-2k}=&4^d \td{P}_{d-k}\,\,, \qquad  \forall\, k\in \llbracket 0,d-1\rrbracket,  \nonumber  \\
S_{0,1}=Q_{0,1}=&\frac{\alpha_d}{\td{Q}_d} \qquad \Leftrightarrow \qquad -\omega=\td{Q}_d=\frac{t_{0,0}}{Q_{0,1}}, \nonumber \\
\omega Q_{\infty,2i-1}=&-\sum_{j=1}^{d-i} \td{P}_j \td{Q}_{i+j}-\td{Q}_{i} \,\,, \qquad \forall\, i\in \llbracket 1,d-1\rrbracket,    \\
R_{\infty,2d-2-2k}=& a_{2k+1}^{(n)} \,\,,\,\, \forall \, k\in \llbracket0 ,d-1\rrbracket \nonumber \\
=&-\Res_{\lambda\to \infty} \lambda^{-2k-2}\left( \frac{1}{4} (2\lambda)^{2n+1} \left( 1 + \td{\mathcal{P}} + \frac{\left( 1 + \td{\mathcal{T}} \right)^2}{1+\td{\mathcal{P}}} \right) - (2 \lambda)^{-2n-1} \td{\mathcal{Q}}^2(1+ \td{\mathcal{P}}) \right). \nonumber
\end{align}
\end{proposition}
\begin{proof}
The leading part of $-\det\check L$ at infinity gives the stated times, with the eigenvalues ordered as in \autoref{NormalizationLaxMatrix}. At zero, the upper-right residue gives $-\alpha_d=-t_{0,0}$. Write $\td{\mathcal{A}}=\lambda\check L_{1,1}$ and decompose $\lambda\check L_{1,2}=\td{\mathcal{B}}_{\mathrm{odd}}+\td{\mathcal{B}}_{\mathrm{even}}-\alpha_d$ into its odd part, its positive even powers and its constant term. By \eqref{CoeffABB},
\[
U:=\td{\mathcal{A}}+\td{\mathcal{B}}_{\mathrm{odd}}=4^d\lambda^{2d+1}(1+\td{\mathcal{P}}),\qquad
\td{\mathcal{B}}_{\mathrm{even}}=-\lambda^2[\lambda^{-2}\td{\mathcal{Q}}(1+\td{\mathcal{P}})]_+.
\]
Comparing coefficients yields the displayed formulas, including $-\omega=\td Q_d$ from the leading even coefficient. Conversely, for any symmetric $\check L$ on the chart, the first identity determines $\td{\mathcal{P}}$, and the second determines $\td{\mathcal{Q}}$ successively in descending degree. These data determine the whole matrix: writing $V=\td{\mathcal{A}}-\td{\mathcal{B}}_{\mathrm{odd}}$, one has
\[
\lambda^2(-\det\check L)=UV+(\td{\mathcal{B}}_{\mathrm{even}}-\alpha_d)^2.
\]
For two matrices with these data and the same formal type, the difference on the left is $O(\lambda^{2d})$ by the determinant expansion in \autoref{PropositionSymmetryLaxMatrices}. On the right it is $U\,\Delta V$, of degree at least $2d+2$ if the odd polynomial $\Delta V$ is nonzero. Thus $V$ also agrees with the matrix of \autoref{LaxMatrixMarta}.

It remains to identify the time matrices. Write $L$ for the common spectral matrix and $\sigma_3=\operatorname{diag}(1,-1)$. With the eigenvalues ordered at infinity, a formal fundamental solution is $F e^{\theta\sigma_3}$, where $F=I+O(\lambda^{-1})$ and $\theta=-\underset{j=1}{\overset{2d+1}{\sum}}\frac{1}{j}t_{\infty,j}\lambda^j$ on the fixed base ($t_{\infty,0}=0$).

Let $v$ be an odd time direction and $D_v$ its lift in either system. Its traceless auxiliary matrix $A_v$ is polynomial in $\lambda$ by \autoref{DefIsomondoromicProb} and \autoref{AppendixAuxiliaryMatrixGeometricDarboux}, since the finite pole is fixed. A time-dependent diagonal change from $\omega=1$ by $G=\operatorname{diag}\left(\sqrt{\omega},\frac{1}{\sqrt{\omega}}\right)$ contributes $(D_vG)G^{-1}=\frac{1}{2\omega}(D_v\omega)\sigma_3$. In the formal gauge, zero curvature gives
\[
B=F^{-1}A_vF-F^{-1}D_vF-(v\theta)\sigma_3,\qquad \partial_\lambda B=[\theta'\sigma_3,B].
\]
The off-diagonal entries of $B$ satisfy $b'=\pm2\theta'b$. A nonzero formal Laurent series cannot satisfy this equation: differentiation lowers its highest degree, whereas multiplication by $\theta'$ raises it by $2d$. The diagonal entries are constant in $\lambda$. Since $D_vF=O(\lambda^{-1})$ and $L=F'F^{-1}+\theta' F\sigma_3 F^{-1}$, taking the polynomial part at infinity gives
\begin{equation}\label{eq:normalized-odd-aux}
A_v=A_v^0+c_v\sigma_3,\qquad A_v^0=\left[\frac{v\theta}{\theta'}\,L\right]_{\infty,+}.
\end{equation}
Here $F'F^{-1}=O(\lambda^{-2})$, $\frac{v\theta}{\theta'}=O(\lambda^{-1})$, and tracelessness removes the scalar constant. The quotient $\frac{v\theta}{\theta'}$ is odd, so $A_v^0(-\lambda)=\sigma_1A_v^0(\lambda)\sigma_1$. In particular $A_v^0(0)$ commutes with $\sigma_1$. To determine $c_v$ without assuming tangency, extend the normalization $\omega=-\frac{t_{0,0}}{Q_{0,1}}$ to the ambient chart. There the upper-right residue is identically $-t_{0,0}$. At a symmetric point put $R=\underset{\lambda\to 0}{\Res} L(\lambda)=-t_{0,0}\sigma_1$. Taking the residue of $D_v L=\partial_\lambda A_v+[A_v,L]$ gives
\[
0=(D_vR)_{12}=[A_v(0),R]_{12}=-2t_{0,0}c_v.
\]
Hence $c_v=0$. On the FN side the constant residue is the same, so the identical argument applies. Thus $A_v=A_v^0$ on both sides, for $v=\partial_z=\partial_{t_{\infty,1}}$ and $v=\partial_{t_k}=-4^k\partial_{t_{\infty,2k+1}}$, $1\le k<d$. The time parity just proved and the spatial parity of $L$ imply $(D_vL)(-\lambda)=-\sigma_1(D_vL)(\lambda)\sigma_1$ by zero curvature. This proves tangency and equality of every restricted evolution.
\end{proof}

The identification also preserves the symplectic form at fixed times.
In what follows, all exterior differentials are taken at fixed $\lambda$.
On the symmetric locus, $\check L_{1,1}$ is an even polynomial. Since
$\check L_{1,2}(\lambda)=\omega\lambda^{-1}\underset{i}{\prod}(\lambda-q_i)$ and
$p_i=\check L_{1,1}(q_i)$, extracting the coefficient of $\lambda^{-1}$ at infinity gives
\[
  \iota^*\Omega=\sum_i\mathrm d q_i\wedge\mathrm d p_i
  =-[\lambda^{-1}]\,\frac{\mathrm d\check L_{1,2}}{\check L_{1,2}}
     \wedge\mathrm d\check L_{1,1}.
\]
By \eqref{CoeffABB}, the entries of $\check L$ satisfy
\[
  \check L_{1,1}+\frac{\check L_{1,2}-\check L_{2,1}}{2}
    =4^d\lambda^{2d}(1+\td{\mathcal P}),\quad
  \frac{\check L_{1,2}+\check L_{2,1}}{2}
    =-\lambda^{-1}\td{\mathcal Q}(1+\td{\mathcal P})+O(\lambda^{-1}).
\]
Since $\check L$ is traceless,
$(\check L_{1,1}+\check L_{1,2})^2+\det\check L
=\check L_{1,2}(2\check L_{1,1}+\check L_{1,2}-\check L_{2,1})$.
Differentiating this identity and using the first of the two relations above gives
\[
 \frac{\mathrm d\check L_{1,2}}{\check L_{1,2}}\wedge\mathrm d\check L_{1,1}
  +\frac{\mathrm d\td{\mathcal P}}{1+\td{\mathcal P}}
     \wedge\frac{\mathrm d(\check L_{1,2}+\check L_{2,1})}{2}
 \qquad =
  \frac{\mathrm d(\det\check L)\wedge\mathrm d(\check L_{1,1}+\check L_{1,2})}
       {\check L_{1,2}(2\check L_{1,1}+\check L_{1,2}-\check L_{2,1})}
  =O(\lambda^{-2}).
\]
Indeed, the numerator is $O(\lambda^{4d-3})$ by the fixed formal type, and the
denominator has degree $4d-1$ at infinity. Using the second relation above and
$\mathrm d\td{\mathcal P}=O(\lambda^{-2})$, we obtain
\[
  \iota^*\Omega
  =-[\lambda^0]\,\mathrm d\td{\mathcal P}\wedge\mathrm d\td{\mathcal Q}
  =\sum_{i=1}^d\mathrm d\td Q_i\wedge\mathrm d\td P_i.
\]
The calculation on the simple-root chart extends throughout the normalized chart of \autoref{PropIdentification} by the rational coordinate formulas. On this chart the last display gives $(\iota^*\Omega)^d\ne0$. Since \autoref{TheoHamiltonianSymmetryReduction} expresses the same form as $\sum_{j=0}^{d-1}\mathrm d S_{\infty,2j}\wedge\mathrm d R_{\infty,2j}$, the $d$ surviving pairs are locally independent there.

The reverse change of variables is triangular. First $\td P_{d-k}=4^{-d}(\omega Q_{\infty,2k}+R_{\infty,2d-2-2k})$ and $\td Q_d=\frac{t_{0,0}}{Q_{0,1}}$. The remaining $\td Q_i$ follow in descending order from the coefficients $\omega Q_{\infty,2i-1}$. We use this correspondence below to compare the Hamiltonians of the first two members.

\section{Examples} \label{sec5}
In this section, we present for $d=1$ and $d=2$ the Hamiltonians obtained from the symmetry reduction of the even Painlev\'{e} IV hierarchy and match them with the results of \cite{mazzocco2007hamiltonian}.
\subsection{Case \texorpdfstring{$d=1$}{d=1}}
For $d=1$, we retain the fixed nonzero coefficient $t_{\infty,3}$ in the formulas to display the scaling. The FN normalization of \autoref{NormalizationLaxMatrix} sets $t_{\infty,3}=-4$. The Lax matrices for the even PIV hierarchy are
\footnotesize{\bea \check{L}_{1,1}(\lambda)&=&-t_{\infty,3}\lambda^2 -t_{\infty,2}\lambda +t_{\infty,3}((Q_{\infty,0})^2- Q_{0,1})- t_{\infty,2}Q_{\infty,0}-P_{\infty,0}+\frac{Q_{0,1}(t_{\infty,3}Q_{\infty,0}+P_{0,1}-t_{\infty,2})}{\lambda},\cr
\check{L}_{1,2}(\lambda)&=&\omega\left(\lambda +Q_{\infty,0}+\frac{Q_{0,1}}{\lambda}\right),\cr
\check{L}_{2,1}(\lambda)&=& \frac{1}{\omega}\Big[ 2t_{\infty,3}\left(-P_{\infty,0}+t_{\infty,3}(Q_{\infty,0})^2-t_{\infty,3}Q_{0,1}-t_{\infty,2}Q_{\infty,0}+t_{\infty,1}\right)\lambda+2t_{\infty,3}Q_{0,1}P_{0,1}\cr&&+ 2(t_{\infty,3}Q_{\infty,0}-t_{\infty,2})P_{\infty,0}
+4((t_{\infty,3})^2Q_{\infty,0}-t_{\infty,2}t_{\infty,3})Q_{0,1}-2(t_{\infty,3})^2(Q_{\infty,0})^3+4t_{\infty,3}t_{\infty,2}(Q_{\infty,0})^2\cr&&-2((t_{\infty,2})^2+t_{\infty,3}t_{\infty,1})Q_{\infty,0}
+2t_{\infty,0}t_{\infty,3}+2t_{\infty,1}t_{\infty,2}+\frac{(t_{0,0})^2-(Q_{0,1})^2(t_{\infty,3}Q_{\infty,0}+P_{0,1}-t_{\infty,2})^2 }{ Q_{0,1} \lambda}
\Big],\cr
\check{L}_{2,2}(\lambda)&=&- \check{L}_{1,1}(\lambda),
\eea}
\normalsize{and}
{\small\setlength{\arraycolsep}{2pt}\bea [\check{A}_{\boldsymbol{\alpha}}(\lambda)]_{1,2}&=& \omega\left(\nu_{\infty,0}^{(\boldsymbol{\alpha})}\lambda +Q_{\infty,0}\nu_{\infty,0}^{(\boldsymbol{\alpha})}+\nu_{\infty,1}^{(\boldsymbol{\alpha})}\right),\cr
[\check{A}_{\boldsymbol{\alpha}}(\lambda)]_{1,1}&=& -t_{\infty,3}\nu_{\infty,0}^{(\boldsymbol{\alpha})}\lambda^2-(t_{\infty,2}\nu_{\infty,0}^{(\boldsymbol{\alpha})}+t_{\infty,3}\nu_{\infty,1}^{(\boldsymbol{\alpha})})\lambda+\frac{1}{2\omega}\mathcal{L}_{\boldsymbol{\alpha}}[\omega]+ (t_{\infty,3}Q_{\infty,0}-t_{\infty,2})\nu_{\infty,1}^{(\boldsymbol{\alpha})} \cr&&
-(P_{\infty,0}+t_{\infty,2}Q_{\infty,0}-t_{\infty,3}(Q_{\infty,0})^2)\nu_{\infty,0}^{(\boldsymbol{\alpha})},\cr
[\check{A}_{\boldsymbol{\alpha}}(\lambda)]_{2,1}&=&\frac{1}{\omega}\Big[2t_{\infty,3}\nu_{\infty,0}^{(\boldsymbol{\alpha})}(-P_{\infty,0}+t_{\infty,3}(Q_{\infty,0})^2-t_{\infty,2}Q_{\infty,0}-t_{\infty,3}Q_{0,1}+t_{\infty,1})\lambda +\alpha_{\infty,2} \cr&&-t_{\infty,3}\mathcal{L}_{\boldsymbol{\alpha}}[Q_{\infty,0}]
-2t_{\infty,3}( t_{\infty,3}(Q_{0,1}-(Q_{\infty,0})^2) +t_{\infty,2}Q_{\infty,0} -t_{\infty,1})\nu_{\infty,1}^{(\boldsymbol{\alpha})}\cr&&
+\nu_{\infty,0}^{(\boldsymbol{\alpha})}\big(4t_{\infty,3}(t_{\infty,3}Q_{\infty,0}-t_{\infty,2})Q_{0,1} -2(t_{\infty,3})^2(Q_{\infty,0})^3+4t_{\infty,2}t_{\infty,3}(Q_{\infty,0})^2\cr&&
+2(t_{\infty,3}Q_{\infty,0}-t_{\infty,2})P_{\infty,0}
-2((t_{\infty,2})^2+t_{\infty,3}t_{\infty,1})Q_{\infty,0} +2t_{\infty,3}t_{\infty,0}+2t_{\infty,2}t_{\infty,1}- t_{\infty,3}\big)
\Big],\cr
[\check{A}_{\boldsymbol{\alpha}}(\lambda)]_{2,2}&=&-[\check{A}_{\boldsymbol{\alpha}}(\lambda)]_{1,1},
\eea}
where
\beq \nu_{\infty,0}^{(\boldsymbol{\alpha})}=\frac{1}{2t_{\infty,3}}\alpha_{\infty,2}\,\,\,,\,\, \nu_{\infty,1}^{(\boldsymbol{\alpha})}=\frac{1}{t_{\infty,3}}\alpha_{\infty,1}-\frac{t_{\infty,2}}{2(t_{\infty,3})^2}\alpha_{\infty,2}.
\eeq

The Hamiltonians for the Darboux coordinates $(Q_{\infty,0},Q_{0,1},P_{\infty,0},P_{0,1})$ with respect to the times $(t_{\infty,1},t_{\infty,2})$ are
\bea &&\mathrm{Ham}_{t_{\infty,1}}(Q_{\infty,0},Q_{0,1},P_{\infty,0},P_{0,1};t_{\infty,1},t_{\infty,2})=\frac{(P_{\infty,0})^2-Q_{0,1}(P_{0,1})^2}{t_{\infty,3}}+\frac{(t_{0,0})^2}{t_{\infty,3} Q_{0,1}}\cr&&-t_{\infty,3} (Q_{0,1})^2-t_{\infty,3}(Q_{\infty,0})^4 +2t_{\infty,2}(Q_{\infty,0})^3
+3t_{\infty,3} Q_{0,1}(Q_{\infty,0})^2-4t_{\infty,2}Q_{0,1}Q_{\infty,0}
\cr&&+\frac{2t_{\infty,3}t_{\infty,1}+(t_{\infty,2})^2}{t_{\infty,3}}Q_{0,1}-Q_{\infty,0}
-\frac{2t_{\infty,3}t_{\infty,1}+(t_{\infty,2})^2}{t_{\infty,3}}(Q_{\infty,0})^2+\frac{2(t_{\infty,3}t_{\infty,0}+t_{\infty,1}t_{\infty,2})}{t_{\infty,3}}Q_{\infty,0},\cr&&
\eea
and
\footnotesize{\bea &&\mathrm{Ham}_{t_{\infty,2}}(Q_{\infty,0},Q_{0,1},P_{\infty,0},P_{0,1};t_{\infty,1},t_{\infty,2})=
\frac{\bigl(t_{\infty,2}-t_{\infty,3}Q_{\infty,0}\bigr)Q_{0,1}(P_{0,1})^{2}}{2(t_{\infty,3})^{2}}-\frac{Q_{0,1}P_{0,1}P_{\infty,0}}{t_{\infty,3}}
+\frac{P_{\infty,0}}{2t_{\infty,3}}\cr&&
-\frac{t_{\infty,2}(P_{\infty,0})^{2}}{2(t_{\infty,3})^{2}}+\frac{(t_{0,0})^{2}\bigl(t_{\infty,3}Q_{\infty,0}-t_{\infty,2}\bigr)}{2(t_{\infty,3})^{2}Q_{0,1}}
+\bigl(Q_{\infty,0}t_{\infty,3}-\tfrac12 t_{\infty,2}\bigr)(Q_{0,1})^{2}
+\frac12 t_{\infty,2}(Q_{\infty,0})^{4}\cr&&-\frac{\bigl((t_{\infty,3})^2Q_{0,1}+2(t_{\infty,2})^{2}\bigr)(Q_{\infty,0})^{3}}{2t_{\infty,3}}
+\frac{Q_{0,1}\bigl(2t_{\infty,0}(t_{\infty,3})^{2}-(t_{\infty,2})^{3}-(t_{\infty,3})^{2}\bigr)}{2(t_{\infty,3})^{2}}\cr&&
-\frac{t_{\infty,2}\bigl((t_{\infty,3})^{2}Q_{0,1}-2t_{\infty,1}t_{\infty,3}
-(t_{\infty,2})^{2}\bigr)(Q_{\infty,0})^{2}}{2(t_{\infty,3})^{2}}+\frac{3Q_{0,1}Q_{\infty,0}(t_{\infty,2})^{2}}{2t_{\infty,3}}
-t_{\infty,1}Q_{0,1}Q_{\infty,0}\cr&&+\frac{\bigl((-2t_{\infty,0}+1)t_{\infty,2}t_{\infty,3}-2t_{\infty,1}(t_{\infty,2})^{2}\bigr)Q_{\infty,0}}{2(t_{\infty,3})^{2}}.
\eea}
\normalsize{We} can also rewrite the Lax matrices and obtain the corresponding Hamiltonians in terms of the Darboux coordinates $(S_{\infty,0},S_{0,1}, R_{\infty,0},R_{0,1})$. In this case we have:
\beq Q_{0,1}=S_{0,1}\,\,,\,Q_{\infty,0}=-S_{\infty,0}\,,\, P_{0,1}=t_{\infty,3}S_{\infty,0}+R_{0,1}\,,\, P_{\infty,0}=t_{\infty,3}(S_{\infty,0})^{2}-t_{\infty,3}S_{0,1}-R_{\infty,0}.
\eeq
We get:
\bea \check{L}_{1,1}(\lambda)&=&-t_{\infty,3}\lambda^2 - t_{\infty,2}\lambda + S_{\infty,0}t_{\infty,2} + R_{\infty,0} + \frac{S_{0,1}(R_{0,1}-t_{\infty,2})}{\lambda},\cr
\check{L}_{1,2}(\lambda)&=&\omega\left(\lambda -S_{\infty,0}+\frac{S_{0,1}}{\lambda}\right),\cr
\check{L}_{2,1}(\lambda)&=&\frac{1}{\omega}\Big[2t_{\infty,3}\left(t_{\infty,2}S_{\infty,0}+R_{\infty,0}+t_{\infty,1}\right)\lambda+\frac{(t_{0,0})^{2}-(S_{0,1})^{2}\left(R_{0,1}-t_{\infty,2}\right)^{2}}
{S_{0,1}\lambda}\cr&&
+2t_{\infty,3}R_{0,1}S_{0,1}
+2\left(t_{\infty,3}S_{\infty,0}+t_{\infty,2}\right)R_{\infty,0}
-2t_{\infty,2}t_{\infty,3}S_{0,1}
+2t_{\infty,2}t_{\infty,3}(S_{\infty,0})^{2}\cr&&
+2\left(t_{\infty,1}t_{\infty,3}+(t_{\infty,2})^{2}\right)S_{\infty,0}
+2\left(t_{\infty,0}t_{\infty,3}+t_{\infty,1}t_{\infty,2}\right)
\Big],
\eea
and
{\small\setlength{\arraycolsep}{2pt}\bea[\check{A}_{\boldsymbol{\alpha}}(\lambda)]_{1,2}&=&\omega\left(\nu_{\infty,0}^{(\boldsymbol{\alpha})}\lambda -S_{\infty,0}\nu_{\infty,0}^{(\boldsymbol{\alpha})}+\nu_{\infty,1}^{(\boldsymbol{\alpha})}\right),\cr
[\check{A}_{\boldsymbol{\alpha}}(\lambda)]_{1,1}&=&-t_{\infty,3}\nu_{\infty,0}^{(\boldsymbol{\alpha})}\lambda^2-(t_{\infty,2}\nu_{\infty,0}^{(\boldsymbol{\alpha})}+t_{\infty,3}\nu_{\infty,1}^{(\boldsymbol{\alpha})})\lambda+\frac{1}{2\omega}\mathcal{L}_{\boldsymbol{\alpha}}[\omega]- (t_{\infty,3}S_{\infty,0}+t_{\infty,2})\nu_{\infty,1}^{(\boldsymbol{\alpha})} \cr&&
+(S_{0,1}t_{\infty,3}+S_{\infty,0}t_{\infty,2}+R_{\infty,0})\nu_{\infty,0}^{(\boldsymbol{\alpha})},\cr
[\check{A}_{\boldsymbol{\alpha}}(\lambda)]_{2,1}&=&\frac{1}{\omega}\Big[
\left(2S_{\infty,0}t_{\infty,2}+2R_{\infty,0}+2t_{\infty,1}\right)\nu^{(\boldsymbol{\alpha})}_{\infty,0}\,t_{\infty,3}\lambda +\alpha_{\infty,2}+2t_{\infty,3}\left(S_{\infty,0}t_{\infty,2}+R_{\infty,0}+t_{\infty,1}\right)\nu_{\infty,1}^{(\boldsymbol{\alpha})}\cr&&
+\nu_{\infty,0}^{(\boldsymbol{\alpha})}\big( 2t_{\infty,3}R_{0,1}S_{0,1}
+2\left(t_{\infty,3}S_{\infty,0}+t_{\infty,2}\right)R_{\infty,0}
-2t_{\infty,2}t_{\infty,3}S_{0,1}
+2t_{\infty,2}t_{\infty,3}(S_{\infty,0})^{2}\cr&&
+2\left(t_{\infty,1}t_{\infty,3}+(t_{\infty,2})^{2}\right)S_{\infty,0}
+2\left((t_{\infty,0}-1)t_{\infty,3}+t_{\infty,1}t_{\infty,2}\right)\big)
\Big].
\eea}
The Hamiltonians are
\bea \label{HAmEx1} 
&&\mathrm{Ham}_{t_{\infty,1}}(S_{\infty,0},S_{0,1},R_{\infty,0},R_{0,1};t_{\infty,1},t_{\infty,2})=\frac{(t_{0,0})^{2}}{S_{0,1}t_{\infty,3}}
-\frac{S_{0,1}(R_{0,1})^{2}}{t_{\infty,3}}
-2S_{0,1}R_{0,1}S_{\infty,0}\cr&&
+2\left(S_{0,1}-(S_{\infty,0})^{2}\right)R_{\infty,0}
+\frac{(R_{\infty,0})^{2}}{t_{\infty,3}}
+4t_{\infty,2}S_{\infty,0}S_{0,1}
+\frac{\left(2t_{\infty,1}t_{\infty,3}+(t_{\infty,2})^{2}\right)S_{0,1}}{t_{\infty,3}}
-2t_{\infty,2}(S_{\infty,0})^{3}\cr&&
-\frac{\left(2t_{\infty,1}t_{\infty,3}+(t_{\infty,2})^{2}\right)(S_{\infty,0})^{2}}{t_{\infty,3}}
-\frac{\left(2t_{\infty,0}t_{\infty,3}+2t_{\infty,1}t_{\infty,2}-t_{\infty,3}\right)S_{\infty,0}}{t_{\infty,3}},
\eea
and
\small{\bea &&\mathrm{Ham}_{t_{\infty,2}}(S_{\infty,0},S_{0,1},R_{\infty,0},R_{0,1};t_{\infty,1},t_{\infty,2})=-\frac{(t_{0,0})^{2}\left(t_{\infty,3}S_{\infty,0}+t_{\infty,2}\right)}{2(t_{\infty,3})^{2}S_{0,1}}\cr&&
+\frac{\left(t_{\infty,3}S_{\infty,0}+t_{\infty,2}\right)S_{0,1}(R_{0,1})^{2}}{2(t_{\infty,3})^{2}}
+\frac{S_{0,1}\left(t_{\infty,3}S_{0,1}+t_{\infty,2}S_{\infty,0}+R_{\infty,0}\right)R_{0,1}}{t_{\infty,3}}
-\frac{t_{\infty,2}(R_{\infty,0})^{2}}{2(t_{\infty,3})^{2}}\cr&&
+\frac{\left(2t_{\infty,2}((S_{\infty,0})^{2}-S_{0,1})
+2t_{\infty,3}S_{0,1}S_{\infty,0}-1\right)R_{\infty,0}}{2t_{\infty,3}}
-t_{\infty,2}(S_{0,1})^{2}
+\frac{(t_{\infty,2})^{2}(S_{\infty,0})^{3}}{t_{\infty,3}}\cr&&
+\frac{\left(2(t_{\infty,2}(S_{\infty,0})^{2}
+t_{\infty,1}S_{\infty,0}+t_{\infty,0}-1)(t_{\infty,3})^{2}-3(t_{\infty,2})^{2}t_{\infty,3}S_{\infty,0}-(t_{\infty,2})^{3}\right)S_{0,1}}{2(t_{\infty,3})^{2}}\cr&&
+\frac{\left(t_{\infty,2}(2t_{\infty,0}-1)t_{\infty,3}+2t_{\infty,1}(t_{\infty,2})^{2}\right)S_{\infty,0}}{2(t_{\infty,3})^{2}}
+\frac{\left(2t_{\infty,1}t_{\infty,2}t_{\infty,3}+(t_{\infty,2})^{3}+(t_{\infty,3})^{2}\right)(S_{\infty,0})^{2}}{2(t_{\infty,3})^{2}}.\cr&&
\eea}

\normalsize{We} can then impose the symmetry which implies $R_{0,1}=0$, $t_{\infty,0}=t_{\infty,2}=0$ and $\omega=-\frac{t_{0,0}}{Q_{0,1}}$. We choose to keep $(Q_{0,1},P_{0,1})$ as local coordinates on the reduced surface and thus we have:
\beq Q_{\infty,0}=-\frac{P_{0,1}}{t_{\infty,3}}\,\,,\,\, P_{\infty,0}=-t_{\infty,3}Q_{0,1}+\frac{(P_{0,1})^{2}}{t_{\infty,3}}+t_{\infty,1}
-\frac{(t_{0,0})^{2}}{2t_{\infty,3}(Q_{0,1})^{2}}.\eeq
The corresponding reduced Lax matrices are:
\footnotesize{\beq\check{L}(\lambda)=\begin{pmatrix}-t_{\infty,3}\lambda^{2}+\frac{(t_{0,0})^{2}}{2t_{\infty,3}(Q_{0,1})^{2}}-t_{\infty,1}& -\frac{t_{0,0}\lambda}{Q_{0,1}}
+\frac{P_{0,1}t_{0,0}}{t_{\infty,3}Q_{0,1}}
-\frac{t_{0,0}}{\lambda}\\
-\frac{t_{0,0}\lambda}{Q_{0,1}}
-\frac{P_{0,1}t_{0,0}}{t_{\infty,3}Q_{0,1}}
-\frac{t_{0,0}}{\lambda}& t_{\infty,3}\lambda^{2}-\frac{(t_{0,0})^{2}}{2t_{\infty,3}(Q_{0,1})^{2}}+t_{\infty,1}
\end{pmatrix}\,\,,\,\,
\check{A}_{\boldsymbol{\alpha}}(\lambda)=\alpha_{\infty,1}\begin{pmatrix}-\lambda& -\frac{t_{0,0}}{t_{\infty,3}Q_{0,1}} \\ -\frac{t_{0,0}}{t_{\infty,3}Q_{0,1}}& +\lambda\end{pmatrix}.
\eeq}
\normalsize{The} evolutions of $Q_{0,1}$ and $P_{0,1}$ with respect to $t_{\infty,1}$ (which is the only deformation time remaining) are given by
\bea \label{EvolutionReduceddEqual1}\partial_{t_{\infty,1}} Q_{0,1}&=& -\frac{2Q_{0,1}P_{0,1}}{t_{\infty,3}},\cr
 \partial_{t_{\infty,1}} P_{0,1}&=&\frac{(t_{0,0})^{2}}{t_{\infty,3}(Q_{0,1})^{2}}-\frac{2(P_{0,1})^{2}}{t_{\infty,3}}-2t_{\infty,1}
+2t_{\infty,3}Q_{0,1},\cr
\partial_{t_{\infty,1}} Q_{\infty,0}&=&-2Q_{0,1}+\frac{2(P_{0,1})^{2}}{(t_{\infty,3})^{2}}+\frac{2t_{\infty,1}}{t_{\infty,3}}-\frac{(t_{0,0})^{2}}{(t_{\infty,3})^{2}(Q_{0,1})^{2}}.
\eea
At fixed time the induced form is $\frac{(t_{0,0})^2}{t_{\infty,3}^2(Q_{0,1})^3} \,dQ_{0,1}\wedge dP_{0,1}$, so this particular pair is noncanonical.

Regarding coordinates $(\mathbf{R},\mathbf{S})$, we have $R_{0,1}=0$, $S_{0,1}=Q_{0,1}$, $R_{\infty,0}= \frac{(t_{0,0})^{2}}{2t_{\infty,3}(S_{0,1})^{2}}-t_{\infty,1}$,  $S_{\infty,0}=-Q_{\infty,0}$ and $P_{0,1}=t_{\infty,3} S_{\infty,0}$. The evolutions are
\bea\label{EvoEx1} \partial_{t_{\infty,1}} S_{0,1}&=&-2S_{0,1}S_{\infty,0},\cr
\partial_{t_{\infty,1}} S_{\infty,0}&=&2S_{0,1}-2(S_{\infty,0})^2 -\frac{2t_{\infty,1}}{t_{\infty,3}} +\frac{(t_{0,0})^2}{(t_{\infty,3})^2(S_{0,1})^2}\cr
&=&2 \frac{t_{0,0}}{\sqrt{2t_{\infty,3}(R_{\infty,0}+t_{\infty,1})}}-2(S_{\infty,0})^2  +\frac{2R_{\infty,0}}{t_{\infty,3}}, \cr
\partial_{t_{\infty,1}} R_{\infty,0}&=&-1+ \frac{2(t_{0,0})^2S_{\infty,0}}{t_{\infty,3} (S_{0,1})^2}=-1+4(R_{\infty,0}+t_{\infty,1})S_{\infty,0}.
\eea

In the surviving Darboux coordinates $(S_{\infty,0},R_{\infty,0})$,
these evolutions are Hamiltonian. Choosing the local square root so that
\[
  S_{0,1}
  =\frac{t_{0,0}}
   {\sqrt{2t_{\infty,3}(R_{\infty,0}+t_{\infty,1})}},
\]
the restriction of the Hamiltonian \eqref{HAmEx1} to the symmetric locus gives
\begin{equation}
  \mathrm{Ham}_{t_{\infty,1}}
  (S_{\infty,0},R_{\infty,0};t_{\infty,1})
  = S_{\infty,0}
       -2(R_{\infty,0}+t_{\infty,1})(S_{\infty,0})^2
       +\frac{(R_{\infty,0})^2}{t_{\infty,3}}+\frac{2t_{0,0}}{t_{\infty,3}}
       \sqrt{2t_{\infty,3}(R_{\infty,0}+t_{\infty,1})}.
\end{equation}
Indeed, its Hamilton equations are
\[
\begin{aligned}
  \partial_{t_{\infty,1}}S_{\infty,0}
  &=\frac{\partial\mathrm{Ham}_{t_{\infty,1}}}{\partial R_{\infty,0}}
    =2 \frac{t_{0,0}}{\sqrt{2t_{\infty,3}(R_{\infty,0}+t_{\infty,1})}}-2(S_{\infty,0})^2  +\frac{2R_{\infty,0}}{t_{\infty,3}},\\
  \partial_{t_{\infty,1}}R_{\infty,0}
  &=-\frac{\partial\mathrm{Ham}_{t_{\infty,1}}}{\partial S_{\infty,0}}
    =-1+4(R_{\infty,0}+t_{\infty,1})S_{\infty,0},
\end{aligned}
\]
in agreement with the evolutions given above in \eqref{EvoEx1}.
In order to recover the standard Flaschka--Newell Lax pair, we define:
\beq q_{\mathrm{FN}}:=-\frac{t_{0,0}}{t_{\infty,3}Q_{0,1}}\,\,,\, p_{\mathrm{FN}}:=-\frac{t_{0,0}P_{0,1}}{2t_{\infty,3}Q_{0,1}}.
\eeq
At fixed time, $dq_{\mathrm{FN}}\wedge dp_{\mathrm{FN}}=-\tfrac12\,\iota^*\Omega$. With the convention $\{q_{\mathrm{FN}},p_{\mathrm{FN}}\}=1$, the $t_{\infty,1}$-flow has Hamiltonian
\beq \mathrm{Ham}_{t_{\infty,1}}(q_{\mathrm{FN}},p_{\mathrm{FN}};t_{\infty,1})=\frac{2(p_{\mathrm{FN}})^2}{t_{\infty,3}} -\frac{1}{8}t_{\infty,3}(q_{\mathrm{FN}})^4+\frac{1}{2}t_{\infty,1}(q_{\mathrm{FN}})^2+t_{0,0}q_{\mathrm{FN}},\eeq
so that $\partial_{t_{\infty,1}}q_{\mathrm{FN}}=\frac{4}{t_{\infty,3}} p_{\mathrm{FN}}$ and $q_{\mathrm{FN}}$ satisfies:
\beq \frac{d^2 q_{\mathrm{FN}}}{dt_{\infty,1}^{2}}  =\frac{4}{t_{\infty,3}}\left(\frac{1}{2}t_{\infty,3} (q_{\mathrm{FN}})^3 -t_{\infty,1}q_{\mathrm{FN}} -t_{0,0}\right), 
\eeq
which can be put into the standard FN PII equation by taking
\beq x=\left(-\frac{4}{t_{\infty,3}}\right)^{\frac{1}{3}} t_{\infty,1}\,\,,\, q=\left(-\frac{4}{t_{\infty,3}}\right)^{-\frac{1}{3}}q_{\mathrm{FN}} \,\,,\,\, \alpha= t_{0,0} \,\,:\,\, \frac{d^2q }{dx^2}=2q^3+xq+\alpha.\eeq

With the FN normalization ($t_{\infty,3}=-4$ and $t_1=1$), the correspondence with the Darboux coordinates $(\td{Q}_1,\td{P}_1)$ from \cite{mazzocco2007hamiltonian} is given by
\begin{align}
    Q_{0,1}=&\frac{t_{0,0}}{\td{Q}_1}, \qquad
\omega=-\td{Q}_1, \qquad t_{\infty,1}=z, \qquad
t_{\infty,3}=-4, \nonumber \\
Q_{\infty,0}=&\frac{4\td{P}_1-R_{\infty,0}}{\omega}=-\frac{4\td{P}_1-\frac{(\td{Q}_1)^2}{2t_{\infty,3}}+t_{\infty,1}}{\td{Q}_1}, \nonumber \\
P_{0,1}=&-t_{\infty,3}Q_{\infty,0}= 4t_{\infty,3}\frac{\td{P}_1}{\td{Q}_1} -\frac{1}{2}\td{Q}_1+\frac{t_{\infty,3}t_{\infty,1}}{\td{Q}_{1}},\nonumber \\
P_{\infty,0}=&-t_{\infty,3}Q_{0,1}+\frac{(P_{0,1})^{2}}{t_{\infty,3}}+t_{\infty,1}
-\frac{(t_{0,0})^{2}}{2t_{\infty,3}(Q_{0,1})^{2}}\nonumber \\
=&-\frac{t_{\infty,3}t_{0,0}}{\td{Q}_1} +\frac{\left(4t_{\infty,3}\td{P}_1-\frac12(\td{Q}_1)^2+t_{\infty,3}t_{\infty,1}\right)^2}{t_{\infty,3}(\td Q_1)^2}+t_{\infty,1}
-\frac{(\td{Q}_1)^2}{2t_{\infty,3}}, \nonumber \\ 
R_{\infty,0}=& \frac{(t_{0,0})^{2}}{2t_{\infty,3}(S_{0,1})^{2}}-t_{\infty,1}=\frac{(\td{Q}_1)^2}{2t_{\infty,3}}-t_{\infty,1}, \nonumber \\
S_{\infty,0}=&-Q_{\infty,0}=\frac{4\td{P}_1-\frac{(\td{Q}_1)^2}{2t_{\infty,3}}+t_{\infty,1}}{\td{Q}_1}, \qquad
S_{0,1}= \frac{t_{0,0}}{\td{Q}_1}. 
\end{align}
In particular, we have:
\beq \td{P}_1=\frac{t_{0,0}P_{0,1}}{4t_{\infty,3}Q_{0,1}}+\frac{(t_{0,0})^2}{8t_{\infty,3}(Q_{0,1})^2} -\frac{1}{4}t_{\infty,1}. 
\eeq 
We can compute the evolutions of $(\td{Q}_1,\td{P}_1)$ with respect to $t_{\infty,1}$ using \eqref{EvolutionReduceddEqual1} and we find that the evolutions are Hamiltonian with
\beq \mathrm{Ham}_{t_{\infty,1}}(\td{Q}_1,\td{P}_1;t_{\infty,1})=4(\td{P}_1)^2-\frac{(\td{Q}_1)^2\td{P}_1}{t_{\infty,3}} +2t_{\infty,1}\td{P}_1-\frac{1}{4}\left(2t_{0,0}-1\right)\td{Q}_1,
\eeq
in agreement with \eqref{Example1}.

\subsection{Case \texorpdfstring{$d=2$}{d=2}}
For $d=2$, the corresponding case of the PIV hierarchy involves lengthy formulas. For simplicity, we refer to \url{https://math.univ-lyon1.fr/~marchal/AdditionalRessources/index.html} for the Maple file providing the complete computations. Let us only recall the general procedure. The Lax matrices are given by \autoref{GeoLaxMatrices} and  \autoref{AppendixAuxiliaryMatrixGeometricDarboux}. The compatibility equations \eqref{compatibility} then give the evolutions of the Darboux coordinates $(Q_{0,1},Q_{\infty,0},Q_{\infty,1},Q_{\infty,2},P_{0,1},P_{\infty,0},P_{\infty,1},P_{\infty,2})$ in terms of the times $(t_{\infty,i})_{1\leq i\leq 4}$. These evolutions are given by the Hamiltonians provided in \autoref{PropHamGeometricDarboux}. Then, one imposes the symmetry condition that is equivalent in this case to
\bea 0&=&t_{\infty,0}=t_{\infty,2}=t_{\infty,4}, \cr
Q_{\infty,2}&=&-\frac{P_{0,1}}{t_{\infty,5}},\cr
Q_{\infty,1}&=&\frac{(P_{0,1})^{2}}{(t_{\infty,5})^{2}}-\frac{P_{\infty,0}}{t_{\infty,5}}+\frac{t_{\infty,3}}{t_{\infty,5}}
-\frac{(t_{0,0})^{2}}{2(t_{\infty,5})^{2}(Q_{0,1})^{2}},\cr
P_{\infty,1}&=&-P_{0,1}Q_{\infty,1}-t_{\infty,5}Q_{\infty,0}+\frac{P_{\infty,0}P_{0,1}}{t_{\infty,5}},\cr
P_{\infty,2}&=&-2P_{0,1}Q_{\infty,0}-t_{\infty,5}Q_{0,1}
-\frac{(P_{0,1})^{4}}{(t_{\infty,5})^{3}}
+\frac{(P_{0,1})^{2}P_{\infty,0}}{(t_{\infty,5})^{2}}
-\frac{t_{\infty,3}(P_{0,1})^{2}}{(t_{\infty,5})^{2}}
+\frac{(P_{\infty,0})^{2}}{t_{\infty,5}}
-\frac{t_{\infty,3}P_{\infty,0}}{t_{\infty,5}}\cr&&+t_{\infty,1}+\frac{3(t_{0,0})^{2}P_{\infty,0}}{2(t_{\infty,5})^{2}(Q_{0,1})^{2}}
-\frac{(t_{0,0})^{2}t_{\infty,3}}{2(t_{\infty,5})^{2}(Q_{0,1})^{2}}
+\frac{3(t_{0,0})^{4}}{8(t_{\infty,5})^{3}(Q_{0,1})^{4}}.
\eea
We choose to keep $(Q_{\infty,0},Q_{0,1},P_{\infty,0},P_{0,1})$ as local coordinates. After imposing the symmetry, the Lax matrices become
\bea
\check{L}_{1,1}(\lambda)&=&-t_{\infty,5}\lambda^{4}-t_{\infty,3}\lambda^{2}
+\frac{(t_{0,0})^{2}}{2(Q_{0,1})^{2}t_{\infty,5}}\lambda^{2}
-t_{\infty,1}+\frac{(P_{0,1})^{2}(t_{0,0})^{2}}{2(t_{\infty,5})^{3}Q_{0,1}^{2}}-\frac{P_{\infty,0}(t_{0,0})^{2}}{(t_{\infty,5})^{2}(Q_{0,1})^{2}}
\cr&&+\frac{t_{\infty,3}(t_{0,0})^{2}}{2(t_{\infty,5})^{2}(Q_{0,1})^{2}}
-\frac{3(t_{0,0})^{4}}{8(t_{\infty,5})^{3}(Q_{0,1})^{4}},\cr
\check{L}_{1,2}(\lambda)&=&-\frac{t_{0,0}}{Q_{0,1}}\lambda^{3}+\frac{P_{0,1}t_{0,0}}{t_{\infty,5}Q_{0,1}}\lambda^{2}
+\left(\frac{t_{0,0}P_{\infty,0}}{t_{\infty,5}Q_{0,1}}-\frac{t_{0,0}(P_{0,1})^{2}}{(t_{\infty,5})^{2}Q_{0,1}}-\frac{t_{\infty,3}t_{0,0}}{t_{\infty,5}Q_{0,1}}+\frac{(t_{0,0})^{3}}{2(Q_{0,1})^{3}(t_{\infty,5})^{2}}\right)\lambda\cr&&
-\frac{t_{0,0}Q_{\infty,0}}{Q_{0,1}}-\frac{t_{0,0}}{\lambda},\cr
\check{L}_{2,1}(\lambda)&=&-\frac{t_{0,0}}{Q_{0,1}}\lambda^{3}-\frac{P_{0,1}t_{0,0}}{t_{\infty,5}Q_{0,1}}\lambda^{2}
+\left(\frac{t_{0,0}P_{\infty,0}}{t_{\infty,5}Q_{0,1}}-\frac{t_{0,0}(P_{0,1})^{2}}{(t_{\infty,5})^{2}Q_{0,1}}-\frac{t_{\infty,3}t_{0,0}}{t_{\infty,5}Q_{0,1}}+\frac{(t_{0,0})^{3}}{2(Q_{0,1})^{3}(t_{\infty,5})^{2}}\right)\lambda\cr&&
+\frac{t_{0,0}Q_{\infty,0}}{Q_{0,1}}-\frac{t_{0,0}}{\lambda},\cr
\check{L}_{2,2}(\lambda) &=& - \check{L}_{1,1}(\lambda),
\eea
and, with $t_{\infty,5}=-16$, the auxiliary matrix is given by 
\bea
[\check{A}_{\boldsymbol{\alpha}}(\lambda)]_{1,1}&=& -\frac{\alpha_{\infty,3}}{3}\lambda^3-\alpha_{\infty,1}\lambda+\frac{\alpha_{\infty,3} (t_{0,0})^2}{1536\,(Q_{0,1})^2}\lambda, \cr
[\check{A}_{\boldsymbol{\alpha}}(\lambda)]_{1,2}&=&\frac{t_{0,0}\alpha_{\infty,3}}{48\,Q_{0,1}}\lambda^2+\frac{t_{0,0}\alpha_{\infty,3}P_{0,1}}{768\,Q_{0,1}}\lambda\cr&&+\frac{t_{0,0}\left(\left((P_{0,1})^{2}\alpha_{\infty,3}+16P_{\infty,0}\alpha_{\infty,3}+768\alpha_{\infty,1}\right)(Q_{0,1})^{2}-\frac{1}{2}(t_{0,0})^{2}\alpha_{\infty,3}\right)}{12288\,(Q_{0,1})^{3}},\cr
[\check{A}_{\boldsymbol{\alpha}}(\lambda)]_{2,1}&=&\frac{t_{0,0}\alpha_{\infty,3}}{48\,Q_{0,1}}\lambda^2-\frac{t_{0,0}\alpha_{\infty,3}P_{0,1}}{768\,Q_{0,1}}\lambda\cr&&+\frac{t_{0,0}\left(\left((P_{0,1})^{2}\alpha_{\infty,3}+16P_{\infty,0}\alpha_{\infty,3}+768\alpha_{\infty,1}\right)(Q_{0,1})^{2}-\frac{1}{2}(t_{0,0})^{2}\alpha_{\infty,3}\right)}{12288\,(Q_{0,1})^{3}},\cr
[\check{A}_{\boldsymbol{\alpha}}(\lambda)]_{2,2} &=& - [\check{A}_{\boldsymbol{\alpha}}(\lambda)]_{1,1}.
\eea
We use the Darboux coordinates $(\td{Q}_i,\td{P}_i)$ of \autoref{sec2}. In particular, \autoref{PropIdentification} gives $\td{Q}_2=-\omega=\frac{t_{0,0}}{Q_{0,1}}$. The change of coordinates is
\bea 
Q_{0,1}&=&\frac{t_{0,0}}{\td{Q}_2},\cr
P_{0,1}&=&-\frac{1}{2}\td{Q}_2-\frac{256\,\td{P}_1}{\td{Q}_2}-\frac{16\,t_{\infty,3}}{\td{Q}_2},\cr
Q_{\infty,0}&=&-\frac{1}{32}\td{Q}_2\td{P}_1+\frac{8(\td{P}_1)^2}{\td{Q}_2}
+\frac{t_{\infty,3}\td{P}_1}{\td{Q}_2}-\frac{1}{16}\td{Q}_1+\frac{(t_{\infty,3})^2}{32\td{Q}_2}-\frac{16\td{P}_2}{\td{Q}_2}-\frac{t_{\infty,1}}{\td{Q}_2},\cr
P_{\infty,0}&=&\frac{1}{64}(\td{Q}_2)^2-\frac{4096(\td{P}_1)^2}{(\td{Q}_2)^2}
-\frac{512t_{\infty,3}\td{P}_1}{(\td{Q}_2)^2}+\frac{16\td{Q}_1}{\td{Q}_2}
-\frac{16(t_{\infty,3})^2}{(\td{Q}_2)^2},
\eea
or equivalently
\bea
\td{P}_1
&=&-\frac{t_{0,0}P_{0,1}}{256\,Q_{0,1}}-\frac{t_{\infty,3}}{16}-\frac{(t_{0,0})^{2}}{512(Q_{0,1})^{2}},\cr
\td{P}_2&=&-\frac{(t_{0,0})^{2}(P_{0,1})^{2}}{131072(Q_{0,1})^{2}}
-\frac{t_{\infty,1}}{16}-\frac{t_{0,0}Q_{\infty,0}}{16\,Q_{0,1}}
-\frac{(t_{0,0})^{2}P_{\infty,0}}{4096(Q_{0,1})^{2}}
+\frac{t_{\infty,3}(t_{0,0})^{2}}{8192(Q_{0,1})^{2}}
+\frac{3(t_{0,0})^{4}}{524288(Q_{0,1})^{4}},\cr
\td{Q}_1&=&\frac{t_{0,0}(P_{0,1})^{2}}{256\,Q_{0,1}}
+\frac{(t_{0,0})^{2}P_{0,1}}{256(Q_{0,1})^{2}}
+\frac{t_{0,0}P_{\infty,0}}{16\,Q_{0,1}},\cr
\td{Q}_2&=&\frac{t_{0,0}}{Q_{0,1}}.
\eea
With these coordinates the Lax matrices are:
\bea \check{L}_{1,1}(\lambda)&=&16\lambda^{4}-\left(t_{\infty,3}+\frac{(\td{Q}_2)^{2}}{32}\right)\lambda^{2}-\frac{(\td{Q}_2)^{2}\td{P}_1}{32}-\frac{\td{Q}_1\td{Q}_2}{16}
+8(\td{P}_1)^{2}+t_{\infty,3}\td{P}_1+\frac{(t_{\infty,3})^{2}}{32}-t_{\infty,1},\cr
\check{L}_{1,2}(\lambda)&=&-\td{Q}_2\lambda^3+\left(16\td{P}_1+t_{\infty,3}+\frac{(\td{Q}_2)^2}{32}\right)\lambda^2
-\left(\td{Q}_2\td{P}_1+\td{Q}_1\right)\lambda-8(\td{P}_1)^2\cr&&+\left(\frac{(\td{Q}_2)^2}{32}-t_{\infty,3}\right)\td{P}_1+\frac{\td{Q}_1\td{Q}_2}{16}-\frac{(t_{\infty,3})^2}{32}+16\td{P}_2+t_{\infty,1}-\frac{t_{0,0}}{\lambda},\cr
\check{L}_{2,1}(\lambda)&=&-\td{Q}_2\lambda^3-\left(16\td{P}_1+t_{\infty,3}+\frac{(\td{Q}_2)^2}{32}\right)\lambda^2
-\left(\td{Q}_2\td{P}_1+\td{Q}_1\right)\lambda+8(\td{P}_1)^2\cr&&-\left(\frac{(\td{Q}_2)^2}{32}-t_{\infty,3}\right)\td{P}_1-\frac{\td{Q}_1\td{Q}_2}{16}+\frac{(t_{\infty,3})^2}{32}-16\td{P}_2-t_{\infty,1}-\frac{t_{0,0}}{\lambda},
\eea
and
\footnotesize{{\footnotesize\setlength{\arraycolsep}{2pt}\bea \check{A}_{t_{\infty,1}}&=&\begin{pmatrix}
    -\lambda&\frac{\td{Q}_2}{16}\\ \frac{\td{Q}_2}{16}&+\lambda
\end{pmatrix},   \\ \check{A}_{t_{\infty,3}}&=&\begin{pmatrix}
    -\frac{1}{3}\lambda^3 +\frac{(\td{Q}_2)^2}{1536}\lambda&\frac{\td{Q}_2}{48}\lambda^2-\frac{(\td{Q}_2)^2+512\td{P}_1+32t_{\infty,3}}{1536}\lambda+\frac{(16\td{P}_1+t_{\infty,3})\td{Q}_2}{768}+\frac{\td{Q}_1}{48}\\ \frac{\td{Q}_2}{48}\lambda^2+\frac{(\td{Q}_2)^2+512\td{P}_1+32t_{\infty,3}}{1536}\lambda+\frac{(16\td{P}_1+t_{\infty,3})\td{Q}_2}{768}+\frac{\td{Q}_1}{48} & \frac{1}{3}\lambda^3 -\frac{(\td{Q}_2)^2}{1536}\lambda \nonumber
\end{pmatrix}.
\eea}}
\normalsize{The} evolutions are Hamiltonian with:
{\small\setlength{\arraycolsep}{2pt}\bea 
&&\mathrm{Ham}_{t_{\infty,1}}(\td{Q}_1,\td{Q}_2,\td{P}_1,\td{P}_2;t_{\infty,1},t_{\infty,3})=
\frac{1}{16}(\td{Q}_2)^{2}\td{P}_2+32\td{P}_1\td{P}_2+2t_{\infty,3}\td{P}_2-2t_{\infty,3}(\td{P}_1)^{2}
\cr&&
-16(\td{P}_1)^{3}-\left(\frac{1}{16}(t_{\infty,3})^{2}-2t_{\infty,1}\right)\td{P}_1
-\frac{1}{16}(\td{Q}_1)^{2}+\frac{1}{16}(1-2t_{0,0})\td{Q}_2,\cr
&&\mathrm{Ham}_{t_{\infty,3}}(\td{Q}_1,\td{Q}_2,\td{P}_1,\td{P}_2;t_{\infty,1},t_{\infty,3})=\frac{1}{48}\td{P}_1\td{P}_2(\td{Q}_2)^2+\frac{1}{24}\td{Q}_1\td{Q}_2\td{P}_2+\frac{16}{3}(\td{P}_2)^2
\cr&&+\frac{1}{768}t_{\infty,3}(\td{Q}_2)^2\td{P}_2-\frac{1}{3}t_{\infty,3}(\td{P}_1)^3
-\frac{16}{3}(\td{P}_1)^2\td{P}_2-\frac{1}{24}(t_{\infty,3})^2(\td{P}_1)^2
+\left(\frac{1}{48}(t_{\infty,3})^2+\frac{2}{3}t_{\infty,1}\right)\td{P}_2\cr&&
+\frac{1}{48}(1-2t_{0,0})\td{Q}_2\td{P}_1-\frac{t_{\infty,3}}{24}
\left(\frac{1}{32}(t_{\infty,3})^2-t_{\infty,1}\right)\td{P}_1
-\frac{t_{\infty,3}}{768}(\td{Q}_1)^2+\frac{3-2t_{0,0}}{48}\td{Q}_1
+\frac{t_{\infty,3}(1-2t_{0,0})}{768}\td{Q}_2.\cr&&
\eea}
Taking $z=t_{\infty,1}$, $t_1=-\frac{t_{\infty,3}}{4}$, $t_{\infty,5}=-16$ and $\alpha_2=t_{0,0}$ in \eqref{Example2} gives
\[
\mathrm{Ham}_{t_{\infty,1}}=\td{\mathcal{H}}^{(2)}_2,\qquad
\mathrm{Ham}_{t_{\infty,3}}=-\frac{1}{4}\td{\mathcal{H}}^{(2)}_1.
\]
The factor $-\frac{1}{4}$ in the second identity comes from $\frac{dt_1}{dt_{\infty,3}}=-\frac{1}{4}$.

\section{Consequences and further discussions} \label{sec6}

On the normalized generic chart of \autoref{PropIdentification}, we have shown that the FN PII hierarchy is realized as the symmetric locus in the space $F_d$ of connections of the even PIV hierarchy: 
\beq \text{FN PII hierarchy  }= \left\{\check{L}(\lambda)\in F_{d} \,\text{ such that }\, \check{L}(-\lambda)=-\sigma_1 \check{L}(\lambda) \sigma_1\right\}.\eeq
This gives a new isomonodromic realization of the FN PII hierarchy. In particular, this identification suggests that the FN PII hierarchy should be viewed as a symmetry reduction inside a larger isomonodromic hierarchy rather than as an isolated system. This opens many perspectives that go beyond the scope of the present work.
\begin{itemize}
     \item \textit{\textbf{Suitable Darboux coordinate parameterizations:}} We have discussed several canonical Darboux coordinate systems giving different parameterizations of the Lax matrix and its auxiliary matrix adapted to different geometric aspects. The oper coordinates are adapted for the oper form of the connection and solving the compatibility equations, which are usually simpler in this gauge. The geometric coordinates are adapted for the expression of the Hamiltonians in terms of the coefficients of the connection. In the surviving geometric Lax coordinates $(S_{\infty,2j},R_{\infty,2j})_{0\leq j\leq d-1}$, every odd-time PIV Hamiltonian generates the corresponding reduced flow by restriction to the symmetric locus (\autoref{TheoHamiltonianSymmetryReduction}). A complementary construction of symmetry-adapted Darboux coordinates and explicit reduced Hamiltonians is given in \cite{MarchalAlameddineLaubFN}.
    \item \textit{\textbf{Symmetric Riemann--Hilbert problems:}} The Riemann--Hilbert (RH) problem for the meromorphic differential system of the FN PII equation was described in \cite{Flaschka1980}. Its symmetry relates the Stokes data and the asymptotics of the horizontal sections in paired sectors. Our result could be interpreted in terms of an RH problem whose symmetry restriction gives the problem described in \cite{Flaschka1980}. More generally, one may wonder if this is an instance of a more general picture on moduli spaces of connections and their analytic topological counterpart: the wild character varieties (the Betti side) \cite{Boalch2018}. In fact, through the irregular RH correspondence, one relates analytically meromorphic connections to character varieties. The natural question is whether one can describe fixed-point loci inside wild character varieties and identify them with moduli spaces of connections giving rise to reduced integrable hierarchies. We believe that the present work is an instance of this bigger picture.  
    \item \textit{\textbf{Folding of affine Dynkin diagrams:}} One of the remarkable features of the Painlev\'e-type equations is their intimate relation with affine Dynkin diagrams initiated by Okamoto \cite{zbMATH00929752} and then generalized further \cite{Boalch2012,doucot2023diagramsirregularconnectionsriemann}. The present construction suggests that the symmetry reduction may admit an interpretation at the level of these underlying affine root systems. It would be interesting to determine whether this reduction induces a folding of the associated affine diagrams, and to identify the corresponding map of root data.
\end{itemize}
Finally, we would like to stress that one motivation for this work is to clarify the link between the isomonodromic and isospectral descriptions of integrable systems. The FN PII hierarchy was discovered by a reduction of the mKdV hierarchy and admits a natural isospectral origin. On the other hand, the even PIV hierarchy arises naturally from the geometry of meromorphic connections. The current work establishes that these different constructions describe the same geometric object after imposing the appropriate symmetry. Consequently, the distinction between these two viewpoints becomes less rigid than previously thought and opens the possibility to merge them explicitly.

\appendix
\makeatletter
\renewcommand*{\subsection}{%
  \@startsection{appendixsubsection}{2}{\z@}%
    {-3.25ex\@plus -1ex\@minus -.2ex}{1.5ex\@plus .2ex}%
    {\normalfont\large\bfseries}}
\makeatother
\numberwithin{equation}{section}
\renewcommand{\theequation}{\thesection.\arabic{equation}}

\section{Expression of the auxiliary matrix}
\subsection{\texorpdfstring{Expression of $A_{\boldsymbol{\alpha}, \mathrm{oper}}(\lambda)$ in the oper Darboux coordinates}{Appendix \theappendixsubsection\space Auxiliary matrix in oper Darboux coordinates}}\label{AppendixAuxiliaryMatrixOperGauge}\label{AppendixA1}
From \cite{marchal2024hamiltonianrepresentationisomonodromicdeformations}, the entries of the matrix $A_{\boldsymbol{\alpha}, \mathrm{oper}}(\lambda)$ are given by
\bea \left[A_{\boldsymbol{\alpha}, \mathrm{oper}}(\lambda)\right]_{1,2}&=&\nu_{\infty,0}^{(\boldsymbol{\alpha})}+\sum_{j=1}^{2d} \frac{\mu_j^{(\boldsymbol{\alpha})}}{\lambda-q_j},\cr
\left[A_{\boldsymbol{\alpha}, \mathrm{oper}}(\lambda)\right]_{1,1}&=&\frac{1}{2\omega}\mathcal{L}_{\boldsymbol{\alpha}}[\omega]- \sum_{j=1}^{2d}\frac{p_j \mu_j^{(\boldsymbol{\alpha})}}{\lambda-q_j},\cr
\left[A_{\boldsymbol{\alpha}, \mathrm{oper}}(\lambda)\right]_{2,1}&=& \partial_{\lambda} \left[A_{\boldsymbol{\alpha}, \mathrm{oper}}(\lambda)\right]_{1,1}+\left[A_{\boldsymbol{\alpha}, \mathrm{oper}}(\lambda)\right]_{1,2}\left[L_{\mathrm{oper}}(\lambda)\right]_{2,1},\cr
\left[A_{\boldsymbol{\alpha}, \mathrm{oper}}(\lambda)\right]_{2,2}&=& \partial_{\lambda} \left[A_{\boldsymbol{\alpha}, \mathrm{oper}}(\lambda)\right]_{1,2}+\left[A_{\boldsymbol{\alpha}, \mathrm{oper}}(\lambda)\right]_{1,1}+\left[A_{\boldsymbol{\alpha}, \mathrm{oper}}(\lambda)\right]_{1,2}\left[L_{\mathrm{oper}}(\lambda)\right]_{2,2},
\eea
where the coefficients $(\mu_1^{(\boldsymbol{\alpha})},\dots, \mu_{2d}^{(\boldsymbol{\alpha})})$ are given by
\beq \begin{pmatrix}1&1 &\dots &\dots &1\\
q_1& q_2&\dots &\dots& q_{2d}\\
\vdots & & & & \vdots\\
\vdots & & & & \vdots\\
q_1^{2d-2}& q_2^{2d-2} &\dots & \dots& q_{2d}^{2d-2}\\
\frac{1}{q_1}&\frac{1}{q_2}& \dots&\dots& \frac{1}{q_{2d}}\end{pmatrix}\begin{pmatrix} \mu^{(\boldsymbol{\alpha})}_1\\ \vdots\\\vdots\\ \mu^{(\boldsymbol{\alpha})}_{2d}\end{pmatrix}=\begin{pmatrix} \nu^{(\boldsymbol{\alpha})}_{\infty,1}\\ \vdots\\ \nu^{(\boldsymbol{\alpha})}_{\infty,2d-1}\\ \nu^{(\boldsymbol{\alpha})}_{\infty,0}\end{pmatrix}.
\eeq

\subsection{\texorpdfstring{Expression of $\check{A}_{\boldsymbol{\alpha}}(\lambda)$ in the geometric Darboux coordinates}{Appendix \theappendixsubsection\space Auxiliary matrix in geometric Darboux coordinates}}\label{AppendixAuxiliaryMatrixGeometricDarboux}\label{AppendixA2}
Entries of the matrix $\check{A}_{\boldsymbol{\alpha}}(\lambda)$ are given by

\begin{align}
    \left[\check{A}_{\boldsymbol{\alpha}}(\lambda)\right]_{1,2}=&
\omega\left(\sum_{i=0}^{2d-1}\nu_{\infty,i}^{(\boldsymbol{\alpha})}\lambda^{2d-1-i} +\sum_{k=0}^{2d-2}\sum_{i=0}^{k}\nu_{\infty,i}^{(\boldsymbol{\alpha})}Q_{\infty,k} \lambda^{k-i}\right), \\
\left[\check{A}_{\boldsymbol{\alpha}}(\lambda)\right]_{1,1}=&\frac{1}{2\omega}\mathcal{L}_{\boldsymbol{\alpha}}[\omega]+t_{\infty,2d+1}\left(\nu_{\infty,0}^{(\boldsymbol{\alpha})}Q_{0,1}+ \sum_{j=1}^{2d-1}\nu_{\infty,j}^{(\boldsymbol{\alpha})}Q_{\infty,j-1}\right)+ \left[\check{L}_{1,1}(\lambda) \left(\sum_{i=0}^{2d-1} \nu_{\infty,i}^{(\boldsymbol{\alpha})}\lambda^{-i}\right)\right]_{\infty,+}. \nonumber
\end{align}

One also has $ \left[\check{A}_{\boldsymbol{\alpha}}(\lambda)\right]_{2,2}=-\left[\check{A}_{\boldsymbol{\alpha}}(\lambda)\right]_{1,1}$ and finally
\bea \left[\check{A}_{\boldsymbol{\alpha}}(\lambda)\right]_{2,1}&=& -\frac{t_{\infty,2d+1}\mathcal{L}_{\boldsymbol{\alpha}}[\omega]}{\omega^2}\lambda -\frac{(t_{\infty, 2d}-t_{\infty,2d+1}Q_{\infty,2d-2})\mathcal{L}_{\boldsymbol{\alpha}}[\omega]}{\omega^2}+ \frac{\alpha_{\infty,2d}-t_{\infty,2d+1}\mathcal{L}_{\boldsymbol{\alpha}}[Q_{\infty,2d-2}]}{\omega} \cr
&&+\left[
\frac{ \underset{i=0}{\overset{2d+1}{\sum}}\underset{j=2d-1}{\overset{4d}{\sum}}\underset{m=0}{\overset{4d-j}{\sum}}t_{\infty,2d+1-m}t_{\infty,j+m-2d+1}\nu_{\infty,i}^{(\boldsymbol{\alpha}) }\lambda^{j-i}  }{\check{L}_{1,2}(\lambda)}\right]_{\infty,+} -\frac{t_{\infty,2d+1}}{\omega}\nu_{\infty,0}^{(\boldsymbol{\alpha}) }  \cr
&&+ \left[\frac{\check{L}_{1,1}(\lambda)}{\check{L}_{1,2}(\lambda)}\left(\check{L}_{1,1}(\lambda)\frac{[\check{A}_{\boldsymbol{\alpha}}(\lambda)]_{1,2}}{\check{L}_{1,2}(\lambda)}-2[\check{A}_{\boldsymbol{\alpha}}(\lambda)]_{1,1}\right)\right]_{\infty,+},
\eea
with $\nu_{\infty,2d}^{(\boldsymbol{\alpha})}:= -\underset{j=1}{\overset{2d-1}{\sum}}\nu_{\infty,j}^{(\boldsymbol{\alpha})} Q_{\infty,j-1} -\nu_{\infty,0}^{(\boldsymbol{\alpha})} Q_{0,1}$ and 
$\nu_{\infty,2d+1}^{(\boldsymbol{\alpha})}:=
-\underset{j=2}{\overset{2d}{\sum}}\nu_{\infty,j}^{(\boldsymbol{\alpha})}Q_{\infty,j-2} -\nu_{\infty,1}^{(\boldsymbol{\alpha})}Q_{0,1}$.

\section{Proof of \autoref{TheoSymplectic}}\label{AppendixProofDarbouxCoord}
The leading time $t_{\infty,2d+1}$ in \eqref{DefTK} is fixed, and we take the additive constants in $\mathbf S_\infty$ independent of the deformation times.

With the coordinate ordering
\[
\begin{aligned}
 \mathbf R&=(R_{\infty,0},\ldots,R_{\infty,2d-2},R_{0,1})^t,&
 \mathbf S&=(S_{\infty,0},\ldots,S_{\infty,2d-2},S_{0,1})^t,\\
 \mathbf Q&=(Q_{\infty,0},\ldots,Q_{\infty,2d-2},Q_{0,1})^t,&
 \mathbf P&=(P_{\infty,0},\ldots,P_{\infty,2d-2},P_{0,1})^t,
\end{aligned}
\]
the change of coordinates is of the form $\mathbf{S}=F(\mathbf{Q})$ and $\mathbf{R}=T(\mathbf{Q}_\infty)\mathbf{P}+\mathbf{K}(\mathbf{Q}_\infty,Q_{0,1})$, i.e. $\mathbf{R}$ is affine in $\mathbf{P}$ while $\mathbf{S}$ is independent of $\mathbf{P}$. In particular we have $d\mathbf{S}= C(\mathbf{Q})d\mathbf{Q}$ where $C_{i,j}:=\frac{\partial S_i}{\partial Q_j}$. For $(i,j)\in \llbracket 1, 2d-1\rrbracket^2$, we have from \autoref{TheodS} that $C_{i,j}= \frac{\partial S_{\infty, i-1}}{\partial Q_{\infty, j-1}}=[T^{-1}]_{j,i}$. Thus, we have $C=(T^{-1})^t$ because the other entries are trivial from the definition of $S_{0,1}$. Moreover, we have:
 \beq dR_i=d\left(\sum_{m} T_{i,m}P_m+K_i\right)=\sum_{m,k} \frac{\partial T_{i,m}}{\partial Q_k}P_m dQ_k +  \sum_{m} T_{i,m} dP_m +  \sum_{k} \frac{\partial K_{i}}{\partial Q_k} dQ_k. \eeq
 Thus we have:
 \bea \sum_{i} dS_i\wedge dR_i&=& \sum_i\left(\sum_j C_{i,j}dQ_j \right) \wedge \left(\sum_{m,k} \frac{\partial T_{i,m}}{\partial Q_k}P_m dQ_k +  \sum_{m} T_{i,m} dP_m +  \sum_{k} \frac{\partial K_{i}}{\partial Q_k} dQ_k\right)\cr
 &=&\sum_{i,j,m} C_{i,j}T_{i,m} dQ_j \wedge dP_m + \sum_{i,j,m,k} C_{i,j} \frac{\partial T_{i,m}}{\partial Q_k} P_m dQ_j\wedge dQ_k +\sum_{i,j,k}C_{i,j} \frac{\partial K_i}{\partial Q_k} dQ_j \wedge dQ_k. \cr &&
 \eea
 Thus, \eqref{CanonicalChangeEq} is satisfied if and only if:
\bea (C_1)\,&:&\sum_{i=1}^{2d} C_{i,j}T_{i,m}=\delta_{j,m} \,\,,\,\, \forall\, (j,m)\in\llbracket 1, 2d\rrbracket^2,\cr
(C_2)\,&:&\,\,\forall\, m\in\llbracket 1, 2d\rrbracket\,:\, \sum_{i=1}^{2d} C_{i,k} \frac{\partial T_{i,m}}{\partial Q_j} = \sum_{i=1}^{2d} C_{i,j} \frac{\partial T_{i,m}}{\partial Q_k}  \,\,\,,\,\, \forall\, 1\leq j,k\leq 2d,\ j\neq k,\cr
 (C_3)\,&:&\,\,\sum_{i=1}^{2d}C_{i,k} \frac{\partial K_i}{\partial Q_j}=\sum_{i=1}^{2d}C_{i,j} \frac{\partial K_i}{\partial Q_k}\,\,\,,\,\, \forall\, 1\leq j,k\leq 2d,\ j\neq k.
\eea
\underline{Proof of $(C_1)$}: $(C_1)$ is equivalent to $C^t= T^{-1}$ which is already proved in \autoref{TheodS}.
 \smallskip
 
 \noindent\underline{Proof of $(C_2)$}: Let us now discuss condition $(C_2)$ depending on the value of $m\in \llbracket 1, 2d\rrbracket$.\\
 For $m=2d$ we have $T_{i,2d}=\delta_{i,2d}$ so that $(C_2)$ is trivially verified (both sides vanish). Since $S_{0,1}=Q_{0,1}$, we have $C_{2d,j}=\delta_{j,2d}$ and $C_{i,2d}=\delta_{i,2d}$.\\
 Let us now consider $m\in \llbracket 1, 2d-1\rrbracket$. For $j=2d$, the l.h.s. of $(C_2)$ vanishes because $T$ is independent of $Q_{0,1}$ so all partial derivatives vanish. For $j=2d$, the r.h.s. of $(C_2)$ is $\frac{\partial T_{2d,m}}{\partial Q_k}$ but since $m<2d$ all $T_{2d,m}$ vanish and the r.h.s. of $(C_2)$ vanishes. By symmetry, $(C_2)$ also holds for $k=2d$. Thus, let us now take $(m,j,k)\in \llbracket 1,2d-1\rrbracket^3$ so that we are only left with the Toeplitz blocks. The definition of $T$ is equivalent to
 \bea T_{i,i}&=&-1\, \,,\,\, \forall \, i\in \llbracket 1,2d-1\rrbracket,\cr
T_{i,m}&=&0 \,\,,\,\, \forall\, m>i,\cr
 T_{i,m}&=&-Q_{\infty, 2d-1-(i-m)} \,\, ,\,\,\, \forall\, 1\leq m<i\leq2d-1.
 \eea
 Thus, for all $(m,j)\in \llbracket 1,2d-1\rrbracket^2$ and $i\in \llbracket 1, 2d-1\rrbracket$, we have $\frac{\partial T_{i,m}}{\partial Q_j}=\frac{\partial T_{i,m}}{\partial Q_{\infty,j-1}}= -\delta_{j, 2d-i+m}\delta_{1\leq m<i\leq 2d-1}$. Hence we get:
 \bea \sum_{i=1}^{2d} C_{i,k} \frac{\partial T_{i,m}}{\partial Q_j}&=&-\sum_{i} C_{i,k} \delta_{j, 2d-i+m}\delta_{1\leq m<i\leq 2d-1}=-\sum_{i=m+1}^{2d-1}C_{i,k} \delta_{j, 2d-i+m}\cr
 &=&-\sum_{i=m+1}^{2d-1}C_{i,k}\delta_{i,2d-j+m}=-C_{2d-j+m,k} \delta_{m+1\leq 2d-j+m\leq 2d-1}\cr
 &=&-C_{2d-j+m,k} \delta_{j\leq 2d-1}\delta_{m+1\leq j}.
\eea
 Let us now recall that $C=(T^{-1})^t$ so that the $(2d-1)\times (2d-1)$ submatrix of $C$ is an upper-triangular Toeplitz matrix whose expression is
 \bea C_{i,i}&=&-1\, \,,\,\, \forall \, i\in \llbracket 1,2d-1\rrbracket,\cr
 C_{i,k}&=&0 \,\,,\,\, \forall\, i>k,\cr
 C_{i,k}&=&b_{2d-1-(k-i)} \,\, ,\,\,\, \forall\, 1\leq i<k\leq2d-1.
 \eea
 Thus, we have by defining $b_{2d-1}=-1$:
 \bea \sum_{i=1}^{2d} C_{i,k} \frac{\partial T_{i,m}}{\partial Q_j}&=&-b_{4d-1+m-k-j}\delta_{1\leq 2d-j+m\leq k\leq 2d-1} \delta_{j\leq 2d-1}\delta_{m+1\leq j}\cr
 &=&-b_{4d-1+m-k-j}\delta_{2d+m\leq k+j}.
 \eea
 Here $j+k\geq 2d+m$ implies $j,k\geq m+1$, since $j,k\leq 2d-1$. Terms with a vanishing indicator are understood as zero. The previous expression is symmetric in $j$ and $k$, which proves $(C_2)$.
 \smallskip
 
 \noindent\underline{Proof of $(C_3)$}: Let us first give the expression of $K_i$. From \eqref{DefTK}, with $Q_{\infty,-1}=Q_{0,1}$ when $d=1$, we have:
 \bea 
 K_{1}&=& t_{\infty,2d+1}(Q_{\infty,2d-2})^2-t_{\infty,2d+1}Q_{\infty,2d-3},
 \cr
 K_{i}&=& t_{\infty,2d+1}Q_{\infty,2d-2}Q_{\infty,2d-1-i}-t_{\infty,2d+1}Q_{\infty,2d-2-i}
 \,\,,\,\, \forall \, 2\leq i\leq 2d-2,\cr
 K_{2d-1}&=&t_{\infty,2d+1}Q_{\infty,2d-2}Q_{\infty,0}-t_{\infty,2d+1}Q_{0,1},\cr
 K_{2d}&=& t_{\infty,2d+1}Q_{\infty,2d-2}.
\eea
 Thus,
 \bea \frac{\partial K_i}{\partial Q_{0,1}}&=&-t_{\infty,2d+1}\delta_{i,2d-1}\,\,\,,\,\, \forall \, 1\leq i\leq 2d-1,\cr
 \frac{\partial K_{2d}}{\partial Q_{0,1}}&=& 0,
 \cr
 \frac{\partial K_1}{\partial Q_{\infty,2d-2}}&=& 2t_{\infty,2d+1}Q_{\infty,2d-2}, 
 \cr
 \frac{\partial K_1}{\partial Q_{\infty,j-1}}&=& -t_{\infty,2d+1}\delta_{j,2d-2} \,\,,\,\, \forall \, 1\leq j\leq 2d-2,\cr
 \frac{\partial K_i}{\partial Q_{\infty,2d-2}}&=& t_{\infty,2d+1}Q_{\infty,2d-1-i} \,\,,\,\, \forall \, 2\leq i\leq 2d-1,\cr
\frac{\partial K_i}{\partial Q_{\infty,j-1}}&=& t_{\infty,2d+1}Q_{\infty,2d-2}\delta_{2d-i,j}-t_{\infty,2d+1}\delta_{2d-1-i,j} 
 \,\,,\,\, \forall \, 2\leq i\leq 2d-1,\ 1\leq j\leq 2d-2,\cr
 \frac{\partial K_{2d}}{\partial Q_{\infty,2d-2}}&=& t_{\infty,2d+1},
\cr
 \frac{\partial K_{2d}}{\partial Q_{\infty,j-1}}&=& 0\,\,,\,\, \forall \, 1\leq j\leq 2d-2.
 \eea
 Thus, we get for $j=2d$ and  $1\leq k\leq 2d-1$ that the l.h.s. of $(C_3)$ is:
 \beq \sum_{i=1}^{2d}C_{i,k} \frac{\partial K_i}{\partial Q_{2d}}= \sum_{i=1}^{2d}C_{i,k} \frac{\partial K_i}{\partial Q_{0,1}}=-t_{\infty,2d+1}C_{2d-1,k}=t_{\infty,2d+1}\delta_{k,2d-1},
 \eeq
 while the r.h.s. of $(C_3)$ is
 \beq \sum_{i=1}^{2d}C_{i,2d} \frac{\partial K_i}{\partial Q_k}=C_{2d,2d}\frac{\partial K_{2d}}{\partial Q_k}=t_{\infty,2d+1}\delta_{k,2d-1},
 \eeq
 so that both sides are equal. By symmetry, the result also holds for $k=2d$ and $j\leq 2d-1$.\\ 
 Let us now consider $j=2d-1$ and $1\leq k\leq 2d-2$. The l.h.s. of $(C_3)$ in this case is
 \bea \label{lhsC3} \sum_{i=1}^{2d}C_{i,k} \frac{\partial K_i}{\partial Q_{2d-1}}&=& \sum_{i=1}^{2d}C_{i,k} \frac{\partial K_i}{\partial Q_{\infty, 2d-2}}\cr
 &=&C_{1,k}(2t_{\infty,2d+1}Q_{\infty,2d-2}
 )+\sum_{i=2}^{2d-1} C_{i,k} t_{\infty,2d+1}Q_{\infty,2d-1-i}+ C_{2d,k}t_{\infty,2d+1}
 \cr
 &=&C_{1,k}(2t_{\infty,2d+1}Q_{\infty,2d-2}
 )+\sum_{i=2}^{2d-1} C_{i,k} t_{\infty,2d+1}Q_{\infty,2d-1-i},
 \eea
 where we have used that $C$ is upper triangular so that $C_{2d,k}=0$ when $k\leq 2d-2$.
 For $j=2d-1$ and $1\leq k\leq 2d-2$, the r.h.s. of $(C_3)$ is
 \bea \label{rhsC3}&&\sum_{i=1}^{2d}C_{i,2d-1} \frac{\partial K_i}{\partial Q_k}=\sum_{i=1}^{2d-1}C_{i,2d-1} \frac{\partial K_i}{\partial Q_{\infty,k-1}}=C_{1,2d-1}\frac{\partial K_1}{\partial Q_{\infty,k-1}}+ \sum_{i=2}^{2d-1}C_{i,2d-1} \frac{\partial K_i}{\partial Q_{\infty,k-1}}\cr
&&=-C_{1,2d-1}t_{\infty,2d+1}\delta_{k,2d-2}+ \sum_{i=2}^{2d-1}C_{i,2d-1} (t_{\infty,2d+1}Q_{\infty,2d-2}\delta_{2d-i,k}-t_{\infty,2d+1}\delta_{2d-1-i,k}
 )\cr
 &&=-C_{1,2d-1}t_{\infty,2d+1}\delta_{k,2d-2}+ C_{2d-k,2d-1} t_{\infty,2d+1}Q_{\infty,2d-2}- C_{2d-1-k,2d-1}t_{\infty,2d+1}\delta_{k\leq 2d-3}
 \cr
 &&=-C_{1,2d-1}t_{\infty,2d+1}\delta_{k,2d-2}+ t_{\infty,2d+1}C_{1,k} Q_{\infty,2d-2}- t_{\infty,2d+1}C_{1,k+1}\delta_{k\leq 2d-3},
 \eea
 where we have used in the last equality the fact that $C$ is upper-triangular Toeplitz in its first $2d-1$ rows and columns, so that $C_{i,j}=C_{1,j-i+1}$ there for $i\leq j$. 
 Dividing \eqref{lhsC3} and \eqref{rhsC3} by $t_{\infty,2d+1}$, we see that $(C_3)$ is equivalent to
 \beq 2C_{1,k}Q_{\infty,2d-2}+\sum_{i=2}^{2d-1} C_{i,k} Q_{\infty,2d-1-i}=-C_{1,2d-1}\delta_{k,2d-2}+ C_{1,k} Q_{\infty,2d-2}-C_{1,k+1}\delta_{k\leq 2d-3},\eeq
 i.e. to
 \beq C_{1,k}Q_{\infty,2d-2}+\sum_{i=2}^{k} C_{1,k-i+1} Q_{\infty,2d-1-i}=-C_{1,2d-1}\delta_{k,2d-2}-C_{1,k+1}\delta_{k\leq 2d-3}.\eeq
 Recall now that $C_{1,k}=b_{2d-k}$ for all $k\in \llbracket 1, 2d-1\rrbracket$ with the convention that $b_{2d-1}=-1$. Thus the identity is equivalent to 
 \beq \sum_{i=1}^{k} b_{2d-1-k+i} Q_{\infty,2d-1-i}=-b_1\delta_{k,2d-2}-b_{2d-1-k}\delta_{k\leq 2d-3}, \eeq
i.e. for all $k\in \llbracket 1, 2d-2\rrbracket$:
\beq -Q_{\infty,2d-1-k}+\sum_{i=1}^{k-1} b_{2d-1-k+i} Q_{\infty,2d-1-i}=-b_{2d-1-k}.\eeq
 But coefficients $(b_k)_{1\leq k\leq 2d-2}$ are related to $(Q_{\infty,k})_{1\leq k\leq 2d-2}$ by
 \beq  \left(-1+ \sum_{n=1}^{2d-2} b_{2d-1-n} \lambda^n\right)\left(1+\underset{i=1}{\overset{2d-2}{\sum}} Q_{\infty,2d-1-i} \lambda^i\right)\overset{\lambda\to 0}{=} -1+ O(\lambda^{2d-1}), \eeq
 which is equivalent to the coefficient identity
 \beq \forall \,k \in \llbracket1,2d-2\rrbracket\,:\, b_{2d-1-k}-Q_{\infty, 2d-1-k} +\sum_{i=1}^{k-1}b_{2d-1-k+i} Q_{\infty,2d-1-i} =0,
 \eeq
 ending the proof for $j=2d-1$ and $k\leq 2d-2$. By symmetry, the result also holds for $k=2d-1$ and $j\leq 2d-2$.\\
 Let us finally consider the last cases corresponding to $1\leq j,k\leq 2d-2$. The l.h.s. of $(C_3)$ is given by
 \bea \label{lhsC3bis}&&\sum_{i=1}^{2d}C_{i,k} \frac{\partial K_i}{\partial Q_j}=\sum_{i=1}^{2d}C_{i,k} \frac{\partial K_i}{\partial Q_{\infty,j-1}}= C_{1,k}\frac{\partial K_1}{\partial Q_{\infty,j-1}}+\sum_{i=2}^{2d}C_{i,k} \frac{\partial K_i}{\partial Q_{\infty,j-1}}\cr
 &&=-t_{\infty,2d+1}\delta_{j,2d-2}C_{1,k}+ \sum_{i=2}^{2d-1}C_{i,k}(t_{\infty,2d+1}Q_{\infty,2d-2}\delta_{2d-i,j}-t_{\infty,2d+1}\delta_{2d-1-i,j}
 )\cr
 &&=-t_{\infty,2d+1}\delta_{j,2d-2}C_{1,k}+ t_{\infty,2d+1}Q_{\infty,2d-2}C_{2d-j,k}\delta_{k+j\geq 2d}\cr&&
 -t_{\infty,2d+1}C_{2d-1-j,k}\delta_{k+j\geq 2d-1}\delta_{j\neq 2d-2}
\cr
 &&=-t_{\infty,2d+1}\delta_{j,2d-2}C_{1,k}+ t_{\infty,2d+1}Q_{\infty,2d-2}C_{1,k+j-2d+1}\delta_{k+j\geq 2d}\cr&&-t_{\infty,2d+1}C_{1,k+j+2-2d}\delta_{k+j\geq 2d-1}\delta_{j\neq 2d-2},
 \eea
 where we have used the Toeplitz structure of $C$: $C_{i,j}=C_{1,j-i+1}$ for $1\leq i\leq j\leq 2d-1$.
 The r.h.s. of $(C_3)$ is given by
 \bea \label{rhsC3bis}&&\sum_{i=1}^{2d}C_{i,j} \frac{\partial K_i}{\partial Q_k}=\sum_{i=1}^{2d}C_{i,j} \frac{\partial K_i}{\partial Q_{\infty,k-1}}= C_{1,j}\frac{\partial K_1}{\partial Q_{\infty,k-1}} + \sum_{i=2}^{2d}C_{i,j}\frac{\partial K_i}{\partial Q_{\infty,k-1}}\cr
 &&= -t_{\infty,2d+1}C_{1,j}\delta_{k,2d-2}+\sum_{i=2}^{2d-1}C_{i,j}(t_{\infty,2d+1}Q_{\infty,2d-2}\delta_{2d-i,k}-t_{\infty,2d+1}\delta_{2d-1-i,k}
 ) \cr
 &&=-t_{\infty,2d+1}C_{1,j}\delta_{k,2d-2} +t_{\infty,2d+1}Q_{\infty,2d-2} C_{2d-k,j}\delta_{k+j\geq 2d} \cr&&
 -t_{\infty,2d+1}C_{2d-1-k,j}\delta_{k+j\geq 2d-1}\delta_{k\neq 2d-2}
 \cr
 &&=-t_{\infty,2d+1}C_{1,j}\delta_{k,2d-2} +t_{\infty,2d+1}Q_{\infty,2d-2} C_{1,j+k-2d+1}\delta_{k+j\geq 2d} \cr&&
 -t_{\infty,2d+1}C_{1,j+k-2d+2}\delta_{k+j\geq 2d-1}\delta_{k\neq 2d-2}.
\eea
Thus the two sides of $(C_3)$ given by \eqref{lhsC3bis} and \eqref{rhsC3bis} are equal if and only if
\beq \label{IdTopProve}\delta_{j,2d-2}C_{1,k}+ C_{1,k+j+2-2d}\delta_{k+j\geq 2d-1}\delta_{j\neq 2d-2}=C_{1,j}\delta_{k,2d-2}+ C_{1,j+k-2d+2}\delta_{k+j\geq 2d-1}\delta_{k\neq 2d-2}.\eeq
This last equality is obvious if both $j$ and $k$ are at most $2d-3$ or if both are equal to $2d-2$. If $j=2d-2$ and $k\leq 2d-3$, then the l.h.s. is $C_{1,k}+0$ while the r.h.s. is $0+C_{1,2d-2+k-2d+2}\delta_{k\geq 1}=C_{1,k}$. The case $k=2d-2$ follows by symmetry. Thus \eqref{IdTopProve} always holds. This ends the proof of condition $(C_3)$. The triangular change $\mathbf Q_\infty\mapsto\mathbf S_\infty$ with diagonal derivatives $-1$ is inverted successively, and $T$ is invertible. This proves \autoref{TheoSymplectic}.

\footnotesize{
\bibliographystyle{plain}
\bibliography{Biblio}
}
\end{document}